\documentclass[11pt]{article}
\usepackage{latexsym}
\usepackage{amstext,amssymb,amsmath}
\usepackage{amsbsy}
\usepackage{fullpage}
\usepackage{epsfig,eepic,gastex} 
\usepackage{luca} 
\usepackage{times}
\def\game{\calg}

\def\mov{\Gamma}

\def\moves{\ita{A}}

\def\movs{\mov}

\def\dest{\mathit{Dest}}
\def\stra{\pi}
\def\bigstra{\Pi}
\def\distr{\cald}
\def\pat{\omega}
\def\pats{\Omega}
\def\seq#1{\langle #1 \rangle}

\def\Prb{\mbox{\rm Pr}}

\def\supp{\mbox{\rm Supp}}

\def\N{\hspace{4pt}\raise 3pt \hbox{\circle{7}} \hspace{4pt} }

\def\go{\rightarrow}
\def\wh{\widehat}

\def\dashlinestretch{30}
\def\outcome{\mbox{\it Outcomes\/}}
\def\trans{\delta}
\def\pr{P}

\def\setm{\setminus}
\def\ve{\varepsilon}

\def\win#1#2#3{\langle\!\langle #1 \rangle\!\rangle_{#2} \big(#3\big)}

\def\bigo{\calo}

\def\sure{{\text{\textit{sure}}}}
\def\almost{{\text{\textit{Almost}}}}
\def\limit{{\text{\textit{Limit}}}}
\def\bounded{{\text{\textit{bounded}}}}

\def\wm{\zeta}

\def\sd{\chi}
\def\dis{\xi}

\def\dopp#1#2#3#4#5{\xi_{#2,#1}^{#3}[#4](#5)}

\def\pre{{\text{\textrm{Pre}}}}
\def\epre{{\text{\textrm{Pospre}}}}

\def\apre{{\text{\textrm{Apre}}}}
\def\lpre{{\text{\textrm{Lpre}}}}

\def\fpre{{\text{\textrm{Fpre}}}}

\newcommand{\mem}{{\mathcal M}}
\newcommand{\DM}{\mathit{DM}}
\def\Inf{{\text{\textit{Inf}}}}

\newcommand{\Reach}{{\text{\textrm{Reach}}}}
\newcommand{\Safe}{{\text{\textrm{Safe}}}}
\newcommand{\Buchi}{{\text{\textrm{B\"uchi}}}}
\newcommand{\coBuchi}{{\text{\textrm{co-B\"uchi}}}}
\newcommand{\Parity}{{\text{\textrm{Parity}}}}
\newcommand{\ParityCond}{{\text{\textrm{Parity}}}(p)}
\newcommand{\coParityCond}{{\text{\textrm{coParity}}}(p)}
\newcommand{\Nats}{\mathbb{N}}

\newcommand{\ov}{\overline}

\def\lpreodd{\text{\textrm{APreOdd}}}
\def\lpreeven{\text{\textrm{APreEven}}}
\def\fpreodd{\text{\textrm{PosPreOdd}}}
\def\fpreeven{\text{\textrm{PosPreEven}}}

\def\frpreodd{\text{\textrm{FrPreOdd}}}
\def\frpreeven{\text{\textrm{FrPreEven}}}

\newcommand{\dcup}{{\makebox[0em][c]{$\bigcup$}\makebox[0em]{$\ast$} }}
\newcommand{\dcap}{{\makebox[0em][c]{$\bigcap$}\makebox[0em]{$\ast$} }}
\newcommand{\sdcup}{\ \ \dcup \ \ }
\newcommand{\sdcap}{\ \ \dcap \ \ }

\newcommand{\FP}{\mathit{FP}}
\newcommand{\IP}{\mathit{IP}}
\newcommand{\FM}{\mathit{FM}}
\newcommand{\IM}{\mathit{IM}}
\newcommand{\wt}{\widetilde}
\newcommand{\unif}{\mathsf{unif}}

\newtheorem{remark}{Remark}

\newcommand{\slopefrac}[2]{\leavevmode\kern.1em
  \raise .5ex\hbox{\the\scriptfont0 #1}\kern-.1em
  /\kern-.15em\lower .25ex\hbox{\the\scriptfont0 #2}}
\newcommand{\half}{\slopefrac{1}{2}}

\begin{document}

\sloppy 

\title{\bf Bounded Rationality in Concurrent Parity Games\thanks{We are indebted to and thank 
anonymous reviewers for extremely helpful comments.}}

\author{Krishnendu Chatterjee \\
\normalsize
  IST Austria (Institute of Science and Technology Austria) \\
  \texttt{Krishnendu.Chatterjee@ist.ac.at}
}

\date{ 
}

\maketitle

\begin{abstract}
We study two-player {\em concurrent} games on finite-state graphs played for an 
infinite number of rounds, where in  each round, the two players (player~1 and 
player~2) choose their moves independently and simultaneously; the current 
state and the two moves determine the successor state. 
The objectives are $\omega$-regular winning conditions specified as 
{\em parity\/} objectives. 
We consider the \emph{qualitative analysis} problems: the computation of 
the \emph{almost-sure} and \emph{limit-sure} winning set of states, where 
player~1 can ensure to win with probability~1 and with probability arbitrarily 
close to~1, respectively. 
In general the almost-sure and limit-sure winning strategies 
require both \emph{infinite-memory} as well as \emph{infinite-precision}
(to describe probabilities).
While the qualitative analysis problem for concurrent parity games with 
infinite-memory, infinite-precision randomized strategies was studied 
in~\cite{dAH00,CdAH11}, we study the \emph{bounded-rationality} problem for 
qualitative analysis of concurrent parity games, where the strategy set for 
player~1 is restricted to bounded-resource strategies. 
In terms of precision, strategies can be deterministic, uniform, 
finite-precision, or infinite-precision; and in terms of memory, strategies 
can be memoryless, finite-memory, or infinite-memory.
We present a precise and complete characterization of the qualitative 
winning sets for all combinations of classes of strategies.
In particular, we show that uniform memoryless strategies are as 
powerful as finite-precision infinite-memory strategies, 
and infinite-precision memoryless strategies are as powerful 
as  infinite-precision finite-memory strategies. 
We show that the winning sets can be computed in $\bigo(n^{2d+3})$ time, 
where $n$ is the size of the game structure and $2d$ is the number of 
priorities (or colors), and our algorithms are symbolic.
The membership problem of whether a state belongs to a winning set
can be decided in NP $\cap$ coNP. 
Our symbolic algorithms are based on characterization of the winning 
sets as $\mu$-calculus formulas, however, our $\mu$-calculus formulas 
are crucially different from the ones for concurrent parity games 
(without bounded rationality); and our memoryless witness strategy 
constructions are significantly different from the infinite-memory witness 
strategy constructions for concurrent parity games.
\end{abstract}

\section{Introduction}\label{sec-intro}

In this work we consider the qualitative analysis (computation of 
almost-sure and limit-sure winning sets) for concurrent parity games.
In prior works~\cite{dAH00,CdAH11} the qualitative analysis for concurrent 
parity games have been studied for the general class of infinite-memory, 
infinite-precision randomized strategies.
In this work, we study the \emph{bounded rationality} problem where the 
resources of the strategies are limited, and establish precise and complete 
characterization of the qualitative analysis of concurrent parity games for 
combinations of resource-limited strategies.
We start with the basic background of concurrent games, parity objectives,
qualitative analysis, and the previous results.

\smallskip\noindent{\em Concurrent games.}
A two-player (player~1 and player~2) concurrent game is played on a finite-state 
graph for an infinite number of rounds, where in each round, the players independently 
choose moves, and the current state and the two chosen moves determine the successor state. 
In \emph{deterministic} concurrent games, the successor state is unique;
in \emph{probabilistic} concurrent games, the successor state is given by 
a probability distribution.
The outcome of the game (or a \emph{play}) is an infinite sequence of states.
These games were introduced by Shapley~\cite{Shapley53}, and have been one of 
the most fundamental and well studied game models in stochastic graph games.
We consider $\omega$-regular objectives; where given an $\omega$-regular set 
$\Phi$ of plays, player~1 wins if the outcome of the game lies in~$\Phi$.
Otherwise, player~2 wins, i.e., the game is zero-sum.
Such games occur in the synthesis and verification of reactive systems
\cite{Church62,RamadgeWonham87,PnueliRosner89}, and $\omega$-regular
objectives (that generalizes regular languages to infinite words) 
provide a robust specification language that can express all 
specifications (such as safety, liveness, fairness) that arise in the analysis 
of reactive systems (see also~\cite{ALW89,Dill89book,ATL-FOCS97}).
Concurrency in moves is necessary for modeling the synchronous interaction of 
components~\cite{AHM00a,AHM01a}.
Parity objectives can express all $\omega$-regular conditions, and we
consider concurrent games with parity objectives.

\smallskip\noindent{\em Qualitative and quantitative analysis.}
The player-1 \emph{value} $v_1(s)$ of the game at a state $s$ is the 
limit probability with which player~1 can guarantee $\Phi$ against all strategies of player~2.
The player-2 \emph{value} $v_2(s)$ is analogously the limit 
probability with which player~2 can ensure that the outcome of 
the game lies outside~$\Phi$.
The \emph{qualitative} analysis of games asks for the computation of the 
set of \emph{almost-sure} winning states where player~1 can ensure $\Phi$ 
with probability~1, and the set of \emph{limit-sure} winning states where 
player~1 can ensure $\Phi$ with probability arbitrarily close to~1 (states with value~1); 
and 
the \emph{quantitative} analysis asks for precise computation of values.
Concurrent (probabilistic) parity games are determined~\cite{Mar98}, i.e., 
for each state $s$ we have $v_1(s)+v_2(s)=1$.
The qualitative analysis for concurrent parity games was studied in~\cite{dAH00,CdAH11} 
and the quantitative analysis in~\cite{dAM04}.

\smallskip\noindent{\em Difference of turn-based and concurrent games.}
Traditionally, the special case of \emph{turn-based\/} games has received 
most attention.
In turn-based games, in each round, only one of the two players has a 
choice of moves.
In turn-based deterministic games, all values are 0 or~1 and can be 
computed using combinatorial algorithms \cite{Thom90,Schewe07,JPZ06};
in turn-based probabilistic games, values can be computed by iterative 
approximation \cite{CH06,Condon92,GH08}.
Concurrent games significantly differ from turn-based games in requirement
of strategies to play optimally.
A \emph{pure} strategy must, in each round, choose a move based on the current state 
and the history (i.e., past state sequence) of the game, 
whereas, a \emph{randomized} strategy in each round chooses a probability
distribution over moves (rather than a single move).
In contrast to turn-based deterministic and probabilistic games with parity 
objectives, where pure memoryless (history-independent) optimal strategies 
exist~\cite{EJ88,Zie98,DJW97,CJH03,KrishThesis}, 
in concurrent games, both randomization and infinite-memory are required 
for limit-sure winning~\cite{dAH00} (also see~\cite{EY06} for results on 
pushdown concurrent games, \cite{HKM09,crg-tcs07} on complexity of 
strategies required in concurrent reachability games, and \cite{EY10,HKLMT11}
on complexity of related concurrent game problems).

\smallskip\noindent{\em Bounded rationality.} 
The qualitative analysis for concurrent parity games with infinite-memory,
infinite-precision randomized strategies was studied in~\cite{dAH00,CdAH11}. 
The strategies for qualitative analysis for concurrent games require two 
different types of infinite resource: (a) infinite-memory, and (b) infinite-precision 
in describing the probabilities in the randomized strategies; 
(see example in~\cite{dAH00} that limit-sure winning in concurrent B\"uchi 
games require both infinite-memory and infinite-precision). 
In many applications, such as synthesis of reactive systems, infinite-memory 
and infinite-precision strategies are not implementable in practice. 
Thus though the theoretical solution of infinite-memory and infinite-precision 
strategies was established in~\cite{dAH00}, the strategies obtained are not 
realizable in practice, and the theory to obtain implementable strategies in 
such games has not been studied before.
In this work we consider the \emph{bounded rationality} problem for qualitative
analysis of concurrent parity games, where player~1 (that represents the controller) 
can play strategies with bounded resource.
To the best of our knowledge this is the first work that considers the bounded 
rationality problem for concurrent $\omega$-regular graph games.
The motivation is clear as controllers obtained from infinite-memory and
infinite-precision strategies are not implementable.

\smallskip\noindent{\em Strategy classification.}
In terms of precision, strategies can be classified as pure (deterministic), uniformly 
random, bounded-finite-precision, finite-precision, and infinite-precision 
(in increasing order of precision to describe probabilities of a randomized strategy).
In terms of memory, strategies can be classified as memoryless, finite-memory and 
infinite-memory. 
In~\cite{dAH00} the almost-sure and limit-sure winning characterization 
under infinite-memory, infinite-precision strategies were presented. 
In this work, we present (i)~a complete and precise characterization of the qualitative 
winning sets for bounded resource strategies, (ii)~symbolic algorithms to compute 
the winning sets, and (iii)~complexity results to determine whether a given 
state belongs to a qualitative winning set. 

\smallskip\noindent{\em Our results.}
Our contributions for bounded rationality in concurrent parity games 
are summarized below.
\begin{enumerate}

\item We show that pure memoryless strategies are as powerful as 
pure infinite-memory strategies in concurrent games.
This result is straight-forward, obtained by a simple reduction to turn-based 
probabilistic games.

\item Uniform memoryless strategies are more powerful than 
pure infinite-memory strategies (the fact that randomization is more powerful
than pure strategies follows from the classical matching pennies game), and we
show that uniform memoryless strategies are as powerful as finite-precision infinite-memory 
strategies. 
Thus our results show that if player~1 has only finite-precision strategies,
then no memory is required and uniform randomization is sufficient.
Hence very simple (uniform memoryless) controllers can be obtained for the
entire class of finite-precision infinite-memory controllers.
The result is obtained by a reduction to turn-based stochastic games, and 
the main technical contribution is the characterization of the winning 
sets for uniform memoryless strategies by a $\mu$-calculus formula.
The $\mu$-calculus formula not only gives a symbolic algorithm, but is also 
in the heart of other proofs of the paper.

\item In case of bounded-finite-precision strategies, the almost-sure and limit-sure 
winning sets coincide. 
For almost-sure winning, uniform memoryless strategies are also as 
powerful as infinite-precision finite-memory strategies. 
In contrast infinite-memory infinite-precision 
strategies are more powerful than uniform memoryless strategies for 
almost-sure winning. 
For limit-sure winning, we show that infinite-precision memoryless strategies 
are more powerful than bounded-finite-precision infinite-memory strategies, and 
infinite-precision memoryless strategies are as powerful as infinite-precision
finite-memory strategies. 
Our results show that if infinite-memory is not available, then no memory is required 
(memoryless strategies are as powerful as finite-memory strategies). 
The result is obtained by using the $\mu$-calculus formula for the uniform 
memoryless case: we show that a $\mu$-calculus formula that combines the 
$\mu$-calculus formula for almost-sure winning for uniform memoryless strategies
and limit-sure winning for reachability with memoryless strategies 
exactly characterizes the limit-sure winning for parity objectives 
for memoryless strategies.
The fact that we show that in concurrent parity games, though infinite-memory strategies
are necessary, memoryless strategies are as powerful as finite-memory strategies, is in contrast
with many other examples of graph games which require infinite-memory.
For example, in multi-dimensional games (such as multi-dimensional mean-payoff games) infinite-memory
strategies are necessary and finite-memory strategies are strictly more powerful than memoryless
strategies~\cite{CDHR10}.

\item As a consequence of the characterization of the winning sets as $\mu$-calculus
formulas we obtain symbolic algorithms to compute the winning sets. 
We show that the winning sets can be computed in $\bigo(n^{2d+3})$ time, 
where $n$ is the size of the game structure and $2d$ is the number of 
priorities (or colors), and our algorithms are symbolic.

\item 
The membership problem of whether a state belongs to a winning set
can be decided in NP $\cap$ coNP.

\end{enumerate}
In short, our results show that if infinite-memory is not available, then 
memory is useless, and if infinite-precision is not available, then uniform 
memoryless strategies are sufficient. 
Let $P,U,b\FP,\FP,\IP$ denote pure, uniform, bounded-finite-precision with bound $b$, 
finite-precision, and infinite-precision strategies, respectively, and $M,\FM,\IM$ denote memoryless, finite-memory, and 
infinite-memory strategies, respectively. 
For $A \in \set{P,U,b\FP,\FP,\IP}$ and $B \in \set{M,\FM,\IM}$, let $\almost_1(A,B,\Phi)$ 
denote the almost-sure winning set under player~1 strategies that are restricted to 
be both $A$ and $B$ for a parity objective $\Phi$ (and similar notation for $\limit_1(A,B,\Phi)$).
Then our results can be summarized by the following equalities and strict inclusion:
(first set of equalities and inequalities)
\[
\begin{array}{rcl}
\almost_1(P,M,\Phi) & = & \almost_1(P,\IM,\Phi) = \limit_1(P,\IM,\Phi) \\
& \subsetneq & \almost_1(U,M,\Phi) = \almost_1(\FP,\IM,\Phi)   \\
& = & \bigcup_{b>0}\limit_1(b\FP,\IM,\Phi) =\almost_1(\IP,\FM,\Phi) \subsetneq \almost_1(\IP,\IM,\Phi);
\end{array}
\]
and (the second set of equalities and inequalities)
\[
\begin{array}{rcl}
\bigcup_{b>0}\limit_1(b\FP,\IM,\Phi) & \subsetneq & \limit_1(\IP,M, \Phi) = \limit_1(\IP,\FM,\Phi) \\
& = & \limit_1(\FP,M,\Phi)=\limit_1(\FP,\IM,\Phi)\subsetneq \limit_1(\IP,\IM,\Phi).
\end{array}
\]

\noindent{\em Comparison with turn-based games and~\cite{CdAH11}.}
Our $\mu$-calculus formulas and the correctness proofs are non-trivial 
generalizations of both the result of~\cite{EJ91} for turn-based deterministic 
parity games and the result of~\cite{crg-tcs07} for concurrent reachability 
games.
Our algorithms, that are obtained by characterization of the winning sets 
as $\mu$-calculus formulas, are considerably more involved than those for 
turn-based games.
Our proof structure of using $\mu$-calculus formulas to characterize the winning sets,
though similar to~\cite{CdAH11}, has several new aspects. 
In contrast to the proof of~\cite{CdAH11} that constructs witness infinite-memory 
strategies for both players from the $\mu$-calculus formulas, our proof constructs memoryless 
witness strategies for player~1 from our new $\mu$-calculus formulas, and furthermore, we 
show that in the complement set of the $\mu$-calculus formulas for every finite-memory 
strategy for player~1 there is a witness memoryless spoiling strategy of the opponent.
Thus the witness strategy constructions are different from~\cite{CdAH11}.
Since our $\mu$-calculus formulas and the predecessor operators are different 
from~\cite{CdAH11} the proofs of the complementations of the $\mu$-calculus formulas are also 
different.
Moreover~\cite{CdAH11} only concerns limit-sure winning and not almost-sure winning.
Note that in~\cite{dAH00} both almost-sure and limit-sure winning was considered,
but as shown in~\cite{CdAH11} the predecessor operators suggested for limit-sure
winning (which was a nested stacked predecessor operator) in~\cite{dAH00} require 
modification for correctness proof, and similar modification is also required for 
almost-sure winning. Thus some results from~\cite{dAH00} related to almost-sure
winning require a careful proof (such as Example~\ref{examp-counter-almost}).

\smallskip\noindent{\em Techniques.} 
All results in the study of concurrent parity 
games~\cite{crg-tcs07,dAH00,dAM04,CdAH11} rely on $\mu$-calculus formulas.
One of the key difficulty in concurrent parity games is that the recursive 
characterization of turn-based games completely fail for concurrent games, and
thus leaves $\mu$-calculus as the only technique for analysis (see page~12 
for the explanation of why the recursive characterization fails). 
A $\mu$-calculus formula is a succinct way of describing a nested iterative 
algorithm, and thus provide a very general technique.
The key challenge and ingenuity is always to come up with the appropriate 
$\mu$-calculus formula with the right predecessor operators (i.e., the right 
algorithm), establish duality (or complementation of the formulas), and  
then construct from $\mu$-calculus formulas the witness strategies in concurrent
games (i.e., the correctness proof). 
Our results are also based on $\mu$-calculus formula characterization (i.e., 
nested iterative algorithms), however, the predecessor operators required 
and construction of witness strategies (which are the heart of the proofs) 
are quite different from the previous results.

\section{Definitions}

In this section we define game structures, strategies, objectives,  
winning modes and give other preliminary definitions; and the basic
definitions are exactly as in~\cite{CdAH11}.

\subsection{Game structures}\label{sec-gamestructures}
\noindent{\bf Probability distributions.}
For a finite set~$A$, a {\em probability distribution\/} on $A$ is a
function $\trans\!:A\mapsto[0,1]$ such that $\sum_{a \in A} \trans(a) = 1$.
We denote the set of probability distributions on $A$ by $\distr(A)$. 
Given a distribution $\trans \in \distr(A)$, we denote by $\supp(\trans) = 
\{x\in A \mid \trans(x) > 0\}$ the {\em support\/} of the distribution 
$\trans$.

\medskip\noindent{\bf Concurrent game structures.} 
A (two-player) {\em concurrent stochastic game structure\/} 
$\game = \langle S, \moves,\mov_1, \mov_2, \trans \rangle$ consists of the 
following components.

\begin{itemize}

\item A finite state space $S$.

\item A finite set $\moves$ of moves (or actions).

\item Two move assignments $\mov_1, \mov_2 \!: S\mapsto 2^\moves
	\setm \emptyset$.  For $i \in \{1,2\}$, assignment
	$\mov_i$ associates with each state $s \in S$ the nonempty
	set $\mov_i(s) \subs \moves$ of moves available to player $i$
	at state $s$.  For technical convenience, we assume that 
	$\mov_i(s) \inters \mov_j(t) = \emptyset$ unless $i=j$ and
	$s=t$, for all $i,j \in \{1,2\}$ and $s,t \in S$. 
	If this assumption is not met, then   
        the moves can be trivially renamed to satisfy the assumption.

\item A probabilistic transition function
	$\trans\!:S\times\moves\times\moves\mapsto \distr(S)$, which
	associates with every state $s \in S$ and moves $a_1 \in
	\mov_1(s)$ and $a_2 \in \mov_2(s)$ a probability
	distribution $\trans(s,a_1,a_2) \in \distr(S)$ for the
	successor state.
\end{itemize}

\medskip\noindent{\bf Plays.}
At every state $s\in S$, player~1 chooses a move $a_1\in\mov_1(s)$,
and simultaneously and independently
player~2 chooses a move $a_2\in\mov_2(s)$.  
The game then proceeds to the successor state $t$ with probability
$\trans(s,a_1,a_2)(t)$, for all $t \in S$. 
For all states $s \in S$ and moves $a_1 \in
\mov_1(s)$ and $a_2 \in \mov_2(s)$, we indicate by 
$\dest(s,a_1,a_2) = \supp(\trans(s,a_1,a_2))$ 
the set of possible successors of $s$ when moves $a_1$, $a_2$
are selected. 
A {\em path\/} or a {\em play\/} of $\game$ is an infinite sequence
$\pat = \langle s_0,s_1,s_2,\ldots\rangle$ of states in $S$ such that for all 
$k\ge 0$, there are moves $a^k_1 \in \mov_1(s_k)$ and $a^k_2 \in \mov_2(s_k)$
such that $s_{k+1} \in \dest(s_k,a^k_1,a^k_2)$.
We denote by $\pats$ the set of all paths. 
For a play $\pat = \seq{s_0, s_1, s_2,\ldots} \in \pats$,   
we define $\Inf(\pat) = 
\set{s \in S \mid \mbox{$s_k = s$ for infinitely many $k \geq 0$}}$
to be the set of states that occur infinitely often in~$\pat$.

\medskip\noindent{\bf Size of a game.} The \emph{size} of a concurrent 
game is the sum of the size of the state space and the number of the 
entries of the transition function. Formally the size of a game is 
$|S| + \sum_{s\in S,a \in \mov_1(s), b\in \mov_2(s)} |\dest(s,a,b)|$.

\smallskip\noindent{\bf Turn-based stochastic games and MDPs.} 
A game structure $\game$ is {\em turn-based stochastic\/} if at every
state at most one player can choose among multiple moves; that is, for
every state $s \in S$ there exists at most one $i \in \{1,2\}$ with
$|\mov_i(s)| > 1$. 
A game structure is a player-2 \emph{Markov decision process} if for all 
$s \in S$ we have $|\mov_1(s)|=1$, i.e., only player-2 has choice of 
actions in the game.

\smallskip\noindent{\bf Equivalent game structures.} 
Given two game structures $\game_1 = \langle S, \moves,\mov_1, \mov_2, \trans_1 \rangle$ and  
$\game_2 = \langle S, \moves,\mov_1, \mov_2, \trans_2 \rangle$ on the same 
state and action space, with a possibly different transition function, we say 
that $\game_1$ is equivalent to $\game_2$ (denoted $\game_1 \equiv \game_2$) if 
for all $s \in S$ and all $a_1 \in \mov_1(s)$ and $a_2 \in \mov_2(s)$ 
we have $\supp(\trans_1(s,a_1,a_2))= \supp(\trans_2(s,a_1,a_2))$.

\subsection{Strategies}\label{sec-strategy}
A {\em strategy\/} for a player is a recipe that describes how to 
extend a play.
Formally, a strategy for player $i\in\{1,2\}$ is a mapping 
$\stra_i\!:S^+\mapsto\distr(\moves)$ that associates with every nonempty 
finite sequence $x \in S^+$ of states, 
representing the past history of the game, 
a probability distribution $\stra_i(x)$ used to select
the next move. 
The strategy $\stra_i$ can prescribe only moves that are available to 
player~$i$;
that is, for all sequences $x\in S^*$ and states $s\in S$, we require that
$\supp(\stra_i(x\cdot s)) \subs \mov_i(s)$.  
We denote by $\bigstra_i$ the set of all strategies for player $i\in\{1,2\}$.

Given a state $s\in S$ and two strategies $\stra_1\in\bigstra_1$ and
$\stra_2\in\bigstra_2$, we define
$\outcome(s,\stra_1,\stra_2)\subseteq\pats$ to be the set of paths
that can be followed by the game, when the game starts {f}rom $s$ and
the players use the strategies $\stra_1$ and~$\stra_2$.  Formally,
$\seq{s_0, s_1, s_2, \ldots} \in \outcome(s,\stra_1,\stra_2)$ if $s_0=s$ and
if for all $k\ge 0$ there exist moves $a^k_1 \in \mov_1(s_k)$ and
$a^k_2 \in \mov_2(s_k)$ such that
\[
  \stra_1(s_0,\ldots,s_k)(a^k_1) > 0, 
  \eqspa 
  \stra_2(s_0,\ldots,s_k)(a^k_2) > 0, 
  \eqspa 
  s_{k+1} \in \dest(s_k,a^k_1,a^k_2). 
  \eqspa 
\]
Once the starting state $s$ and the strategies $\stra_1$ and $\stra_2$
for the two players have been chosen, 
the probabilities of events are uniquely defined~\cite{VardiP85}, where an {\em
event\/} $\cala\subseteq\pats$ is a measurable set of
paths\footnote{To be precise, we should define events as
measurable sets of paths {\em sharing the same initial state,} and
we should replace our events with families of events, indexed by their
initial state \cite{Kemeny}.  
However, our (slightly) improper definition leads to 
more concise notation.}. 
For an event $\cala\subseteq\pats$, we denote by
$\Prb_s^{\stra_1,\stra_2}(\cala)$ the probability that a path belongs to 
$\cala$ when the game starts {f}rom $s$ and the players use the strategies 
$\stra_1$ and~$\stra_2$.

\paragraph{Classification of strategies.}

We classify strategies according to their use of \emph{randomization} 
and \emph{memory}. 
We first present the classification according to randomization.
\begin{enumerate}

\item \emph{(Pure).} A strategy $\stra$ is {\em pure (deterministic)\/} if for 
all $x \in S^+$ there exists $a \in \moves$ such that $\stra(x)(a) = 1$.
Thus, deterministic strategies are equivalent to functions $S^+ \mapsto
\moves$.

\item \emph{(Uniform).} A strategy $\stra$ is {\em uniform\/} if for 
all $x \in S^+$ we have $\stra(x)$ is uniform over its support, i.e., 
for all $a \in \supp(\stra(x))$ we have 
$\stra(x)(a)= \frac{1}{ |\supp(\stra(x))|}$.

\item \emph{(Finite-precision).} For $b \in \Nats$, a strategy $\stra$ is 
{\em $b$-finite-precision\/} 
if for all $x \in S^+$ and all actions $a$ we have $\stra(x)(a)=\frac{i}{j}$, 
where $i,j \in \Nats$ and $0 \leq i \leq j \leq b$ and $j >0$, i.e.,
the probability of an action played by the strategy is a multiple of some 
$\frac{1}{\ell}$, with $\ell \in \Nats$ such that $\ell \leq b$.

\end{enumerate}
We denote by $\bigstra_i^P, \bigstra_i^U, \bigstra_i^{b\FP}, \bigstra_i^{\FP}=
\bigcup_{b>0} \bigstra_i^{b\FP}$ and $\bigstra_i^{\IP}$ the set of pure 
(deterministic), uniform, bounded-finite-precision with bound $b$, 
finite-precision, and infinite-precision (or general) strategies for 
player~$i$, respectively. 
Observe that we have the following strict inclusion: 
$\bigstra_i^P \subset \bigstra_i^U \subset \bigstra_i^{\FP} \subset \bigstra_i^{\IP}$.

\begin{enumerate}
\item \emph{(Finite-memory).} Strategies in general are \emph{history-dependent} and
can be represented as follows:
let $\mem$ be a set called \emph{memory} to remember the history of plays (the set
$\mem$ can be infinite in general).
A strategy with memory can be described as a pair of functions:
(a) a \emph{memory update} function $\stra_{u}: S \times \mem \mapsto \mem$,
that given the memory $\mem$ with the information about the history and
the current state updates the memory; and
(b) a \emph{next move} function $\stra_{n}: S\times\mem \mapsto \distr(\moves)$
that given the memory and the current state specifies the next move of
the player.
A strategy is \emph{finite-memory} if the memory $\mem$ is finite.

\item \emph{(Memoryless).} A \emph{memoryless} strategy is independent of the history of play and 
only depends on the current state.
Formally, for a memoryless strategy $\stra$ we have $\stra(x\cdot s) =
\stra (s)$ for all $s \in S$ and all $x \in S^*$.
Thus memoryless strategies are equivalent to functions 
$S \mapsto \distr(\moves)$.
\end{enumerate}
We denote by $\bigstra_i^M, \bigstra_i^{\FM}$ and 
$\bigstra_i^{\IM}$ the set of memoryless, 
finite-memory, and infinite-memory (or general) strategies 
for player~$i$, respectively. 
Observe that we have the following strict inclusion: 
$\bigstra_i^M \subset \bigstra_i^{\FM} \subset \bigstra_i^{\IM}$.

\subsection{Objectives}
We specify objectives for the players by providing 
the set of \emph{winning plays} $\Phi \subseteq \pats$ for each player.
In this paper we study only zero-sum games \cite{RagFil91,FilarVrieze97}, 
where the objectives of the two players are complementary.
A general class of objectives are the Borel objectives~\cite{Kechris}. 
A \emph{Borel objective} $\Phi \subseteq S^\omega$ is a Borel set in the 
Cantor topology on~$S^\omega$. 
In this paper we consider \emph{$\omega$-regular objectives}~\cite{Thom90},
which lie in the first $2\half$ levels of the Borel hierarchy
(i.e., in the intersection of $\Sigma_3$ and~$\Pi_3$).
We will consider the following $\omega$-regular objectives.

\begin{itemize}
\item
  \emph{Reachability and safety objectives.}
  Given a set $T\subseteq S$ of ``target'' states, the reachability
  objective requires that some state of $T$ be visited.
  The set of winning plays is thus 
  $\Reach(T) = 
  \set{\pat=\seq{s_0, s_1, s_2,\ldots} \in \pats \mid
    \exists k \ge 0.\ s_k \in T}$. 
  Given a set $F \subseteq S$, the safety objective requires that only 
  states of  $F$ be visited.
  Thus, the set of winning plays is 
  $\Safe(F) = \set{ \pat=\seq{s_0, s_1, s_2, \ldots} \in \Omega \mid
   \forall k \ge 0. \ s_k \in F}$.

\item \emph{B\"uchi and co-B\"uchi objectives.} 
  Given a set $B \subseteq S$ of ``B\"uchi'' states, 
  the B\"uchi objective requires that  $B$ is visited 
  infinitely often. 
  Formally, the set of winning plays is
  $\Buchi(B)= \set{ \pat \in \pats \mid \ \Inf(\pat) \cap
    B\neq \emptyset}$.
  Given $C \subseteq S$, the co-B\"uchi objective requires that all
  states visited infinitely often are in $C$. 
  Formally, the set  of winning plays is 
  $\coBuchi(C) = \set{ \pat \in \pats \mid \ \Inf(\pat) \subseteq C}$.

\item
  \emph{Parity objectives.}
  For $c,d \in \Nats$, we let $[c..d] = \set{c, c+1, \ldots, d}$. 
  Let $p : S \mapsto [0..d]$ be a function that assigns a \emph{priority}
  $p(s)$ to every state $s \in S$, where $d \in \Nats$.
  The \emph{Even parity objective} requires that the maximum 
  priority visited infinitely often is even. Formally, the set
  of winning plays is defined as
  $\ParityCond= 
  \set{\pat \in \pats \mid
  \max\big(p(\Inf(\pat))\big) \text{ is even }}$.
  The dual  \emph{Odd parity objective} is defined as
  $\coParityCond =
  \set{\pat \in \pats \mid
  \max\big(p(\Inf(\pat))\big) \text{ is odd }}$.
Note that for a priority function $p : S \mapsto \set{1,2}$, an even
parity objective $\ParityCond$ is equivalent to the B\"{u}chi
objective $\Buchi(p^{-1}(2))$, i.e., the B\"uchi set consists of the
states with priority~$2$.
Hence B\"uchi and co-B\"uchi objectives are simpler and special cases
of parity objectives.
\end{itemize}

Given a set $U \subs S$ we use usual LTL notations 
$\Box U, \Diamond U, \Box \Diamond U$ and $\Diamond \Box U$ to denote 
$\Safe(U), \Reach(U), \Buchi(U)$ and $\coBuchi(U)$, respectively.
Parity objectives are of special importance as they can express all
$\omega$-regular objectives, and hence all commonly used specifications 
in verification~\cite{Thom90}.

\subsection{Winning modes}

Given an objective $\Phi$,  
for all initial states $s \in S$, the set of paths $\Phi$
is measurable for all choices of the strategies of the player~\cite{VardiP85}. 
Given an initial state $s \in S$, an objective $\Phi$, and a class
$\bigstra_1^{\calc}$ of strategies we
consider the following {\em winning modes\/} for player~1:
\begin{description} 

\item[Almost.] 
We say that player~1 {\em wins almost surely\/} with the class 
$\bigstra_1^{\calc}$ if the player has a
strategy in $\bigstra_1^{\calc}$ to win with probability~1, or
$
  \exists \stra_1 \in \bigstra_1^{\calc} \qdot 
  \forall \stra_2 \in \bigstra_2 \qdot
  \Prb_s^{\stra_1,\stra_2}(\Phi) = 1
$.

\item[Limit.] 
We say that player~1 {\em wins limit surely\/} with the class 
$\bigstra_1^{\calc}$ if the player can ensure
to win with probability arbitrarily close to~1 with $\bigstra_1^{\calc}$, 
in other words, for all $\varepsilon>0$ there is a strategy for player~1 
in $\bigstra_1^{\calc}$ that ensures to win
with probability at least $1-\varepsilon$. Formally we have 
$
  \sup_{\stra_1 \in \bigstra_1^{\calc}}
    \inf_{\stra_2 \in \bigstra_2}
    \Prb_s^{\stra_1,\stra_2}(\Phi) = 1
$.

\end{description}
%
%
We abbreviate the winning modes by $\almost$ and $\limit$, respectively.
We call these winning modes the {\em qualitative\/} winning modes.
Given a game structure $G$,
for $C_1 \in \set{P,U,\FP,\IP}$ and $C_2\in \set{M,\FM,\IM}$ we denote 
by $\almost_1^G(C_1,C_2,\Phi)$ (resp. $\limit_1^G(C_1,C_2,\Phi)$) 
the set of almost-sure (resp. limit-sure) winning states for player~1 in $G$ 
when the strategy set for player~1 is restricted to 
$\bigstra_1^{C_1} \cap \bigstra_1^{C_2}$.
If the game structure $G$ is clear from the context we omit the superscript $G$.
Note that there is a subtle difference between the set 
$\bigcup_{b>0}\limit_1(b\FP,C_2,\Phi)$ that asks for a global bound on 
precision independent of $\varepsilon>0$ (i.e., a bound $b$ that is sufficient
for every $\ve>0$), whereas for the set 
$\limit_1(\bigcup_{b>0}b\FP,C_2,\Phi)=
\limit_1(\FP,C_2,\Phi)$ the bound on precision may depend on $\varepsilon>0$
(i.e., for every $\ve>0$ a bound $b$ on precision).

\subsection{Mu-calculus, complementation, and levels} 
\label{sec-levels}

Consider a mu-calculus expression $\Psi = \mu X \qdot \psi(X)$ over a
finite set $S$, where $\psi: 2^S \mapsto 2^S$ is monotonic.
The least fixpoint $\Psi = \mu X \qdot \psi(X)$ is equal
to the limit $\lim_{k \go \infty} X_k$, where $X_0 = \emptyset$,
and $X_{k+1} = \psi(X_k)$.  
For every state $s \in \Psi$, we define the {\em level\/} $k \geq 0$
of $s$ to be the integer such that $s \not\in X_k$ and $s \in X_{k+1}$.  
The greatest fixpoint $\Psi = \nu X \qdot \psi(X)$ is
equal to the limit $\lim_{k \go \infty} X_k$, where $X_0 = S$, and
$X_{k+1} = \psi(X_k)$.   
For every state $s \not\in \Psi$, we define the {\em level\/} $k \geq 0$ of
$s$ to be the integer such that $s \in X_k$ and $s \not\in X_{k+1}$.  
The {\em height\/} of a mu-calculus expression 
$\lambda X \qdot \psi(X)$, where $\lambda \in \set{\mu, \nu}$, 
is the least integer $h$ such that $X_h = \lim_{k \go \infty} X_k$.
An expression of height $h$ can be computed in $h+1$ iterations. 
Given a mu-calculus expression $\Psi=\lambda X \qdot \psi(X)$, where 
$\lambda \in \set{\mu, \nu}$, the complement $\no \Psi = S \setm \Psi$ of
$\lambda$ is given by 
$\overline{\lambda} X \qdot \no \psi (\no X)$, 
where $\overline{\lambda} = \mu$ if $\lambda = \nu$, and 
$\overline{\lambda} = \nu$ if $\lambda = \mu$.
For details of $\mu$-calculus see~\cite{Kozen83mu,EJ91}.

\smallskip\noindent{\bf Mu-calculus formulas and algorithms.}
As descrived above that $\mu$-calculus formulas $\Psi = \mu X \qdot \psi(X)$ 
(resp.  $\Psi = \nu X \qdot \psi(X)$) represent an iterative algorithm 
that successively iterates $\psi(X_k)$ till the least (resp. greatest)
fixpoint is reached. 
Thus in general, a $\mu$-calculus formulas with nested $\mu$ and $\nu$ 
operators represents a nested iterative algorithm.
Intuitively, a $\mu$-calculus formula is a succinct representation of
a nested iterative algorithm.

\medskip\noindent{\bf Distributions and one-step transitions.} 
Given a state $s \in S$, we denote by $\sd^s_1 = \distr(\mov_1(s))$
and $\sd^s_2 = \distr(\mov_2(s))$ the sets of probability
distributions over the moves at $s$ available to player 1 and~2, respectively. 
Moreover, for $s \in S$, $X \subs S$, $\dis_1 \in \sd^s_1$, and 
$\dis_2 \in \sd^s_2$ we denote by 
\[
  \pr_s^{\dis_1,\dis_2}(X) = 
	\sum_{a \in \mov_1(s)} \; 
	\sum_{b \in \mov_2(s)} \;
	\sum_{t \in X} \dis_1(a) \cdot \dis_2(b) \cdot \trans(s,a,b)(t) 
\]
the one-step probability of a transition into $X$ when players 1 and~2
play at $s$ with distributions $\dis_1$ and $\dis_2$, respectively. 
Given a state $s$ and distributions $\dis_1 \in \sd^s_1$ and 
$\dis_2 \in \sd^s_2$ we denote by 
$\dest(s,\dis_1,\dis_2)=\set{t \in S \mid \pr_2^{\dis_1,\dis_2}(t) >0}$ 
the set of states that have positive probability of transition from $s$ 
when the players play $\dis_1$ and $\dis_2$ at $s$.
For actions $a$ and $b$ we have $\dest(s,a,b) =
\set{t \in S \mid \trans(s,a,b)(t) >0}$ as the set of possible successors 
given $a$ and $b$.
For $A \subs \mov_1(s)$ and $B \subs \mov_2(s)$ we have 
$\dest(s,A,B)= \bigcup_{a\in A, b \in B} \dest(s,a,b)$.

\begin{theo}{}\label{theo-turn-based}
The following assertions hold:
\begin{enumerate}
\item \emph{\cite{CJH03}} For all turn-based stochastic game structures $G$ with a parity objective
$\Phi$ we have 
\[
\almost_1(P,M,\Phi) 
=\almost_1(\IP,\IM,\Phi) 
=\limit_1(P,M,\Phi)
=\limit_1(\IP,\IM,\Phi)
\] 
\item \emph{\cite{dAH00}} Let $G_1$ and $G_2$ be two equivalent game structures with a 
parity objective $\Phi$, then we have 
\[
1.\ \almost_1^{G_1}(\IP,\IM,\Phi) 
=\almost_1^{G_2}(\IP,\IM,\Phi); \quad
2.\ \limit_1^{G_1}(\IP,\IM,\Phi)
=\limit_1^{G_2}(\IP,\IM,\Phi)
\]
\end{enumerate}
\end{theo}

\section{Pure, Uniform and Finite-precision Strategies}
In this section we present our results for pure, uniform and 
finite-precision strategies. 
We start with the characterization for pure strategies.

\subsection{Pure strategies}
The following result shows that for pure strategies, memoryless strategies 
are as strong as infinite-memory strategies, and the almost-sure and
limit-sure sets coincide.
  
\begin{prop}{}\label{prop_pure}
Given a concurrent game structure $G$ and a parity objective $\Phi$ we have 
\[
\begin{array}{l}
\almost_1^G(P,M, \Phi) = \almost_1^G(P, \FM,\Phi) =\almost_1^G(P, \IM,\Phi) = \\\limit_1^G(P,M, \Phi) = \limit_1^G(P, \FM,\Phi) =\limit_1^G(P, \IM,\Phi). 
\end{array}
\]
\end{prop}
\begin{proof}
The result is obtained as follows: we show that 
$\almost_1^G(P,M,\Phi)= \almost_1^G(P,\IM,\Phi)=\limit_1^G(P,\IM,\Phi)$ and all 
the other equalities follow (by inclusion of strategies). 
The main argument is as follows: given $G$ we construct a turn-based stochastic
game $\wh{G}$ where player~1 first choses an action, then player~2 chooses an 
action, and then the game proceeds as in $G$. 
Intuitively, we divide each step of the concurrent game into two steps, 
in the first step player~1 chooses an action, and then player~2 
responds with an action in the second step.
Then it is straightforward to establish that the almost-sure (resp. limit-sure)
winning set for pure and infinite-memory strategies in $G$ coincides with 
the almost-sure (resp. limit-sure) winning set for pure and infinite-memory 
strategies in $\wh{G}$. 
Since $\wh{G}$ is a turn-based stochastic game, by Theorem~\ref{theo-turn-based} (part~1), 
it follows that the almost-sure and limit-sure winning set in $\wh{G}$ coincide and 
they are same for memoryless and infinite-memory strategies.

We now present the formal reduction. 
Let $G= \langle S, \moves, \mov_1,\mov_2,\trans \rangle$ and let the 
parity objective $\Phi$ be described by a priority function $p$. 
We construct $\wh{G}= \langle \wh{S}, \wh{\moves}, \wh{\mov}_1, \wh{\mov}_2,
\wh{\trans} \rangle$ with priority function $\wh{p}$ as follows:
\begin{enumerate}
\item $\wh{S}= S \cup \set{(s,a) \mid s \in S, a \in \mov_1(s)}$; 
\item $\wh{\moves} = \moves \cup \set{\bot}$ where $\bot \not\in \moves$; 
\item for $s \in \wh{S} \cap S$ we have 
$\wh{\mov}_1(s)=\mov_1(s)$ and $\wh{\mov}_2(s) =\set{\bot}$; and 
for $(s,a) \in \wh{S}$ we have 
$\wh{\mov}_2((s,a))=\mov_2(s)$ and $\wh{\mov}_1((s,a)) =\set{\bot}$; 
and 
\item for $s \in \wh{S} \cap S$ and $a \in \mov_1(s)$ we have 
$\wh{\trans}(s,a,\bot)(s,a)=1$; 
and for $(s,a) \in \wh{S}$ and $b \in \mov_2(s)$ we have
$\wh{\trans}((s,a),\bot,b)=\trans(s,a,b)$;
\item the function $\wh{p}$ in $\wh{G}$ is as follows: 
for $s \in \wh{S} \cap S$ we have $\wh{p}(s)=p(s)$ and 
for $(s,a) \in \wh{S}$ we have $\wh{p}((s,a))=p(s)$.
\end{enumerate}
Observe that the reduction is linear (i.e., $\wh{G}$ is linear in the 
size of $G$).
It is straightforward to establish by mapping of pure strategies of 
player~1 in $G$ and $\wh{G}$ that 
\[
\begin{array}{lcl}
(a)\ \almost_1^G(P,M,\Phi) & = & \almost_1^{\wh{G}}(P,M,\wh{\Phi}) \cap S, \\ 
(b)\ \almost_1^G(P,\IM,\Phi) & = & \almost_1^{\wh{G}}(P,\IM,\wh{\Phi}) \cap S, \\
(c)\ \limit_1^G(P,M,\Phi) & = & \limit_1^{\wh{G}}(P,M,\wh{\Phi}) \cap S, \\
(d)\ \limit_1^G(P,\IM,\Phi) & = & \limit_1^{\wh{G}}(P,\IM,\wh{\Phi}) \cap S;
\end{array}
\]
where $\wh{\Phi}=\Parity(\wh{p})$.
It follows from Theorem~\ref{theo-turn-based} (part~1) that  
\[
\almost_1^{\wh{G}}(P,M,\wh{\Phi})= \almost_1^{\wh{G}}(P,\IM,\wh{\Phi})= 
\limit_1^{\wh{G}}(P,M,\wh{\Phi}) =\limit_1^{\wh{G}}(P,\IM,\wh{\Phi}).
\]
Hence the desired result follows.
\qed
\end{proof}

\noindent{\bf Algorithm and complexity.} The proof of the above proposition 
gives a linear reduction to turn-based stochastic games. 
Thus the set $\almost_1(P,M,\Phi)$ can be computed using the algorithms 
for turn-based stochastic parity games (such as~\cite{CJH03}).
We have the following results.

\begin{theo}{}
Given a concurrent game structure $G$, a parity objective $\Phi$, and a 
state $s$, whether $s \in \almost_1(P,\IM,\Phi)=\limit_1(P,\IM,\Phi)$ 
can be decided in NP $\cap$ coNP.
\end{theo}

\subsection{Uniform and Finite-precision}
In this subsection we will present the characterization for uniform and 
finite-precision strategies.

\begin{examp}{}\label{examp-pure-uniform}
It is easy to show that $\almost_1(P,M,\Phi) \subsetneq \almost_1(U,M,\Phi)$
by considering the \emph{matching penny} game.
The game has two states $s_0$ and $s_1$. 
The state $s_1$ is an \emph{absorbing} state (a state with only self-loop as 
outgoing edge; see state $s_1$ of Fig~\ref{figure:buchi-lim}) and the goal is 
to reach $s_1$ (equivalently infinitely often visit $s_1$). 
At $s_0$ the actions available for both players are $\set{a,b}$. 
If the actions match the next state is $s_1$, otherwise $s_0$. 
By playing $a$ and $b$ uniformly at random at $s_0$, the 
state $s_1$ is reached with probability~1,
whereas for any pure strategy the counter-strategy that plays exactly the
opposite action in every round ensures $s_1$ is never reached.
\qed
\end{examp}

We now show that uniform memoryless strategies are as powerful as 
finite-precision infinite-memory strategies and the almost-sure and 
limit-sure sets coincide for finite-precision strategies.
We start with two notations.

\smallskip\noindent{\bf Uniformization of a strategy.} 
Given a strategy $\stra_1$ for player~1, we define a strategy $\stra_1^u$ 
that is obtained from $\stra_1$ by uniformization as follows: 
for all $w \in S^+$ and all $a \in \supp(\stra_1(w))$ we have 
$\stra_1^u(w)(a) = \frac{1}{|\supp(\stra_1(w))|}$.
We will use the following notation for uniformization: 
$\stra_1^u=\unif(\stra_1)$.


\begin{prop}{}\label{prop-uniform-fp}
Given a concurrent game structure $G$ and a parity objective $\Phi$ we have 
\[
\begin{array}{l}
\almost_1^G(U,M, \Phi) = \almost_1^G(U, \FM,\Phi) =\almost_1^G(U, \IM,\Phi) = \\
\limit_1^G(U,M, \Phi) = \limit_1^G(U, \FM,\Phi) =\limit_1^G(U, \IM,\Phi) = \\
\almost_1^G(\FP,M, \Phi) = \almost_1^G(\FP, \FM,\Phi) =\almost_1^G(\FP, \IM,\Phi) = \\
\bigcup_{b>0}\limit_1^G(b\FP,M, \Phi) = \bigcup_{b>0}\limit_1^G(b\FP, \FM,\Phi) =\bigcup_{b>0}\limit_1^G(b\FP, \IM,\Phi) 
\end{array}
\]
\end{prop}
\begin{proof}
The result is obtained as follows: we show that 
$\almost_1^G(U,M,\Phi)= \almost_1^G(\FP,\IM,\Phi)=\limit_1^G(\FP,\IM,\Phi)$ 
and all the other equalities follow (by inclusion of strategies). 
The key argument is as follows: fix a bound $b$, and we consider the set of 
$b$-finite-precision strategies in $G$. 
Given $G$ we construct a turn-based stochastic game $\wt{G}$ where player~1 first 
chooses a $b$-finite-precision distribution, then player~2 chooses an 
action, and then the game proceeds as in $G$. 
Intuitively, we divide each step of the concurrent game into two steps, 
in the first step player~1 chooses a $b$-finite precision distribution, 
and then in the second step player~2 responds with an action.
Then we establish that the almost-sure (resp. limit-sure) winning set 
for $b$-finite-precision and infinite-memory strategies in $G$ coincides with 
the almost-sure (resp. limit-sure) winning set for $b$-finite-precision  
and infinite-memory strategies in $\wt{G}$. 
Since $\wt{G}$ is a turn-based stochastic game, by Theorem~\ref{theo-turn-based}, it follows
that the almost-sure and limit-sure winning set in $\wt{G}$ coincide and 
they are same for memoryless and infinite-memory strategies. 
Thus we obtain a $b$-finite-precision memoryless almost-sure winning strategy
$\stra_1$ in $G$ and then we show the uniform memoryless 
$\stra_1^u=\unif(\stra_1)$ obtained from uniformization of $\stra_1^u$ is 
a uniform memoryless almost-sure winning strategy in $G$. 
Thus it follows that for any finite-precision infinite-memory almost-sure 
winning strategy, there is a uniform memoryless almost-sure winning strategy.

We now present the formal reduction. 
Let $G= \langle S, \moves, \mov_1,\mov_2,\trans \rangle$ and let the 
parity objective $\Phi$ be described by a priority function $p$. 
For a given bound $b$, 
let $\wt{f}(s,b)= \set{f: \mov_1(s) \mapsto [0,1] \mid \forall a \in \mov_1(s)
\text{ we have } f(a) = \frac{i}{j}, i,j \in \Nats, 0 \leq i\leq j\leq b, j>0
\text{ and } \sum_{a \in \mov_1(s)} f(a)=1 
}$ 
denote the set of $b$-finite-precision distributions at $s$.
We construct $\wt{G}= \langle \wt{S}, \wt{\moves}, \wt{\mov}_1, \wt{\mov}_2,
\wt{\trans} \rangle$ with priority function $\wt{p}$ as follows:
\begin{enumerate}
\item $\wt{S}= S \cup \set{(s,f) \mid s \in S, f \in \wt{f}(s,b)}$; 
\item $\wt{\moves} = \moves \cup \set{f \mid s\in S, f \in \wt{f}(s,b)} \cup \set{\bot}$ where $\bot \not\in \moves$; 
\item for $s \in \wt{S} \cap S$ we have 
$\wt{\mov}_1(s)=\wt{f}(s,b)$ and $\wt{\mov}_2(s) =\set{\bot}$; and 
for $(s,f) \in \wt{S}$ we have 
$\wt{\mov}_2((s,f))=\mov_2(s)$ and $\wt{\mov}_1((s,f)) =\set{\bot}$; 
and 
\item for $s \in \wt{S} \cap S$ and $f \in \wt{f}(s,b)$ we have 
$\wt{\trans}(s,f,\bot)(s,f)=1$; 
and for $(s,f) \in \wt{S}$, $b \in \mov_2(s)$ and $t \in S$ we have
$\wh{\trans}((s,f),\bot,b)(t)=\sum_{a \in \mov_1(s)} f(a) \cdot 
\trans(s,a,b)(t)$;
\item the function $\wt{p}$ in $\wt{G}$ is as follows: 
for $s \in \wt{S} \cap S$ we have $\wt{p}(s)=p(s)$ and 
for $(s,f) \in \wt{S}$ we have $\wt{p}((s,f))=p(s)$.
\end{enumerate}
Observe that given $b\in \Nats$ the set $\wt{f}(s,b)$ is finite and 
thus $\wt{G}$ is a finite-state turn-based stochastic game.
It is straightforward to establish mapping of $b$-finite-precision
strategies of player~1 in $G$ and with pure strategies in $\wh{G}$, i.e.,
we have 
\[
\begin{array}{lcl}
(a)\ \almost_1^G(b\FP,M,\Phi) & = & \almost_1^{\wt{G}}(P,M,\wt{\Phi}) \cap S, \\ 
(b)\ \almost_1^G(b\FP,\IM,\Phi) & = & \almost_1^{\wt{G}}(P,\IM,\wt{\Phi}) \cap S, \\
(c)\ \limit_1^G(b\FP,M,\Phi) & = & \limit_1^{\wt{G}}(P,M,\wt{\Phi}) \cap S, \\
(d)\ \limit_1^G(b\FP,\IM,\Phi) & = & \limit_1^{\wt{G}}(P,\IM,\wt{\Phi}) \cap S,
\end{array}
\]
where $\wt{\Phi}=\Parity(\wt{p})$ and $b\FP$ denote the set of $b$-finite-precision strategies in $G$.
By Theorem~\ref{theo-turn-based} we have 
\[
\almost_1^{\wt{G}}(P,M,\wt{\Phi})= \almost_1^{\wt{G}}(P,\IM,\wt{\Phi})= 
\limit_1^{\wt{G}}(P,M,\wt{\Phi}) =\limit_1^{\wt{G}}(P,\IM,\wt{\Phi}).
\]
Consider a pure memoryless strategy $\wt{\stra}_1$ in $\wt{G}$ that 
is almost-sure winning from $Q=\almost_1^{\wt{G}}(P,M,\wt{\Phi})$, and let 
$\stra_1$ be the corresponding $b$-finite-precision memoryless strategy 
in $G$.
Consider the uniform memoryless strategy $\stra_1^u=\unif(\stra_1)$ in 
$G$. 
The strategy $\stra_1$ is an almost-sure winning strategy from $Q \cap S$.
The player-2 MDP $G_{\stra_1}$ and $G_{\stra_1^u}$ are equivalent, i.e., 
$G_{\stra_1} \equiv G_{\stra_1^u}$ and hence it follows from Theorem~\ref{theo-turn-based} 
that $\stra_1^u$ is an almost-sure winning strategy for all states in $Q \cap S$.
Hence the desired result follows.
\qed
\end{proof}

\smallskip\noindent{\bf Computation of $\almost_1(U,M,\Phi)$.} 
It follows from Proposition~\ref{prop-uniform-fp} that the computation of 
$\almost_1(U,M,\Phi)$ can be achieved 
by a reduction to turn-based stochastic game. 
We now present the main technical result of this subsection which presents a 
symbolic algorithm to compute $\almost_1(U,M,\Phi)$.
The symbolic algorithm developed in this section is crucial for analysis of 
infinite-precision finite-memory strategies, where the reduction to turn-based 
stochastic game cannot be applied.
The symbolic algorithm is obtained via $\mu$-calculus formula characterization.
We first discuss the comparison of our proof with the results of~\cite{CdAH11} 
and then discuss why the recursive characterization of turn-based games fails
in concurrent games.

\smallskip\noindent{\bf Comparison with~\cite{CdAH11}.} Our proof structure based
on induction on the structure of $\mu$-calculus formulas is similar to the 
proofs in~\cite{CdAH11}. In some aspects the proofs are tedious adaptation but
in most cases there are many subtle issues and we point them below. 
First, in our proof the predecessor operators are different from the 
predecessor operators of~\cite{CdAH11}.
Second, in our proof from the $\mu$-calculus formulas we construct uniform 
memoryless strategies as compared to infinite memory strategies in~\cite{CdAH11}.
Finally, since our predecessor operators are different the proof for 
complementation of the predecessor operators (which is a crucial component 
of the proof) is completely different.

\smallskip\noindent{\bf Failue of recursive characterization.} 
In case of turn-based games there are recursive characterization of the 
winning set with attractors (or alternating reachability).
However such characterization fails in case of concurrent games.
The intuitive reason is as follows: once an attractor is taken it 
may rule out certain action pairs (for example, action pair $a_1$ and
$b_1$ must be ruled out, whereas action pair $a_1$ and $b_2$ may be 
allowed in the remaining game graph), and hence the complement of an 
attractor maynot satisfy the required sub-game property.
We now elaborate the above discussion. 
The failure of attractor based characterization is probably best explained 
for limit-sure coB\"uchi games~\cite{dAH00}. 
In turn-based games the coB\"uchi algorithm is as follows: 
(i)~compute safety winning region; (ii)~compute attractor to the safety winning
region and then obtain a sub-game and recurse. 
In concurrent games the first step is to compute safety winning region $X_0$.
In the next iteration what needs to be computed is the set where player~1 can
ensure either safety or limit-sure reachability to $X_0$.
However, player~1 may fail to ensure limit-sure reachability to $X_0$ because 
player~2 may then play to violate limit-sure reachability but safety is ensured; 
and player~1 may fail to ensure safety because player~2 may play so that 
safety is violated but limit-sure reachability succeeds. 
Intuitively player~1 will play distributions to ensure that no matter what 
player~2 plays either safety or limit-sure reachability is ensured, but player~1 
cannot control which one; the distribution of player~1 takes care of certain
actions of player~2 due to safety and other actions due to limit-sure reachability.
This intuitively means the limit-sure reachability and safety cannot be decoupled
and they are present in the level of the predecessor operators.
Also limit reachability does not depend only on the target, but also an outer 
fix point (which also makes recursive characterization only based on the target
set difficult).
For limit-sure reachability an outer greatest fix point is added, and nested
with two inner fix points for safety or limit reachability.
Informally, limit-sure reachability rules out certain pair of actions, but 
that cannot be used to obtain a subgame as it is inside a greatest fix point,
and moreover ruling out pair of actions does not give a subgame
(see Section~4.2 of ~\cite{dAH00} for further details).
While this very informally describes the issues for coB\"uchi games,
in general the situation is more complicated, because for more priorities 
the solution is limit-sure for parity with one less priority or limit-sure 
reachability, and for general limit-sure winning with parity objectives 
infinie-memory is required as compared to safety where memoryless strategies 
are sufficient; and all the complications need to be handled at the level of 
predecessor operators.
Thus for qualitative analysis of concurrent games the only formal way (known so far) 
to express the winning sets is by $\mu$-calculus. 
For details, see examples in~\cite{dAH00,crg-tcs07} why the recursive characterization fails. 
In concurrent games while the winning sets are described as $\mu$-calculus formulas,
the challenge always is to
(i)~come up with the right $\mu$-calculus formula with the 
appropriate predecessor operators (i.e., the right algorithm), 
(ii)~construct witness winning strategies from the 
$\mu$-calculus formulas (i.e., the correctness proof), and 
(iii)~show the complementation of the $\mu$-calculus formulas 
(i.e., the correctness for the opponent).

\smallskip\noindent{\bf Strategy constructions.} Since the recursive 
characterization of turn-based games fails for concurrent games, our 
results show that the generalization of the $\mu$-calculus formulas 
for turn-based games can characterize the desired winning sets. 
Moreover, our correctness proofs that establish the correctness of
the $\mu$-calculus formulas present explicit witness strategies 
from the $\mu$-calculus formulas.
Morover, in all cases the witness counter strategies for player~2 is 
memoryless, and thus our results answer questions related to 
bounded rationality for both players.

We now introduce the predecessor operators for the $\mu$-calculus formula
required for our symbolic algorithms.

\smallskip\noindent{\bf Basic predecessor operators.} 
We recall the {\em predecessor\/} operators $\pre_1$ (pre) and $\apre_1$ (almost-pre), defined for all $s \in S$ and $X,Y \subs S$ by: 
\[
\begin{array}{rcl}
\pre_1(X) &= & \set{s \in S \mid  
	\exists \dis_1 \in \sd^s_1 \qdot 
	\forall \dis_2 \in \sd^s_2 \qdot 
	\pr_s^{\dis_1,\dis_2} (X) = 1}; \\
\apre_1(Y,X)  &= & \set{s \in S \mid 
	\exists \dis_1 \in \sd^s_1 \qdot 
	\forall \dis_2 \in \sd^s_1 \qdot
	\pr_s^{\dis_1,\dis_2} (Y)=1 \land 
	\pr_s^{\dis_1,\dis_2} (X) > 0 } 
  \eqpun . 
\end{array}
\]
Intuitively, the $\pre_1(X)$ is the set of states such that player~1 can ensure 
that the next state is in $X$ with probability~1, and 
$\apre_1(Y,X)$ is the set of states such that player~1 can ensure that the 
next state is in $Y$ with probability~1 and in $X$ with positive probability.

\smallskip\noindent{\bf Principle of general predecessor operators.} 
While the operators $\apre$ and $\pre$ suffice for solving B\"uchi
games, for solving general parity games, we require predecessor operators 
that are best understood as the combination of the basic predecessor operators. 
We use the operators $\sdcup$ and $\sdcap$ to combine predecessor 
operators; the operators $\sdcup$ and $\sdcap$ are different from 
the usual union $\cup$ and intersection $\cap$.
Roughly, let $\alpha$ and $\beta$ be two set of states for 
two predecessor operators, 
then the set $\alpha \sdcap \beta$ requires that the
distributions of player~1 satisfy the conjunction of the
conditions stipulated by $\alpha$ and $\beta$; similarly, 
$\sdcup$ corresponds to disjunction. 
We first introduce the operator $\apre \sdcup \pre$.
For all $s \in S$ and $X_1,Y_0, Y_1 \subs S$, we define 
\begin{eqnarray} 
  \apre_1(Y_1,X_1)  \sdcup   \pre_1(Y_0) = 
	\setb{s \in S  \mid  
	\exists \dis_1 \in \sd^s_1.
	\forall \dis_2 \in \sd^s_2. 
	\left[\begin{array}{c}
	  (\pr_s^{\dis_1,\dis_2}(X_1) > 0 \land \pr_s^{\dis_1,\dis_2}(Y_1)=1) \\
	  \bigvee \\
	  \pr_s^{\dis_1,\dis_2}(Y_0) =1
	\end{array} \right]
	} 
  \eqpun.
\nonumber
\end{eqnarray}
Note that the above formula corresponds to a disjunction of the
predicates for $\apre_1$ and $\pre_1$.
However, it is important to note that the distributions $\dis_1$ for 
player~1 to satisfy  ($\dis_2$ for player~2 to falsify) the predicate must 
be {\em the same.} 
In other words, $\apre_1(Y_1,X_1)  \sdcup   \pre_1(Y_0)$ is 
{\em not\/} equivalent to $\apre_1(Y_1,X_1)  \cup   \pre_1(Y_0)$.

\smallskip\noindent{\bf General predecessor operators.}
We first introduce two predecessor operators as follows:
\begin{eqnarray*}
\lefteqn{\lpreodd_1(i,Y_n,X_n,\ldots, Y_{n-i}, X_{n-i})} \\[1ex]
& = & \apre_1(Y_n,X_n) \sdcup \apre_1(Y_{n-1},X_{n-1}) \sdcup \cdots \sdcup
\apre_1(Y_{n-i},X_{n-i}); \\[2ex] 
\lefteqn{\lpreeven_1(i,Y_n,X_n,\ldots, Y_{n-i}, X_{n-i},Y_{n-i-1})} \\[1ex] 
& = & \apre_1(Y_n,X_n) \sdcup \apre_1(Y_{n-1},X_{n-1}) \sdcup \cdots \sdcup
\apre_1(Y_{n-i},X_{n-i}) \sdcup \pre_1(Y_{n-i-1}). 
\end{eqnarray*}
The formal expanded definitions of the above operators are as follows:
\[
\begin{array}{l}
\lpreodd_1(i,Y_n,X_n,\ldots, Y_{n-i}, X_{n-i}) =  \\[1ex] 
\setb{s \in S  \mid  
\exists \dis_1 \in \sd^s_1.
\forall \dis_2 \in \sd^s_2. 
\left[\begin{array}{c}
  (\pr_s^{\dis_1,\dis_2}(X_n) > 0 \land \pr_s^{\dis_1,\dis_2}(Y_n)=1) \\
  \bigvee \\
  (\pr_s^{\dis_1,\dis_2}(X_{n-1}) > 0 \land \pr_s^{\dis_1,\dis_2}(Y_{n-1}) =1) \\
  \bigvee \\
  \vdots \\
  \bigvee \\
  (\pr_s^{\dis_1,\dis_2}(X_{n-i}) > 0 \land  \pr_s^{\dis_1,\dis_2}(Y_{n-i}) =1) 
\end{array} \right]
} \eqpun.
\end{array}
\]

\[
\begin{array}{l}
\lpreeven_1(i,Y_n,X_n,\ldots, Y_{n-i}, X_{n-i},Y_{n-i-1}) =  \\[1ex] 
\setb{s \in S  \mid  
\exists \dis_1 \in \sd^s_1.
\forall \dis_2 \in \sd^s_2. 
\left[\begin{array}{c}
  (\pr_s^{\dis_1,\dis_2}(X_n) > 0 \land \pr_s^{\dis_1,\dis_2}(Y_n)=1) \\
  \bigvee \\
  (\pr_s^{\dis_1,\dis_2}(X_{n-1}) > 0 \land \pr_s^{\dis_1,\dis_2}(Y_{n-1}) =1) \\
  \bigvee \\
  \vdots \\
  \bigvee \\
  (\pr_s^{\dis_1,\dis_2}(X_{n-i}) > 0 \land  \pr_s^{\dis_1,\dis_2}(Y_{n-i}) =1) \\
  \bigvee \\
 (\pr_s^{\dis_1,\dis_2}(Y_{n-i-1})=1) 
\end{array} \right]
} \eqpun.
\end{array}
\]
Observe that the above definition can be inductively written as follows:
\begin{enumerate}
\item We have $\lpreodd_1(0,Y_n,X_n)=  \apre_1(Y_n,X_n)$ and for $i\geq 1$ we have 
\begin{eqnarray*}
\lefteqn{\lpreodd_1(i,Y_n,X_n,\ldots,Y_{n-i},X_{n-i})} \\[1ex]
 & = & \apre_1(Y_n,X_n)  \sdcup  \lpreodd_1(i-1,Y_{n-1},X_{n-1},\ldots,Y_{n-i},X_{n-i})  
\end{eqnarray*}

\item We have $\lpreeven_1(0,Y_n,X_n,Y_{n-1}) =  \apre_1(Y_n,X_n) \sdcup \pre_1(Y_{n-1})$ and 
for $i \geq 1$ we have 
\begin{eqnarray*}
\lefteqn{\lpreeven_1(i,Y_n,X_n,\ldots,Y_{n-i},X_{n-i},Y_{n-i-1})} \\[1ex]
 & = & \apre_1(Y_n,X_n)  \sdcup \lpreeven_1(i-1,Y_{n-1},X_{n-1},\ldots,Y_{n-i},X_{n-i},Y_{n-i-1})  
\end{eqnarray*}
\end{enumerate}

\smallskip\noindent{\bf Dual operators.} 
The {\em predecessor\/} operators $\epre_2$ (positive-pre) and 
$\apre_2$ (almost-pre), defined for all $s \in S$ and $X,Y \subs S$ by: 
\[
\begin{array}{rcl}
\epre_2(X) &= & \set{s \in S \mid  
	\forall \dis_1 \in \sd^s_1 \qdot 
	\exists \dis_2 \in \sd^s_2 \qdot 
	\pr_s^{\dis_1,\dis_2} (X) > 0}; \\
\apre_2(Y,X)  &= & \set{s \in S \mid 
	\forall \dis_1 \in \sd^s_1 \qdot 
	\exists \dis_2 \in \sd^s_1 \qdot
	\pr_s^{\dis_1,\dis_2} (Y)=1 \land 
	\pr_s^{\dis_1,\dis_2} (X) > 0 } 
  \eqpun . 
\end{array}
\]
Observe that player~2 is only required to play counter-distributions $\dis_2$ 
against player~1 distributions $\dis_1$.
We now introduce two positive predecessor operators as follows:
\begin{eqnarray*}
\lefteqn{\fpreodd_2(i,Y_n,X_n,\ldots,Y_{n-i},X_{n-i})} \\[1ex]
& = & \epre_2(Y_n) \sdcup
\apre_2(X_{n},Y_{n-1}) \sdcup \cdots \sdcup \apre_2(X_{n-i+1},Y_{n-i}) \sdcup \pre_2(X_{n-i}) \\[2ex]
\lefteqn{\fpreeven_2(i,Y_n,X_n,\ldots,Y_{n-i},X_{n-i},Y_{n-i-1})} \\[1ex]
& = & \epre_2(Y_n) \sdcup
\apre_2(X_{n},Y_{n-1}) \\[1ex]
& & \qquad \sdcup \cdots \sdcup \apre_2(X_{n-i+1},Y_{n-i}) \sdcup \apre_2(X_{n-i},Y_{n-i-1}) 
\end{eqnarray*}
The formal expanded definitions of the above operators are as follows:
\[
\begin{array}{l}
\fpreodd_2(i,Y_n,X_n,\ldots, Y_{n-i},X_{n-i}) = \\[1ex]
 \setb{s \in S \mid   
  \forall \dis_1  \in \sd^s_1. 
  \exists \dis_2  \in \sd^s_2. 
  \left[ \begin{array}{c}
	(\pr_s^{\dis_1,\dis_2}(Y_n)>0) \\
	\bigvee \\
	(\pr_s^{\dis_1,\dis_2}(Y_{n-1})  > 0 \land \pr_s^{\dis_1,\dis_2}(X_{n}) =1) \\
        \bigvee \\
	(\pr_s^{\dis_1,\dis_2}(Y_{n-2})  > 0 \land \pr_s^{\dis_1,\dis_2}(X_{n-1}) =1) \\
        \bigvee \\
        \vdots \\
        \bigvee \\
	(\pr_s^{\dis_1,\dis_2}(Y_{n-i})  > 0 \land \pr_s^{\dis_1,\dis_2}(X_{n-i+1}) =1) \\
        \bigvee \\
	(\pr_s^{\dis_1,\dis_2}(X_{n-i})=1) 
  \end{array} \right]
 } \eqpun .
\end{array}
\]
\[
\begin{array}{l}
\fpreeven_2(i,Y_n,X_n,\ldots, Y_{n-i},X_{n-i},Y_{n-i-1}) = \\[1ex]
\setb{s \in S \mid   
\forall \dis_1  \in \sd^s_1. 
\exists \dis_2  \in \sd^s_2. 
  \left[ \begin{array}{c}
	(\pr_s^{\dis_1,\dis_2}(Y_n)>0 ) \\
	\bigvee \\
	(\pr_s^{\dis_1,\dis_2}(Y_{n-1})  > 0 \land \pr_s^{\dis_1,\dis_2}(X_{n}) =1) \\
        \bigvee \\
	(\pr_s^{\dis_1,\dis_2}(Y_{n-2})  > 0 \land \pr_s^{\dis_1,\dis_2}(X_{n-1}) =1) \\
        \bigvee \\
        \vdots \\
        \bigvee \\
	(\pr_s^{\dis_1,\dis_2}(Y_{n-i-1})  >0 \land \pr_s^{\dis_1,\dis_2}(X_{n-i}) =1) 
  \end{array} \right]
 } \eqpun .
\end{array}
\]
The above definitions can be alternatively written as follows
\[
\begin{array}{rcl}
\fpreodd_2(i,Y_n,X_n,\ldots,Y_{n-i},X_{n-i}) & = & \\[1ex]
 \epre_2(Y_n)  &\sdcup & \lpreeven_2(i-1,X_{n},Y_{n-1},\ldots,X_{n-i+1},Y_{n-i},X_{n-i}); 
\end{array}
\]
\[
\begin{array}{rcl}
\fpreeven_2(i,Y_n,X_n,\ldots,Y_{n-i},X_{n-i},Y_{n-i-1}) & = & \\[1ex]
 \epre_2(Y_n)  &\sdcup & \lpreodd_2(i,X_{n},Y_{n-1},\ldots,X_{n-i},Y_{n-i-1}).
\end{array}
\]

\begin{remark}
Observe that if the predicate $\epre_2(Y_n)$ is removed from the 
predecessor operator $\fpreodd_2(i,Y_n,X_n,\ldots,Y_{n-i},X_{n-i})$ 
(resp. $\fpreeven_2(i,Y_n,X_n,\ldots,Y_{n-i},X_{n-i},Y_{n-i-1})$), 
then we obtain the operator  
$\lpreeven_2(i-1,X_{n},Y_{n-1},\ldots,X_{n-i+1},Y_{n-i},X_{n-i})$ 
(resp.  $\lpreodd_2(i,X_{n},Y_{n-1},\ldots,X_{n-i},Y_{n-i-1})$).
\end{remark}

We first show how to characterize the set of almost-sure winning states for 
uniform memoryless strategies  and its complement for parity games 
using the above predecessor operators.
We will prove the following result by induction.

\begin{enumerate}

\item \emph{Case~1.} For a parity function $p:S \mapsto [0..2n-1]$ the following
assertions hold.

\noindent (a)~For all $T \subs S$ we have 
$W \subs \almost_1(U,M, \ParityCond \cup \diam T)$, where $W$ is defined 
as follows:
\beq 
\nonumber
\nu Y_n.  \mu X_n. \nu Y_{n-1}. \mu X_{n-1}. \cdots \nu Y_1. \mu X_1. \nu Y_0. 
\left[
\begin{array}{c}
T \\
\cup \\
B_{2n-1} \cap \lpreodd_1(0,Y_n,X_n) \\
\cup \\ 
B_{2n-2} \cap \lpreeven_1(0,Y_n,X_n,Y_{n-1}) \\
\cup \\
B_{2n-3} \cap \lpreodd_1(1,Y_n,X_n,Y_{n-1},X_{n-1}) \\
\cup \\
B_{2n-4} \cap \lpreeven_1(1,Y_n,X_n,Y_{n-1},X_{n-1},Y_{n-2}) \\
\vdots \\
B_{1} \cap \lpreodd_1(n-1,Y_n,X_n, \ldots,Y_1,X_1) \\
\cup \\
B_0 \cap \lpreeven_1(n-1,Y_n,X_n,\ldots,Y_1,X_1,Y_0)
\end{array}
\right]
\eeq
We refer to the above expression as the \emph{almost-expression} for case~1.
If in the above formula we replace $\lpreodd_1$ by $\lpreodd_2$ and 
$\lpreeven_1$ by $\lpreeven_2$ then we obtain the 
\emph{dual almost-expression} for case~1. 
From the same argument as correctness of the almost-expression and the fact
that counter-strategies for player~2 are against memoryless strategies for 
player~1 we obtain that if the dual almost-expression is $W_D$ for 
$T=\emptyset$, then 
$W_D \subs \set{s \in S \mid \forall \stra_1 \in \bigstra_1^M. 
\exists \stra_2 \in \bigstra_2.\ \Prb_s^{\stra_1,\stra_2}(\coParityCond) =1}$.

(b)~We have 
$Z \subs \no \almost_1(U,M, \ParityCond)$, 
where $Z$ is defined as follows 
\beq 
\nonumber
\mu Y_n.  \nu X_n. \mu Y_{n-1}. \nu X_{n-1}. \cdots \mu Y_1. \nu X_1. \mu Y_0. 
\left[
\begin{array}{c}
B_{2n-1} \cap \fpreodd_2(0,Y_n,X_n) \\
\cup \\ 
B_{2n-2} \cap \fpreeven_2(0,Y_n,X_n,Y_{n-1}) \\
\cup \\
B_{2n-3} \cap \fpreodd_2(1,Y_n,X_n,Y_{n-1},X_{n-1}) \\
\cup \\
B_{2n-4} \cap \fpreeven_2(1,Y_n,X_n,Y_{n-1},X_{n-1},Y_{n-2}) \\
\vdots \\
B_{1} \cap \fpreodd_2(n-1,Y_n,X_n, \ldots,Y_1,X_1) \\
\cup \\
B_0 \cap \fpreeven_2(n-1,Y_n,X_n,\ldots,Y_1,X_1,Y_0)
\end{array}
\right]
\eeq
We refer to the above expression as the \emph{positive-expression} for case~1.

\item \emph{Case~2.} For a parity function $p:S \mapsto [1..2n]$ the following
assertions hold.

\noindent(a)~For all $T \subs S$ we have 
$W \subs \almost_1(U,M, \ParityCond \cup \diam T)$, where $W$ is defined 
as follows:
\beq 
\nonumber
\nu Y_{n-1}. \mu X_{n-1}. \cdots \nu Y_1. \mu X_1. \nu Y_0. \mu X_0 
\left[
\begin{array}{c}
T \\
\cup \\
B_{2n} \cap \pre_1(Y_{n-1}) \\
\cup \\ 
B_{2n-1} \cap \lpreodd_1(0,Y_{n-1},X_{n-1}) \\
\cup \\
B_{2n-2} \cap \lpreeven_1(0,Y_{n-1},X_{n-2},Y_{n-2}) \\
\cup \\
B_{2n-3} \cap \lpreodd_1(1,Y_{n-1},X_{n-1},Y_{n-2},X_{n-2}) \\
\vdots \\
B_{2} \cap \lpreeven_1(n-2,Y_{n-1},X_{n-1}, \ldots,Y_1,X_1,Y_0) \\
\cup \\
B_1 \cap \lpreodd_1(n-1,Y_{n-1},X_{n-1},\ldots,Y_0,X_0)
\end{array}
\right]
\eeq
We refer to the above expression as the almost-expression for case~2.
If in the above formula we replace $\lpreodd_1$ by $\lpreodd_2$ and 
$\lpreeven_1$ by $\lpreeven_2$ then we obtain the 
dual almost-expression for case~2. 
Again, if the dual almost-expression is $W_D$ for 
$T=\emptyset$, then 
$W_D \subs \set{s \in S \mid \forall \stra_1 \in \bigstra_1^M. \exists \stra_2 
\in \bigstra_2.\  \Prb_s^{\stra_1,\stra_2}(\coParityCond) =1}$.

\noindent(b)~We have $Z \subs \no \almost_1(U,M, \ParityCond)$, 
where $Z$ is defined as follows 
\beq 
\nonumber
\mu Y_{n-1}. \nu X_{n-1}. \cdots \mu Y_1. \nu X_1. \mu Y_0. \nu X_0 
\left[
\begin{array}{c}
B_{2n} \cap \epre_2(Y_{n-1}) \\
\cup \\ 
B_{2n-1} \cap \fpreodd_2(0,Y_{n-1},X_{n-1}) \\
\cup \\
B_{2n-2} \cap \fpreeven_2(0,Y_{n-1},X_{n-2},Y_{n-2}) \\
\cup \\
B_{2n-3} \cap \fpreodd_2(1,Y_{n-1},X_{n-1},Y_{n-2},X_{n-2}) \\
\vdots \\
B_{2} \cap \fpreeven_2(n-2,Y_{n-1},X_{n-1}, \ldots,Y_1,X_1,Y_0) \\
\cup \\
B_1 \cap \fpreodd_2(n-1,Y_{n-1},X_{n-1},\ldots,Y_0,X_0)
\end{array}
\right]
\eeq
We refer to the above expression as the positive-expression for case~2.

\end{enumerate}

\smallskip\noindent{\bf The comparison to Emerson-Jutla $\mu$-calculus formula for turn-based games.}
We compare our $\mu$-calculus formula with the $\mu$-calculus formula of Emerson-Jutla~\cite{EJ91} 
to give an intuitive idea of the construction of the formula. 
We first present the formula for Case~2 and then for Case~1.

\smallskip\noindent{\em Case~2.}
For turn-based deterministic games with parity 
function $p:S \to [1..2n]$, it follows from the results of 
Emerson-Jutla~\cite{EJ91}, that the sure-winning (that is equivalent
to the almost-sure winning) set for the objective $\ParityCond \cup \diam T$ is given by the following 
$\mu$-calculus formula:
\beq 
\nonumber
\nu Y_{n-1}. \mu X_{n-1}. \cdots \nu Y_1. \mu X_1. \nu Y_0. \mu X_0 
\left[
\begin{array}{c}
T \\
\cup \\
B_{2n} \cap \pre_1(Y_{n-1}) \\
\cup \\ 
B_{2n-1} \cap \pre_1(X_{n-1}) \\
\cup \\
B_{2n-2} \cap \pre_1(Y_{n-2}) \\
\cup \\
B_{2n-3} \cap \pre_1(X_{n-2}) \\
\vdots \\
B_{2} \cap \pre_1(Y_0) \\
\cup \\
B_1 \cap \pre_1(X_0)
\end{array}
\right]
\eeq
The formula for the almost-expression for case~2 is similar to the 
above $\mu$-calculus formula and is obtained by replacing the 
$\pre_1$ operators with appropriate $\lpreodd_1$ and $\lpreeven_1$ 
operators.

\smallskip\noindent{\em Case~1.}
For turn-based deterministic games with parity 
function $p:S \to [0..2n-1]$, it follows from the results of 
Emerson-Jutla~\cite{EJ91}, that the sure-winning (that is 
equivalent to the almost-sure winning) set 
for the objective $\ParityCond \cup \diam T$ is given by the following 
$\mu$-calculus formula:
\beq 
\nonumber
\mu X_n. \nu Y_{n-1}. \mu X_{n-1}. \cdots \nu Y_1. \mu X_1. \nu Y_0. 
\left[
\begin{array}{c}
T \\
\cup \\
B_{2n-1} \cap \pre_1(X_n) \\
\cup \\ 
B_{2n-2} \cap \pre_1(Y_{n-1}) \\
\cup \\
B_{2n-3} \cap \pre_1(X_{n-1}) \\
\cup \\
B_{2n-4} \cap \pre_1(Y_{n-2}) \\
\vdots \\
B_{1} \cap \pre_1(X_1) \\
\cup \\
B_0 \cap \pre_1(Y_0)
\end{array}
\right]
\eeq
The formula for the almost-expression for case~1 is similar to the 
above $\mu$-calculus formula and is obtained by (a)~adding one 
quantifier alternation $\nu Y_n$; and (b)~replacing the 
$\pre_1$ operators with appropriate $\lpreodd_1$ and $\lpreeven_1$ 
operators.

\medskip\noindent{\bf Proof structure.}
The base case follows from the coB\"uchi and B\"uchi case: it 
follows from the results of~\cite{dAH00} since for B\"uchi and 
coB\"uchi objectives, uniform memoryless almost-sure winning strategies 
exist and our $\mu$-calculus formula coincide with the $\mu$-calculus 
formula to describe the almost-sure winning set for B\"uchi and 
coB\"uchi objectives.
The proof of induction proceeds in four steps as follows:
\begin{enumerate}
\item \emph{Step~1.} We assume the correctness of case~1 and case~2, 
and then extend the result to parity objective with parity function 
$p:S \mapsto [0..2n]$, i.e., we add a max even priority.
The result is obtained as follows: 
for the correctness of the almost-expression we use the correctness of case~1 and 
for complementation we use the correctness of case~2.

\item \emph{Step~2.} We assume the correctness of step~1 and extend the
result to parity objectives with parity function
$p:S \mapsto [1..2n+1]$, i.e., we add a max odd priority.
The result is obtained as follows: 
for the correctness of the almost-expression we use the correctness of case~2 and 
for complementation we use the correctness of step~1.

\item \emph{Step~3.} We assume correctness of step~2 and extend the 
result to parity objectives with parity function 
$p:S \mapsto [1..2n+2]$.
This step  adds a max even priority and the proof will be similar to step~1.
The result is obtained as follows: 
for the correctness of the almost-expression we use the correctness of step~2 and 
for complementation we use the correctness of step~1.

\item \emph{Step~4.} We assume correctness of step~3 and extend the 
result to parity objectives with parity function 
$p:S \mapsto [0..2n+1]$.
This step adds a max odd priority and the proof will be similar to step~2.
The result is obtained as follows: 
for the correctness of the almost-expression we use the correctness of step~1 and 
for complementation we use the correctness of step~3.
\end{enumerate}

We first present two technical lemmas that will be used in the correctness 
proofs.
First we define prefix-independent events.

\medskip\noindent{\bf Prefix-independent events.}
We say that an event or objective is \emph{prefix-independent} if it is independent of 
all finite prefixes. 
Formally, an event or objective $\cala$ is prefix-independent if, for all 
$u, v \in S^*$ and $\omega \in S^\omega$, we have 
$u \omega \in \cala$ iff $v \omega \in \cala$. 
Observe that parity objectives are defined based on the states 
that appear infinitely often along a play, and hence independent of
all finite prefixes, so that, parity objectives are prefix-independent
objectives.

\begin{lem}{(Basic $\apre$ principle).}\label{lemm:basicapre}
Let $X \subs Y \subs Z \subs S$ and $s \in S$ be such that $Y=X \cup \set{s}$
and $s \in \apre_1(Z,X)$.
For all prefix-independent events $\cala \subs \bo (Z \setm Y)$, the following assertion holds:
\begin{quote}
 Assume that there exists a uniform memoryless 
 $\stra_1 \in \bigstra_1^U \cap \bigstra_1^M$ 
 such that for all $\stra_2 \in \bigstra_2$ and for all $z \in Z \setm Y$ 
 we have 
 \[\Prb_z^{\stra_1,\stra_2}(\cala \cup \diam Y) =1.
 \]
 Then there exists a  uniform memoryless 
 $\stra_1 \in \bigstra_1^U \cap \bigstra_1^M$ 
 such that for all $\stra_2 \in \bigstra_2$ we have 
 \[
 \Prb_s^{\stra_1,\stra_2}(\cala \cup \diam X) =1.
 \]
\end{quote}
\end{lem}
\begin{proof} 
Since $s \in \apre_1(Z,X)$, player~1 can play a uniform 
memoryless distribution $\dis_1$ at $s$ to ensure that the 
probability of staying in $Z$ is~1 and with positive probability $\eta>0$ the set 
$X$ is reached.
In $Z \setm Y$ player~1 fixes a uniform memoryless strategy to ensure that 
$\cala \cup \diam Y$ is satisfied with probability~1.
Fix a counter strategy $\stra_2$ for player~2.
If $s$ is visited infinitely often, then since there is a probability of at least
$\eta>0$ to reach $X$, it follows that $X$ is reached with probability~1. 
If $s$ is visited finitely often, then from some point on $\bo (Z \setm Y)$ is satisfied, 
and then $\cala$ is ensured with probability~1. 
Thus the desired result follows.
\qed
\end{proof}

\begin{lem}{(Basic principle of repeated reachability).}
\label{lemm:basic-repeated-reach}
Let $T \subs S$, $B \subs S$ and $W \subs S$ be sets and $\cala$ be 
a prefix-independent objective such that 
\[
W \subs \almost_1(U,M,{\diam T \cup \diam (B \cap \pre_1(W)) \cup \cala}).
\]
Then
\[
W \subs \almost_1(U,M, { \diam T \cup \bo \diam B  \cup \cala}).
\]
\end{lem}
\begin{proof}
Let $Z=B \cap \pre_1(W)$.
For all states  $s \in W\setm (Z\cup T)$, there is a uniform memoryless
player~1 strategy $\stra_1$ that ensures that against 
all player~2 strategies $\stra_2$ we have 
\[
\Prb_s^{\stra_1,\stra_2}\big(\diam (T \cup Z) \cup \cala\big)=1.
\]
For all states in  $Z$ player~1 can ensure that the successor state is in $W$ 
(since $\pre_1(W)$ holds in $Z$).
Consider a strategy $\stra_1^*$ as follows:
for states $s \in Z$ play a uniform memoryless strategy for $\pre_1(W)$ to ensure
that the next state is in $W$;
for states $s\in W \setm (Z\cup T)$ play the uniform memoryless strategy $\stra_1$.
Let us denote by $\diam_k Z \cup \diam T$ to be the set of paths that visits
$Z$ at least $k$-times or visits $T$ at least once. 
Observe that 
$\lim_{k \to\infty} \big( \diam_k Z \cup \diam T \big)
\subs \bo \diam B \cup \diam T$.
Hence for all $s \in W$ and for all $\stra_2 \in \bigstra_2$ we have 
\[
\begin{array}{rcl}
\Prb_s^{\stra_1^*,\stra_2}(\bo \diam B \cup \diam T \cup \cala)
& \geq &
\displaystyle 
\Prb_s^{\stra_1^*,\stra_2}
\big(\diam Z \cup \diam T \cup \cala\big) 
\cdot  \prod_{k=1}^\infty \Prb_s^{\stra_1^*,\stra_2}\big(\diam_{k+1} Z \cup \diam T \cup \cala \mid 
\diam_k Z \cup \diam T \cup \cala\big) \\[1ex]
& = & 1.
\end{array}
\]
The desired result follows.
\qed
\end{proof}

\medskip\noindent{\bf Correctness of step~1.} We now proceed with the
proof of step~1 and by inductive hypothesis we will assume that case~1 and
case~2 hold.

\begin{lem}{}\label{lemm:step-1-limit1}
For a parity function $p:S \mapsto [0..2n]$, and for all $T \subs S$, we have 
$W \subs \almost_1(U,M,{\ParityCond \cup \diam T})$, where $W$ is defined 
as follows:
\beq 
\nonumber
\nu Y_n.  \mu X_n. \nu Y_{n-1}. \mu X_{n-1}. \cdots \nu Y_1. \mu X_1. \nu Y_0. 
\left[
\begin{array}{c}
T \\
\cup \\
B_{2n} \cap \pre_1(Y_n) \\
\cup \\
B_{2n-1} \cap \lpreodd_1(0,Y_n,X_n) \\
\cup \\ 
B_{2n-2} \cap \lpreeven_1(0,Y_n,X_n,Y_{n-1}) \\
\cup \\
B_{2n-3} \cap \lpreodd_1(1,Y_n,X_n,Y_{n-1},X_{n-1}) \\
\cup \\
B_{2n-4} \cap \lpreeven_1(1,Y_n,X_n,Y_{n-1},X_{n-1},Y_{n-2}) \\
\vdots \\
B_{1} \cap \lpreodd_1(n-1,Y_n,X_n, \ldots,Y_1,X_1) \\
\cup \\
B_0 \cap \lpreeven_1(n-1,Y_n,X_n,\ldots,Y_1,X_1,Y_0)
\end{array}
\right]
\eeq
\end{lem}
\begin{proof} 
We first present the intuitive explanation of obtaining the $\mu$-calculus 
formula.

\medskip\noindent{\em Intuitive explanation of the $\mu$-calculus formula.} 
The $\mu$-calculus formula of the lemma is obtained from the almost-expression 
for case~1 by just adding the expression $B_{2n} \cap \pre_1(Y_n)$.

To prove the result we first rewrite $W$ as follows:
\beq 
\nonumber
\nu Y_n.  \mu X_n. \nu Y_{n-1} \mu X_{n-1} \cdots \nu Y_1. \mu X_1. \nu Y_0. 
\left[
\begin{array}{c}
T \cup (B_{2n} \cap \pre_1(W)) \\
\cup \\
B_{2n-1} \cap \lpreodd_1(0,Y_n,X_n) \\
\cup \\ 
B_{2n-2} \cap \lpreeven_1(0,Y_n,X_n,Y_{n-1}) \\
\cup \\
B_{2n-3} \cap \lpreodd_1(1,Y_n,X_n,Y_{n-1},X_{n-1}) \\
\cup \\
B_{2n-4} \cap \lpreeven_1(1,Y_n,X_n,Y_{n-1},X_{n-1},Y_{n-2}) \\
\vdots \\
B_{1} \cap \lpreodd_1(n-1,Y_n,X_n, \ldots,Y_1,X_1) \\
\cup \\
B_0 \cap \lpreeven_1(n-1,Y_n,X_n,\ldots,Y_1,X_1,Y_0)
\end{array}
\right]
\eeq
The rewriting is obtained as follows: since $W$ is the fixpoint $Y_n$, 
we replace $Y_n$ in the $B_{2n} \cap \pre_1(Y_n)$ by $W$.
Treating $T \cup (B_{2n} \cap \pre_1(W))$, as the set $T$ for the almost-expression
for case~1, we obtain from the inductive hypothesis that 
\[
W \subs \almost_1(U,M, \ParityCond \cup \diam(T \cup (B_{2n} \cap \pre_1(W))) ).
\]
By Lemma~\ref{lemm:basic-repeated-reach}, with $B=B_{2n}$ and $\cala=\ParityCond$ we 
obtain that 
\[
W \subs \almost_1(U,M, {\ParityCond \cup \diam T \cup \bo \diam B_{2n} }).
\]
Since $B_{2n}$ is the maximal priority and it is even we have 
$\bo \diam B_{2n} \subs \ParityCond$.
Hence $W \subs \almost_1(U,M,{\ParityCond \cup \diam T })$ and the result follows.
\qed
\end{proof}

\begin{lem}{}\label{lemm:step-1-limit2}
For a parity function $p:S \mapsto [0..2n]$, 
we have $Z \subs \no \almost_1(U,M, \ParityCond)$, where $Z$ is defined as follows 
\beq 
\nonumber
\mu Y_n.  \nu X_n. \mu Y_{n-1}. \nu X_{n-1}. \cdots \mu Y_1. \nu X_1. \mu Y_0. 
\left[
\begin{array}{c}
B_{2n} \cap \epre_2(Y_n) \\
\cup \\
B_{2n-1} \cap \fpreodd_2(0,Y_n,X_n) \\
\cup \\ 
B_{2n-2} \cap \fpreeven_2(0,Y_n,X_n,Y_{n-1}) \\
\cup \\
B_{2n-3} \cap \fpreodd_2(1,Y_n,X_n,Y_{n-1},X_{n-1}) \\
\cup \\
B_{2n-4} \cap \fpreeven_2(1,Y_n,X_n,Y_{n-1},X_{n-1},Y_{n-2}) \\
\vdots \\
B_{1} \cap \fpreodd_2(n-1,Y_n,X_n, \ldots,Y_1,X_1) \\
\cup \\
B_0 \cap \fpreeven_2(n-1,Y_n,X_n,\ldots,Y_1,X_1,Y_0)
\end{array}
\right]
\eeq
\end{lem}
\begin{proof} 
For $k \geq 0$, let $Z_k$ be the set of states of level $k$ in the above 
$\mu$-calculus expression.
We will show that in $Z_k$ for every memoryless strategy for player~1, 
player~2 can ensure that either $Z_{k-1}$ is reached with positive probability or 
else $\coParityCond$ is satisfied with probability~1.
Since $Z_0=\emptyset$, it would follow by induction that 
$Z_k \cap \almost_1(U,M,\ParityCond)=\emptyset$ and the desired result 
will follow.

 We simplify the computation of $Z_k$ given $Z_{k-1}$ and allow that 
$Z_k$ is obtained from $Z_{k-1}$ in the following two ways.
\begin{enumerate} 
\item Add a set states satisfying $B_{2n} \cap \epre_2(Z_{k-1})$, and if such a 
non-emptyset is added, then clearly against any memoryless stratgy for player~1, 
player~2 can ensure from $Z_k$ that $Z_{k-1}$ is reached with positive probability.
Thus the inductive case follows.
\item Add a set of states satisfying the following condition:
\beq 
\nonumber
\nu X_n. \mu Y_{n-1}. \nu X_{n-1}. \cdots \mu Y_1. \nu X_1. \mu Y_0. 
\left[
\begin{array}{c}
B_{2n-1} \cap \fpreodd_2(0,Z_{k-1},X_n) \\
\cup \\ 
B_{2n-2} \cap \fpreeven_2(0,Z_{k-1},X_n,Y_{n-1}) \\
\cup \\
B_{2n-3} \cap \fpreodd_2(1,Z_{k-1},X_n,Y_{n-1},X_{n-1}) \\
\cup \\
B_{2n-4} \cap \fpreeven_2(1,Z_{k-1},X_n,Y_{n-1},X_{n-1},Y_{n-2}) \\
\vdots \\
B_{1} \cap \fpreodd_2(n-1,Z_{k-1},X_n, \ldots,Y_1,X_1) \\
\cup \\
B_0 \cap \fpreeven_2(n-1,Z_{k-1},X_n,\ldots,Y_1,X_1,Y_0)
\end{array}
\right]
\eeq
If the probability of reaching to $Z_{k-1}$ is not positive, then the 
following conditions hold:
\begin{itemize}
\item If the probability to reach $Z_{k-1}$ is not positive, then the
predicate $\epre_2(Z_{k-1})$ vanishes from the predecessor operator 
$\fpreodd_2(i,Z_{k-1},X_{n},Y_{n-1},\ldots,Y_{n-i},X_{n-i})$, and thus 
the operator simplifies to the simpler predecessor operator 
$\lpreeven_2(i-1,X_{n},Y_{n-1},\ldots,Y_{n-i},X_{n-i})$.
\item If the probability to reach $Z_{k-1}$ is not positive, then the 
$\epre_2(Z_{k-1})$ vanishes from the predecessor operator 
$\fpreeven_2(i,Z_{k-1},X_{n},Y_{n-1},\ldots,Y_{n-i},X_{n-i},Y_{n-i-1})$, and 
thus the operator simplifies to the predecessor operator
$\lpreodd_2(i,X_{n},Y_{n-1},\ldots,Y_{n-i},X_{n-i},Y_{n-i-1})$. 
\end{itemize}
Hence either the probability to reach $Z_{k-1}$ is positive, or if the probability to reach $Z_{k-1}$ is not
positive, then the above $\mu$-calculus expression simplifies to
\beq 
\nonumber
Z^*= \nu X_n. \mu Y_{m-1} \nu X_{m-1} \cdots \mu Y_1. \nu X_1. \mu Y_0. 
\left[
\begin{array}{c}
B_{2n-1} \cap \pre_2(X_n) \\
\cup \\ 
B_{2n-2} \cap \lpreodd_2(0,X_n,Y_{n-1}) \\
\cup \\
B_{2n-3} \cap \lpreeven_2(1,X_n,Y_{n-1},X_{n-1}) \\
\cup \\
B_{2n-4} \cap \lpreodd_2(1,X_n,Y_{n-1},X_{n-1},Y_{n-2}) \\
\vdots \\
B_{1} \cap \lpreeven_2(n-2,X_n, \ldots,Y_1,X_1) \\
\cup \\
B_0 \cap \lpreodd_2(n-1,X_n,\ldots,Y_1,X_1,Y_0)
\end{array}
\right].
\eeq
We now consider the parity function $p+1:S\mapsto[1 .. 2n]$, and 
observe that the above formula is same as the dual almost-expression for 
case~2. 
By inductive hypothesis on the dual almost-expression 
we have $Z^* \subs 
\set{s \in S \mid \forall \stra_1 \in \bigstra_1^M. \exists \stra_2 
\in \bigstra_2. \Pr_s^{\stra_1,\stra_2} 
({\coParityCond})=1}$
(since $\Parity(p+1)=\coParityCond$).
Hence the desired claim follows.
\end{enumerate}
The result follows from the above case analysis.
\qed
\end{proof}

\medskip\noindent{\bf Correctness of step~2.} We now prove correctness of
step~2 and we will rely on the correctness of step~1 and the inductive hypothesis.
Since correctness of step~1 follows from the inductive hypothesis, we obtain 
the correctness of step~2 from the inductive hypothesis.

\begin{lem}{}\label{lemm:step-2-limit1}
For a parity function $p:S \mapsto [1..2n+1]$, and for all $T \subs S$ we have 
$W \subs \almost_1(U,M,{\ParityCond \cup \diam T})$, where $W$ is defined 
as follows:
\beq 
\nonumber
\nu Y_n. \mu X_n. 
\nu Y_{n-1}. \mu X_{n-1}. \cdots  \nu Y_0. \mu X_0 
\left[
\begin{array}{c}
T \\
\cup \\
B_{2n+1} \cap \lpreodd_1(0,Y_n,X_n) \\
\cup \\
B_{2n} \cap \lpreeven_1(0,Y_n,X_n,Y_{n-1}) \\
\cup \\ 
B_{2n-1} \cap \lpreodd_1(1,Y_n,X_n,Y_{n-1},X_{n-1}) \\
\cup \\
B_{2n-2} \cap \lpreeven_1(1,Y_n,X_n,Y_{n-1},X_{n-2},Y_{n-2}) \\
\cup \\
B_{2n-3} \cap \lpreodd_1(2,Y_n,X_n,Y_{n-1},X_{n-1},Y_{n-2},X_{n-2}) \\
\vdots \\
B_{2} \cap \lpreeven_1(n-1,Y_n,X_n,Y_{n-1},X_{n-1}, \ldots,Y_1,X_1,Y_0) \\
\cup \\
B_1 \cap \lpreodd_1(n,Y_n,X_n,Y_{n-1},X_{n-1},\ldots,Y_0,X_0)
\end{array}
\right]
\eeq
\end{lem}
\begin{proof}
We first present an intuitive explanation about the how the $\mu$-calculus formula
is obtained.

\medskip\noindent{\em Intuitive explanation of the $\mu$-calculus formula.} 
The $\mu$-calculus expression is obtained from the almost-expression for case~2: 
we add a $\nu Y_n. \mu X_n$ (adding a quantifier alternation of the $\mu$-calculus
formula), and every $\lpreodd$ and $\lpreeven$ predecessor operators are modified by adding 
$\apre_1(Y_n,X_n)  \sdcup$ with the  respective predecessor operators, and we add 
$B_{2n+1} \cap \lpreodd_1(0,Y_n,X_n)$.

We first reformulate the algorithm for computing $W$ in an equivalent form. 
\beq 
\nonumber
W= \mu X_n. \nu Y_{n-1}. \mu X_{n-1}. \cdots \nu Y_0. \mu X_0 
\left[
\begin{array}{c}
T \\
\cup \\
B_{2n+1} \cap \lpreodd_1(0,W,X_n) \\
\cup \\
B_{2n} \cap \lpreeven_1(0,W,X_n,Y_{n-1}) \\
\cup \\ 
B_{2n-1} \cap \lpreodd_1(1,W,X_n,Y_{n-1},X_{n-1}) \\
\cup \\
B_{2n-2} \cap \lpreeven_1(1,W,X_n,Y_{n-1},X_{n-2},Y_{n-2}) \\
\cup \\
B_{2n-3} \cap \lpreodd_1(2,W,X_n,Y_{n-1},X_{n-1},Y_{n-2},X_{n-2}) \\
\vdots \\
B_{2} \cap \lpreeven_1(n-1,W,X_n,Y_{n-1},X_{n-1}, \ldots,Y_1,X_1,Y_0) \\
\cup \\
B_1 \cap \lpreodd_1(n,W,X_n,Y_{n-1},X_{n-1},\ldots,Y_0,X_0)
\end{array}
\right]. 
\eeq 
The reformulation is obtained as follows: since $W$ is the fixpoint of 
$Y_n$ we replace $Y_n$ by $W$ everywhere in the $\mu$-calculus formula.
The above mu-calculus formula is a least fixpoint and thus computes $W$ 
as the limit of a sequence of sets $W_0 = T$, $W_1$, $W_2$, \ldots. 
At each iteration, both states in $B_{2n+1}$ and states satisfying $B_{\leq 2n}$ 
can be added. 
The fact that both types of states can be added complicates the
analysis of the algorithm. 
To simplify the correctness proof, we formulate an alternative
algorithm for the computation of $W$;
an iteration will add either a single $B_{2n+1}$ state, or a set of 
$B_{\leq 2n}$ states. 

To obtain the simpler algorithm, notice that the set of variables 
$Y_{n-1},X_{n-1},\ldots,Y_0,X_0$ does not appear as an argument of the 
$\lpreodd_1(0,W,X_n)=\apre_1(W,X_n)$ operator. 
Hence, each $B_{2n+1}$-state can be added without regards to 
$B_{\leq 2n}$-states that are not already in $W$. 
Moreover, since the $\nu Y_{n-1}. \mu X_{n-1}.  \ldots \nu Y_0. \mu X_0$ 
operator applies only to $B_{\leq 2n}$-states, $B_{2n+1}$-states can be added one at a time. 
From these considerations, we can reformulate the algorithm for the
computation of $W$ as follows. 

The algorithm computes $W$ as an increasing sequence 
$T= T_0 \subset T_1 \subset T_2 \subset \cdots \subset T_m = W$ 
of states, where $m \geq 0$. 
Let $L_i = T_i \setm T_{i-1}$ and 
the sequence is computed by computing $T_i$ as follows, for $0 < i \leq m$: 
\begin{enumerate} 

\item either the set $L_i=\set{s}$ is a singleton such that 
$s \in \apre_1(W, T_{i-1}) \cap B_{2n+1}$. 

\item or the set $L_i$ consists of states in $B_{\leq 2n}$ such that 
$L_i$ is a subset of the following expression 
\beq
\nonumber
\nu Y_{n-1}. \mu X_{n-1}. \cdots \nu Y_0. \mu X_0 
\left[
\begin{array}{c}
B_{2n} \cap \lpreeven_1(0,W,T_{i-1},Y_{n-1}) \\
\cup \\ 
B_{2n-1} \cap \lpreodd_1(1,W,T_{i-1},Y_{n-1},X_{n-1}) \\
\cup \\
B_{2n-2} \cap \lpreeven_1(1,W,T_{i-1},Y_{n-1},X_{n-2},Y_{n-2}) \\
\cup \\
B_{2n-3} \cap \lpreodd_1(2,W,T_{i-1},Y_{n-1},X_{n-1},Y_{n-2},X_{n-2}) \\
\vdots \\
B_{2} \cap \lpreeven_1(n-1,W,T_{i-1},Y_{n-1},X_{n-1}, \ldots,Y_1,X_1,Y_0) \\
\cup \\
B_1 \cap \lpreodd_1(n,W,T_{i-1},Y_{n-1},X_{n-1},\ldots,Y_0,X_0)
\end{array}
\right] 
\eeq
\end{enumerate}
The proof that $W \subs \almost_1(U,M, {\ParityCond \cup \diam T})$ 
is based on an induction on the sequence 
$T = T_0 \subset T_1 \subset T_2 \subset \cdots \subset T_m = W$.
For $1 \leq i \leq m$, let $V_i = W \setm T_{m-i}$, so that
$V_1$ consists of the last block of states that has been added, $V_2$
to the two last blocks, and so on until $V_m = W$. 
We prove by induction on $i \in \{1, \ldots, m\}$, from $i=1$ to $i=m$, 
that for all $s \in V_i$, 
there exists a uniform memoryless strategy $\stra_1$ for player~1 such that 
for all $\stra_2 \in \bigstra_2$ we have
\[
\Prb_s^{\stra_1,\stra_2}\big(\diam T_{m-i} \cup \ParityCond \big) =1. 
\]

Since the base case is a simplified version of the induction step, we
focus on the latter.
There are two cases, depending on whether $V_i \setm V_{i-1}$
is composed of $B_{2n+1}$ or of $B_{\leq 2n}$-states.
Also it will follow from Lemma~\ref{lemm:dual} that there always exists 
uniform distribution to witness that a state satisfy the required 
predecessor operator. 
\begin{itemize}

\item If $V_i \setm V_{i-1} \subs B_{2n+1}$, then $V_i \setm V_{i-1} = \set{s}$
for some $s \in S$ and $s \in \apre_1(W,T_{m-i})$.
The result then follows from the application of the basic $\apre$ principle 
(Lemma~\ref{lemm:basicapre}) with 
$Z=W$, $X =T_{m-i}$, $Z\setm Y= V_{i-1}$ and $\cala= \ParityCond$.

\item If $V_i \setm V_{i-1} \subs B_{\leq 2n}$, then we analyze the 
predecessor operator that $s\in V_i \setm V_{i-1}$ satisfies.
The predecessor operator are essentially the predecessor operator of 
the almost-expression for case~2 modified by the addition of the operator
$\apre_1(W,T_{m-i}) \sdcup$.
If player~2 plays such the $\apre_1(W,T_{m-i})$ part of the predecessor operator
gets satisfied, then the analysis reduces to the previous case, and 
player~1 can ensure that $T_{m-i}$ is reached with probability~1.
Once we rule out the possibility of $\apre_1(W,T_{m-i})$, then the 
$\mu$-calculus expression simplifies to the almost-expression of case~2, i.e.,
\beq
\nonumber
\nu Y_{n-1}. \mu X_{n-1}. \cdots \nu Y_0. \mu X_0 
\left[
\begin{array}{c}
B_{2n} \cap \pre_1(Y_{n-1}) \\
\cup \\ 
B_{2n-1} \cap \lpreodd_1(0,Y_{n-1},X_{n-1}) \\
\cup \\
B_{2n-2} \cap \lpreeven_1(0,Y_{n-1},X_{n-2},Y_{n-2}) \\
\cup \\
B_{2n-3} \cap \lpreodd_1(1,Y_{n-1},X_{n-1},Y_{n-2},X_{n-2}) \\
\vdots \\
B_{2} \cap \lpreeven_1(n-2,Y_{n-1},X_{n-1}, \ldots,Y_1,X_1,Y_0) \\
\cup \\
B_1 \cap \lpreodd_1(n-1,Y_{n-1},X_{n-1},\ldots,Y_0,X_0)
\end{array}
\right] 
\eeq
This ensures that if we rule out 
$\apre_1(W,T_{m-i})$ form the predecessor operators, then by inductive
hypothesis (almost-expression for case~2) we have 
$L_i \subs \almost_1(U,M,{\ParityCond})$, and if $\apre_1(W,T_{m-i})$ is satisfied
then $T_{m-i}$ is ensured to reach with probability~1.
Hence player~1 can ensure that 
for all $s \in V_i$, there
is a uniform memoryless strategy $\stra_1$ for player~1 such that for all 
$\stra_2$ for player~2 we have  
\[
\Prb_s^{\stra_1,\stra_2}\big(\diam T_{m-i} \cup \ParityCond \big) = 1. 
\]
\end{itemize}
This completes the inductive proof.
With $i=m$ we obtain that 
there exists a uniform memoryless strategy 
$\stra_1$ such that for all states $s \in V_m=W$ and for all $\stra_2$ 
we have 
$\Prb_s^{\stra_1,\stra_2}(\diam T_0 \cup \ParityCond) =1$.
Since $T_0=T$, the desired result follows.
\qed
\end{proof}

\begin{lem}{}\label{lemm:step-2-limit2}
For a parity function $p:S \mapsto [1..2n+1]$ we have 
$Z \subs \no \almost_1(U,M,{\ParityCond})$, where $Z$ is defined 
as follows:
\beq 
\nonumber
\mu Y_n. \nu X_n. 
\mu Y_{n-1}. \nu X_{n-1}. \cdots \mu Y_0. \nu X_0 
\left[
\begin{array}{c}
B_{2n+1} \cap \fpreodd_2(0,Y_n,X_n) \\
\cup \\
B_{2n} \cap \fpreeven_2(0,Y_n,X_n,Y_{n-1}) \\
\cup \\ 
B_{2n-1} \cap \fpreodd_2(1,Y_n,X_n,Y_{n-1},X_{n-1}) \\
\cup \\
B_{2n-2} \cap \fpreeven_2(1,Y_n,X_n,Y_{n-1},X_{n-2},Y_{n-2}) \\
\cup \\
B_{2n-3} \cap \fpreodd_2(2,Y_n,X_n,Y_{n-1},X_{n-1},Y_{n-2},X_{n-2}) \\
\vdots \\
B_{2} \cap \fpreeven_2(n-1,Y_n,X_n,Y_{n-1},X_{n-1}, \ldots,Y_1,X_1,Y_0) \\
\cup \\
B_1 \cap \fpreodd_2(n,Y_n,X_n,Y_{n-1},X_{n-1},\ldots,Y_0,X_0)
\end{array}
\right]
\eeq
\end{lem}
\begin{proof}
For $k \geq 0$, let $Z_k$ be the set of states of level $k$ in the above 
$\mu$-calculus expression.
We will show that in $Z_k$ player~2 can ensure that either $Z_{k-1}$ 
is reached with positive probability or else $\coParityCond$ is satisfied
with probability~1.
Since $Z_0=\emptyset$, it would follow by induction that 
$Z_k \cap \almost_1(U,M,{\ParityCond})=\emptyset$ and the desired result 
will follow.

 We obtain of $Z_k$ from $Z_{k-1}$ as follows:
\beq 
\nonumber
\nu X_n. \mu Y_{n-1}. \nu X_{n-1}. \cdots \mu Y_0. \nu X_0 
\left[
\begin{array}{c}
B_{2n+1} \cap \fpreodd_2(0,Z_{k-1},X_n) \\
\cup \\
B_{2n} \cap \fpreeven_2(0,Z_{k-1},X_n,Y_{n-1}) \\
\cup \\ 
B_{2n-1} \cap \fpreodd_2(1,Z_{k-1},X_n,Y_{n-1},X_{n-1}) \\
\cup \\
B_{2n-2} \cap \fpreeven_2(1,Z_{k-1},X_n,Y_{n-1},X_{n-2},Y_{n-2}) \\
\cup \\
B_{2n-3} \cap \fpreodd_2(2,Z_{k-1},X_n,Y_{n-1},X_{n-1},Y_{n-2},X_{n-2}) \\
\vdots \\
B_{2} \cap \fpreeven_2(n-1,Z_{k-1},X_n,Y_{n-1},X_{n-1}, \ldots,Y_1,X_1,Y_0) \\
\cup \\
B_1 \cap \fpreodd_2(n,Z_{k-1},X_n,Y_{n-1},X_{n-1},\ldots,Y_0,X_0)
\end{array}
\right]
\eeq
If player~1 risks into moving to $Z_{k-1}$ with positive probability, then 
the inductive case is proved as $Z_{k-1}$ is reached with positive probability.
If the probability of reaching to $Z_{k-1}$ is not positive, then the 
following conditions hold:
\begin{itemize}
\item If the probability to reach $Z_{k-1}$ is not positive, then the
predicate $\epre_2(Z_{k-1})$ vanishes from the predecessor operator 
$\fpreodd_2(i,Z_{k-1},X_{n},Y_{n-1},\ldots,Y_{n-i},X_{n-i})$, and thus 
the operator simplifies to the simpler predecessor operator 
$\lpreeven_2(i-1,X_{n},Y_{n-1},\ldots,Y_{n-i},X_{n-i})$.
\item If the probability to reach $Z_{k-1}$ is not positive, then the 
$\epre_2(Z_{k-1})$ vanishes from the predecessor operator 
$\fpreeven_2(i,Z_{k-1},X_{n},Y_{n-1},\ldots,Y_{n-i},X_{n-i},Y_{n-i-1})$, and 
thus the operator simplifies to the predecessor operator
$\lpreodd_2(i,X_{n},Y_{n-1},\ldots,Y_{n-i},X_{n-i},Y_{n-i-1})$. 
\end{itemize}
Hence either the probability to reach $Z_{k-1}$ is positive, or the probability to reach $Z_{k-1}$ is not positive,
then the above $\mu$-calculus expression simplifies to
\beq 
\nonumber
Z^*= 
\nu X_n. \mu Y_{n-1}. \nu X_{n-1}. \cdots  \mu Y_0. \nu X_0 
\left[
\begin{array}{c}
B_{2n+1} \cap \pre_2(X_n) \\
\cup \\
B_{2n} \cap \lpreodd_2(0,X_n,Y_{n-1}) \\
\cup \\ 
B_{2n-1} \cap \lpreeven_2(0,X_n,Y_{n-1},X_{n-1}) \\
\cup \\
B_{2n-2} \cap \lpreodd_2(1,X_n,Y_{n-1},X_{n-2},Y_{n-2}) \\
\cup \\
B_{2n-3} \cap \lpreeven_2(1,X_n,Y_{n-1},X_{n-1},Y_{n-2},X_{n-2}) \\
\vdots \\
B_{2} \cap \lpreodd_2(n-1,X_n,Y_{n-1},X_{n-1}, \ldots,Y_1,X_1,Y_0) \\
\cup \\
B_1 \cap \lpreeven_2(n-1,X_n,Y_{n-1},X_{n-1},\ldots,Y_0,X_0)
\end{array}
\right].
\eeq
We now consider the parity function $p-1:S \mapsto [0 .. 2n]$ 
and by the correctness of the dual almost-expression for step~1  
(Lemma~\ref{lemm:step-1-limit1}) (with the roles of player~1 and player~2 
exchanged and player~2 plays against memoryless strategies for player~1, as in 
Lemma~\ref{lemm:step-1-limit2}) 
we have $Z^* \subs \set{s \in S \mid \forall \stra_1 \in \bigstra_1^M. \exists \stra_2 \in \bigstra_2.\ \Pr_s^{\stra_1,\stra_2}(\coParityCond)=1}$ 
(since $\coParityCond=\Parity(p-1)$).
Hence the result follows.
\qed
\end{proof}

\medskip\noindent{\bf Correctness of step~3.}
The correctness of step~3 is similar to correctness of step~1. 
Below we present the proof sketches (since they are similar to step~1).

\begin{lem}{}\label{lemm:step-3-limit1}
For a parity function $p:S \mapsto [1..2n+2]$, and for all $T \subs S$, we have 
$W \subs \almost_1(U,M,{\ParityCond \cup \diam T})$, where $W$ is defined 
as follows:
\beq 
\nonumber
\nu Y_n. \mu X_n. 
\nu Y_{n-1}. \mu X_{n-1}. \cdots  \nu Y_0. \mu X_0 
\left[
\begin{array}{c}
T \\
\cup \\
B_{2n+2} \cap \pre_1(Y_n) \\
\cup \\
B_{2n+1} \cap \lpreodd_1(0,Y_n,X_n) \\
\cup \\
B_{2n} \cap \lpreeven_1(0,Y_n,X_n,Y_{n-1}) \\
\cup \\ 
B_{2n-1} \cap \lpreodd_1(1,Y_n,X_n,Y_{n-1},X_{n-1}) \\
\cup \\
B_{2n-2} \cap \lpreeven_1(1,Y_n,X_n,Y_{n-1},X_{n-2},Y_{n-2}) \\
\cup \\
B_{2n-3} \cap \lpreodd_1(2,Y_n,X_n,Y_{n-1},X_{n-1},Y_{n-2},X_{n-2}) \\
\vdots \\
B_{2} \cap \lpreeven_1(n-1,Y_n,X_n,Y_{n-1},X_{n-1}, \ldots,Y_1,X_1,Y_0) \\
\cup \\
B_1 \cap \lpreodd_1(n,Y_n,X_n,Y_{n-1},X_{n-1},\ldots,Y_0,X_0)
\end{array}
\right]
\eeq
\end{lem}
\begin{proof} 
The proof is almost identical to the proof of Lemma~\ref{lemm:step-1-limit1}. 
Similar to step~1 (Lemma~\ref{lemm:step-1-limit1}), we add a max even priority.
The proof of the result is essentially identical to the proof of Lemma~\ref{lemm:step-1-limit1} (almost 
copy-paste of the proof),
the only modification is instead of the correctness of the almost-expression of case~1 
we need to consider the correctness of the almost-expression for step~2 (i.e.,
Lemma~\ref{lemm:step-2-limit1} for parity function $p:S \mapsto [1..2n+1]$).
\qed
\end{proof}

\begin{lem}{}\label{lemm:step-3-limit2}
For a parity function $p:S \mapsto [1..2n+2]$  
we have $Z \subs \no \almost_1(U,M,{\ParityCond})$, where $Z$ is defined 
as follows:
\beq 
\nonumber
\mu Y_n. \nu X_n. 
\mu Y_{n-1}. \nu X_{n-1}. \cdots \mu Y_0. \nu X_0 
\left[
\begin{array}{c}
B_{2n+2} \cap \epre_2(Y_n) \\
\cup \\
B_{2n+1} \cap \fpreodd_2(0,Y_n,X_n) \\
\cup \\
B_{2n} \cap \fpreeven_2(0,Y_n,X_n,Y_{n-1}) \\
\cup \\ 
B_{2n-1} \cap \fpreodd_2(1,Y_n,X_n,Y_{n-1},X_{n-1}) \\
\cup \\
B_{2n-2} \cap \fpreeven_2(1,Y_n,X_n,Y_{n-1},X_{n-2},Y_{n-2}) \\
\cup \\
B_{2n-3} \cap \fpreodd_2(2,Y_n,X_n,Y_{n-1},X_{n-1},Y_{n-2},X_{n-2}) \\
\vdots \\
B_{2} \cap \fpreeven_2(n-1,Y_n,X_n,Y_{n-1},X_{n-1}, \ldots,Y_1,X_1,Y_0) \\
\cup \\
B_1 \cap \fpreodd_2(n,Y_n,X_n,Y_{n-1},X_{n-1},\ldots,Y_0,X_0)
\end{array}
\right]
\eeq
\end{lem}
\begin{proof}
The proof of the result is identical to the proof of Lemma~\ref{lemm:step-1-limit2} (almost copy-paste of
the proof),
the only modification is instead of the correctness of the almost-expression of case~2 
we need to consider the correctness of the almost-expression for step~1 (i.e.,
Lemma~\ref{lemm:step-1-limit1}). 
This is because in the proof, after we rule out states in $B_{2n+2}$ and 
analyze the sub-formula as in Lemma~\ref{lemm:step-1-limit1},  
we consider parity function $p-1:S \mapsto [0 .. 2n]$ and
then invoke the correctness of Lemma~\ref{lemm:step-1-limit1}.
\qed
\end{proof}

\medskip\noindent{\bf Correctness of step~4.}
The correctness of step~4 is similar to correctness of step~2.
Below we present the proof sketches (since they are similar to step~2).

\begin{lem}{}\label{lemm:step-4-limit1}
For a parity function $p:S \mapsto [0..2n+1]$, and for all $T \subs S$, we have 
$W \subs \almost_1(U,M,{\ParityCond \cup \diam T})$, where $W$ is defined 
as follows:
\beq 
\nonumber
\nu Y_{n+1}.  \mu X_{n+1}. \cdots \nu Y_1. \mu X_1. \nu Y_0. 
\left[
\begin{array}{c}
T \\
\cup \\
B_{2n+1} \cap \lpreodd_1(0,Y_{n+1},X_{n+1}) \\
\cup \\
B_{2n} \cap \lpreeven_1(0,Y_{n+1},X_{n+1},Y_n) \\
\cup \\
B_{2n-1} \cap \lpreodd_1(1,Y_{n+1},X_{n+1},Y_n,X_n) \\
\cup \\ 
B_{2n-2} \cap \lpreeven_1(1,Y_{n+1},X_{n+1},Y_n,X_n,Y_{n-1}) \\
\cup \\
B_{2n-3} \cap \lpreodd_1(2,Y_{n+1},X_{n+1},Y_n,X_n,Y_{n-1},X_{n-1}) \\
\cup \\
B_{2n-4} \cap \lpreeven_1(2,Y_{n+1},X_{n+1},Y_n,X_n,Y_{n-1},X_{n-1},Y_{n-2}) \\
\vdots \\
B_{1} \cap \lpreodd_1(n,Y_{n+1},X_{n+1},Y_n,X_n, \ldots,Y_1,X_1) \\
\cup \\
B_0 \cap \lpreeven_1(n,Y_{n+1},X_{n+1},Y_n,X_n,\ldots,Y_1,X_1,Y_0)
\end{array}
\right]
\eeq
\end{lem}
\begin{proof}
Similar to step~2 (Lemma~\ref{lemm:step-2-limit1}), we add a max odd priority.
The proof of the result is essentially identical to the proof of Lemma~\ref{lemm:step-2-limit1}
(almost copy-paste of the proof),
the only modification is instead of the correctness of the almost-expression of case~2 
we need to consider the correctness of the almost-expression for step~1 (i.e.,
Lemma~\ref{lemm:step-1-limit1} for parity function $p:S \mapsto [0..2n]$).
\qed
\end{proof}

\begin{lem}{}\label{lemm:step-4-limit2}
For a parity function $p:S \mapsto [0..2n+1]$  
we have $Z \subs \no \almost_1(U,M,{\ParityCond})$, where $Z$ is defined 
as follows:
\beq 
\nonumber
\mu Y_{n+1}.  \nu X_{n+1}. \cdots \mu Y_1. \nu X_1. \mu Y_0. 
\left[
\begin{array}{c}
B_{2n+1} \cap \fpreodd_2(0,Y_{n+1},X_{n+1}) \\
\cup \\
B_{2n} \cap \fpreeven_2(0,Y_{n+1},X_{n+1},Y_n) \\
\cup \\
B_{2n-1} \cap \fpreodd_2(1,Y_{n+1},X_{n+1},Y_n,X_n) \\
\cup \\ 
B_{2n-2} \cap \fpreeven_2(1,Y_{n+1},X_{n+1},Y_n,X_n,Y_{n-1}) \\
\cup \\
B_{2n-3} \cap \fpreodd_2(2,Y_{n+1},X_{n+1},Y_n,X_n,Y_{n-1},X_{n-1}) \\
\cup \\
B_{2n-4} \cap \fpreeven_2(2,Y_{n+1},X_{n+1},Y_n,X_n,Y_{n-1},X_{n-1},Y_{n-2}) \\
\vdots \\
B_{1} \cap \fpreodd_2(n,Y_{n+1},X_{n+1},Y_n,X_n, \ldots,Y_1,X_1) \\
\cup \\
B_0 \cap \fpreeven_2(n,Y_{n+1},X_{n+1},Y_n,X_n,\ldots,Y_1,X_1,Y_0)
\end{array}
\right]
\eeq
\end{lem}
\begin{proof}
The proof of the result is identical to the proof of Lemma~\ref{lemm:step-2-limit2}
(almost copy-paste of the proof),
the only modification is instead of the correctness of the almost-expression of step~1 
(Lemma~\ref{lemm:step-1-limit1}) we need to consider the correctness of the 
almost-expression for step~3 (i.e., Lemma~\ref{lemm:step-3-limit1}).
This is because in the proof, while we analyze the sub-formula as in Lemma~\ref{lemm:step-3-limit1},  
we consider parity function $p+1:S \mapsto [1 .. 2n+2]$ and
then invoke the correctness of Lemma~\ref{lemm:step-3-limit1}. 
\qed
\end{proof}

Observe that above we presented the correctness for the almost-expressions 
for case~1 and case~2, and the correctness proofs for the dual 
almost-expressions are identical.
We now present the duality of the predecessor operators.
We first present some notations required for the proof.

\smallskip\noindent{\bf Destination or possible successors of 
moves and distributions.} 
Given a state $s$ and distributions $\dis_1 \in \sd^s_1$ and 
$\dis_2 \in \sd^s_2$ we denote by 
$\dest(s,\dis_1,\dis_2)=\set{t \in S \mid \pr_2^{\dis_1,\dis_2}(t) >0}$ 
the set of states that have positive probability of transition from $s$ 
when the players play $\dis_1$ and $\dis_2$ at $s$.
For actions $a$ and $b$ we have $\dest(s,a,b) =
\set{t \in S \mid \trans(s,a,b)(t) >0}$ as the set of possible successors 
given $a$ and $b$.
For $A \subs \mov_1(s)$ and $B \subs \mov_2(s)$ we have 
$\dest(s,A,B)= \bigcup_{a\in A, b \in B} \dest(s,a,b)$.

\begin{lem}{(Duality of predecessor operators).}\label{lemm:dual}
The following assertions hold.
\begin{enumerate}
\item Given $X_n \subs X_{n-1} \subs \cdots \subs X_{n-i} \subs Y_{n-i} \subs
Y_{n-i+1} \subs \cdots \subs Y_n$, 
we have 
\[
\fpreodd_2(i,\no Y_n,\no X_n, \ldots,\no Y_{n-i},\no X_{n-i}) =
\no \lpreodd_1(i,Y_n,X_n, \ldots,Y_{n-i},X_{n-i}).
\]

\item Given $X_n \subs X_{n-1} \subs \cdots \subs X_{n-i} \subs Y_{n-i-1}
\subs Y_{n-i} \subs
Y_{n-i+1} \subs \cdots \subs Y_n$, 
we have
\[
\begin{array}{rcl}
& \fpreeven_2 & (i,\no Y_n,\no X_n, \ldots,\no Y_{n-i},\no X_{n-i},\no Y_{n-i-1}) \\
& = &
\no \lpreeven_1(i,Y_n,X_n, \ldots,Y_{n-i},X_{n-i},Y_{n-i-1}).
\end{array}
\]

\item For all $s \in S$, if 
$s \in \lpreodd_1(i,Y_n,X_n, \ldots,Y_{n-i},X_{n-i})$ 
(resp. $s \in \lpreeven_1(i,Y_n,X_n, \ldots,Y_{n-i},X_{n-i},Y_{n-i-1})$), 
then there exists uniform distribution $\dis_1$ to witness that  
$s \in \lpreodd_1(i,Y_n,X_n, \ldots,Y_{n-i},X_{n-i})$ 
(resp. $s \in \lpreeven_1(i,Y_n,X_n, \ldots,Y_{n-i},X_{n-i},Y_{n-i-1})$). 
\end{enumerate}
\end{lem}
\begin{proof}
We present the proof for part 1, and the proof for second part is analogous.
To present the proof of the part 1, we first present the proof for the 
case when $n=2$ and $i=2$. This proof already has all the ingredients of the 
general proof. After presenting the proof we present the general case.

\smallskip\noindent{\bf Claim.}
We show that for $X_1 \subs X_0 \subs Y_0 \subs Y_1$ we have 
\[
\epre_2(\no Y_1) \sdcup \apre_2(\no X_1, \no Y_0) \sdcup \pre_2(\no X_0) = 
\no (\apre_1(Y_1, X_1) \sdcup  \apre_1(Y_0,X_0)).
\]
We now present the following two case analysis for the proof.
\begin{enumerate}
\item A subset $U \subs \mov_1(s)$ is \emph{good} if both the following 
conditions hold:
\begin{enumerate}
\item \emph{Condition~1.} For all $b \in \mov_2(s)$ and for all $a \in U$ we
have $\dest(s,a,b) \subs Y_1$ (i.e., $\dest(s,U,b) \subs Y_1$); and 
\item \emph{Condition~2.} For all $b \in \mov_2(s)$ one of the following conditions hold:
\begin{enumerate}
\item either there exists $a \in U$ such that $\dest(s,a,b) \cap X_1 \neq \emptyset$ (i.e., $\dest(s,U,b) \cap X_1 \neq \emptyset$); 
or 
\item for all $a \in U$ we have $\dest(s,a,b) \subs Y_0$ (i.e., 
$\dest(s,U,b) \subs Y_0$) and for some $a \in U$ 
we have $\dest(s,a,b) \cap X_0 \neq \emptyset$ (i.e., 
$\dest(a,U,b) \cap X_0 \neq \emptyset$).
\end{enumerate}
\end{enumerate}
We show that if there is a good set $U$, then 
$s \in \apre_1(Y_1,X_1) \sdcup \apre_1(Y_0,X_0)$. 
Given a good set $U$, consider the \emph{uniform} distribution $\dis_1$ 
that plays all actions in $U$ uniformly at random.
Consider an action $b \in \mov_2(s)$ and the following assertions hold:
\begin{enumerate}
\item By condition~1 we have $\dest(s,\dis_1,b) \subs Y_1$.
\item By condition~2 we have either (i)~$\dest(s,\dis_1,b) \cap X_1 \neq \emptyset$
(if condition 2.a holds); 
or (ii)~$\dest(s,\dis_1,b) \subs Y_0$, and $\dest(s,\dis_1,b) \cap X_0 \neq \emptyset$
(if condition 2.b holds).
\end{enumerate}
It follows that in all cases we have 
(i)~either $\dest(s,\dis_1,b) \subs Y_1$ and $\dest(s,\dis_1,b) \cap X_1 
\neq \emptyset$, or 
(ii)~$\dest(s,\dis_1,b) \subs Y_0$ and $\dest(s,\dis_1,b) \cap X_0 
\neq \emptyset$.
It follows that $\dis_1$ is a uniform distribution witness to show that 
$s \in \apre_1(Y_1,X_1) \sdcup \apre_1(Y_0,X_0)$.

\item We now show that if there is no good set $U$, then 
$s \in \epre_2(\no Y_1) \sdcup \apre_2(\no X_1, \no Y_0) \sdcup \pre_2(\no X_0)$. 
Given a set $U$, if $U$ is not good, then (by simple complementation argument) 
one of the following conditions must hold:
\begin{enumerate}
\item \emph{Complementary Condition~1.} There exists $b \in \mov_2(s)$ and 
$a \in U$ such that 
$\dest(s,a,b) \cap \no Y_1 \neq \emptyset$; or 
\item \emph{Complementary Condition~2.} There exists $b \in \mov_2(s)$ such that 
both the following conditions hold:
\begin{enumerate}
\item for all $a \in U$ we have $\dest(s,a,b) \subs \no X_1$; 
and
\item there exists $a \in U$ such that $\dest(s,a,b) \cap \no Y_0 \neq \emptyset$ 
or for all $a \in U$ we have $\dest(s,a,b) \subs \no X_0$.
\end{enumerate}
\end{enumerate}
Since there is no good set, for every set $U \subs \mov_1(s)$, there is a counter action $b=c(U) 
\in \mov_2(s)$, such that one of the complementary conditions hold. 
Consider a distribution $\dis_1$ for player~1, and 
let $U=\supp(\dis_1)$. Since $U$ is not a good set, consider a counter action
$b=c(U)$ satisfying the complementary conditions. 
We now consider the following cases:
\begin{enumerate}
\item If complementary condition 1 holds, then $\dest(s,\dis_1,b) \cap \no Y_1 \neq \emptyset$ 
(i.e., $\epre_2(\no Y_1)$ is satisfied). 
\item Otherwise complementary condition 2 holds, and by 2.a we have 
$\dest(s,\dis_1,b) \subs \no X_1$. 
\begin{enumerate}
\item if 
there exists $a \in U$ such that $\dest(s,a,b) \cap \no Y_0 \neq \emptyset$, 
then $\dest(s,\dis_1,b) \cap \no Y_0 \neq \emptyset$ (hence $\apre_2(\no X_1, \no Y_0)$ holds); 
\item otherwise for all $a \in U$ we have $\dest(s,a,b) \subs \no X_0$, hence 
$\dest(s,\dis_1,b) \subs \no X_0$ (hence $\pre_2(\no X_0)$ holds).
\end{enumerate}
\end{enumerate}
The claim follows.
\end{enumerate}

\smallskip\noindent{\bf General case.} We now present the result for 
the general case which is a generalization of the previous case.
We present the details here, and will omit it in later proofs, where 
the argument is similar.
Recall that we have the following inclusion: 
$X_n \subs X_{n-1} \subs \ldots \subs X_{n-i} \subs Y_{n-i} \subs \ldots Y_{n-1} \subs Y_n$.
\begin{enumerate}
\item A subset $U \subs \mov_1(s)$ is \emph{good} if both the following 
conditions hold: for all $b \in \mov_2(s)$ 
\begin{enumerate}
\item \emph{Condition~1.} 
For all $a \in U$ we have $\dest(s,a,b) \subs Y_n$ (i.e., 
$\dest(s,U,b) \subs Y_n$); and 
\item \emph{Condition~2.} 
There exists $0 \leq j \leq i$, such that 
for all $a \in U$ we have $\dest(s,a,b) \subs Y_{n-j}$ (i.e., 
$\dest(s,U,b) \subs Y_{n-j}$), and 
for some $a \in U$ we have $\dest(s,a,b) \cap X_{n-j} \neq \emptyset$ 
(i.e., $\dest(s,U,b) \cap X_{n-j} \neq \emptyset$). 
\end{enumerate}
We show that if there is a good set $U$, then 
$s \in \apre_1(i,Y_n,X_n,\ldots, Y_{n-i},X_{n-i})$. 
Given a good set $U$, consider the \emph{uniform} distribution $\dis_1$ 
that plays all actions in $U$ uniformly at random.
Consider an action $b \in \mov_2(s)$ and the following assertions hold:
\begin{enumerate}
\item By condition~1 we have $\dest(s,\dis_1,b) \subs Y_{n}$.
\item By condition~2 we have 
for some $0 \leq j \leq i$, we have $\dest(s,\dis_1,b) \subs Y_{n-j}$, 
and $\dest(s,\dis_1,b) \cap X_{n-j} \neq \emptyset$ 
(i.e., $\apre_1(Y_{n-j}, X_{n-j})$ holds).
\end{enumerate}
It follows that $\dis_1$ is a uniform distribution witness to show that 
$s \in \apre_1(Y_{n},X_{n}, \ldots,Y_{n-i},X_{n-i})$.

\item We now show that if there is no good set $U$, then 
$s \in \fpreodd_2(i,\no Y_n, \no X_n, \ldots, \no Y_{n-i},\no X_{n-i})$. 
Given a set $U$, if $U$ is not good, then we show that one of the following conditions 
must hold: there exists $b \in \mov_2(s)$ such that
\begin{enumerate}
\item \emph{Complementary Condition~1 (CC1).} 
$\dest(s,U,b) \cap \no Y_n \neq \emptyset$; or 
\item \emph{Complementary Condition~2 (CC2).} 
there exists $0\leq j <i$ such that 
$\dest(s,U,b) \subs \no X_{n-j}$ and 
$\dest(s,U,b) \cap \no Y_{n-j-1} \neq \emptyset$; or
\item \emph{Complementary Condition~3 (CC3).} 
$\dest(s,U,b) \subs \no X_{n-i}$.
\end{enumerate}
Consider a set $U$ that is not good, and let $b$ be an action that witness that 
$U$ is not good. We show that $b$ satisfies one of the complementary conditions. 
\begin{itemize}
\item If $\dest(s,U,b) \cap \no Y_n \neq \emptyset$, then we are done as CC1 is 
satisfied. 
Otherwise, we have $\dest(s,U,b) \subs Y_n$, then we must have $\dest(s,U,b) 
\subs \no X_n$ (otherwise the action $b$ would satisfy the condition 
$\dest(s,U,b) \subs Y_n$ and $\dest(s,U,b) \cap X_n \neq \emptyset$, and cannot be 
a witness that $U$ is not good).
Now we continue: if $\dest(s,U,b) \cap \no Y_{n-1} \neq \emptyset$, then we are done, as we have a 
witness that $\dest(s,U,b) \subs \no X_n$ and $\dest(s,U,b) \cap \no Y_{n-1} \neq \emptyset$.
If $\dest(s, U, b) \subs Y_{n-1}$, then again since $b$ is witness to show that 
$U$ is not good, we must have $\dest(s,U,b) \subs \no X_{n-1}$. 
We again continue, and if we have 
$\dest(s,U,b) \cap \no Y_{n-2} \neq \emptyset$, we are done, or else we 
continue and so on.
Thus we either find a witness $0\leq j< i$ to satisfy CC2, or else in the end 
we have that $\dest(s,U,b) \subs \no X_{n-i}$ (satisfies CC3).
\end{itemize}
Since there is no good set, for every set $U \subs \mov_1(s)$, there is a counter action $b=c(U) 
\in \mov_2(s)$, such that one of the complementary conditions hold. 
Consider a distribution $\dis_1$ for player~1, and 
let $U=\supp(\dis_1)$. Since $U$ is not a good set, consider a counter action
$b=c(U)$ satisfying the complementary conditions. 
We now consider the following cases:
\begin{enumerate}
\item If CC1 1 holds, then $\dest(s,U,b) \cap \no Y_n \neq \emptyset$ 
(hence also $\dest(s,\dis_1,b) \cap \no Y_n \neq \emptyset$) 
(i.e., $\epre_2(\no Y_n)$ is satisfied). 
\item Else if CC2 holds, then for some $0\leq j <i$ and we have 
$\dest(s,U,b) \subs \no X_{n-j}$ and $\dest(s,U,b) \subs Y_{n-j-1}$
(hence also $\dest(s,\dis_1,b) \subs \no X_{n-j}$ and $\dest(s,\dis_1,b) \subs Y_{n-j-1})$
(i.e.,  $\apre_2(\no X_{n}, \no Y_{n-1}) \sdcup \apre_2(\no X_{n-1},\no Y_{n-2}) \sdcup \ldots 
\sdcup \apre_2(\no X_{n-i+1},\no Y_{n-i})$ holds). 
\item Otherwise CC3 holds and we have $\dest(s,U,b) \subs \no X_{n-i}$, 
(hence also $\dest(s,\dis_1,b) \subs \no X_{n-i}$)
(i.e., $\pre_2(\no X_{n-i})$ holds).
\end{enumerate}
The claim follows.
\end{enumerate}

The result for part~3 follows as in the above proofs we have always
constructed uniform witness distribution. 
\qed
\end{proof}

\medskip\noindent{\bf Characterization of $\almost_1(U,M,\Phi)$ set.}
From Lemmas~\ref{lemm:step-1-limit1}---\ref{lemm:step-4-limit2}, and
the duality of predecessor operators (Lemma~\ref{lemm:dual})
we obtain the following result characterizing the almost-sure winning 
set for uniform memoryless strategies for parity objectives.

\begin{theo}{}\label{theo-uniform-memless}
For all concurrent game structures $\game$ over state space $S$, 
for all parity objectives $\ParityCond$ for player~1, the following assertions hold.
\begin{enumerate}
\item If $p:S \mapsto [0..2n-1]$, then $\almost_1(U,M,\ParityCond)=W$, where 
$W$ is defined as follows
\beq\label{eq-rabin-all-limit-1}
\nu Y_n.  \mu X_n. 
\cdots. \nu Y_1. \mu X_1. \nu Y_0. \\ 
\left[
\begin{array}{c}
B_{2n-1} \cap \lpreodd_1(0,Y_n,X_n) \\
\cup \\ 
B_{2n-2} \cap \lpreeven_1(0,Y_n,X_n,Y_{n-1}) \\
\cup \\
B_{2n-3} \cap \lpreodd_1(1,Y_n,X_n,Y_{n-1},X_{n-1}) \\
\cup \\
B_{2n-4} \cap \lpreeven_1(1,Y_n,X_n,Y_{n-1},X_{n-1},Y_{n-2}) \\
\vdots \\
B_{1} \cap \lpreodd_1(n\!-\!1,Y_n,X_n, \ldots,Y_1,X_1) \\
\cup \\
B_0 \cap \lpreeven_1(n\!-\!1,Y_n,X_n,\ldots,Y_1,X_1,Y_0)
\end{array}
\right]
\eeq
and $B_i=p^{-1}(i)$ is the set of states with priority~$i$, for 
$i\in [0..2n-1]$.

\item If $p:S \mapsto [1..2n]$, then $\almost_1(U,M,\ParityCond)=W$, where
$W$ is defined as follows 
\beq \label{eq-rabin-all-limit-2}
\nu Y_{n-1}. \mu X_{n-1}. \cdots. 
\nu Y_0. \mu X_0 
\left[
\begin{array}{c}
B_{2n} \cap \pre_1(Y_{n-1}) \\
\cup \\ 
B_{2n-1} \cap \lpreodd_1(0,Y_{n-1},X_{n-1}) \\
\cup \\
B_{2n-2} \cap \lpreeven_1(0,Y_{n-1},X_{n-2},Y_{n-2}) \\
\cup \\
B_{2n-3} \cap \lpreodd_1(1,Y_{n-1},X_{n-1},Y_{n-2},X_{n-2}) \\
\vdots \\
B_{2} \cap \lpreeven_1(n-2,Y_{n-1},X_{n-1}, \ldots,Y_1,X_1,Y_0) \\
\cup \\
B_1 \cap \lpreodd_1(n-1,Y_{n-1},X_{n-1},\ldots,Y_0,X_0)
\end{array}
\right]
\eeq
and $B_i=p^{-1}(i)$ is the set of states with priority~$i$, for 
$i\in [1..2n]$.

\item 
The set $\almost_1(U,M,\ParityCond)$ can be computed symbolically 
using the expressions (\ref{eq-rabin-all-limit-1}) and (\ref{eq-rabin-all-limit-2}) 
in time $\bigo(|S|^{2n+1} \cdot \sum_{s \in S} 2^{|\mov_1(s) \cup \mov_2(s)|})$.

\item Given a state $s \in S$ whether $s \in \almost_1(U,M,\ParityCond)$ 
can be decided in NP $\cap$ coNP.
\end{enumerate}
\end{theo}

\smallskip\noindent{\bf Ranking function for $\mu$-calculus formula.} 
Given a $\mu$-calculus formula of alternation-depth (the nesting depth of 
$\nu$-$\mu$-operators), the \emph{ranking} function maps every state to a tuple of 
$d$-integers, such that each integer is at most  the size of the state space.
For a state that satisfies the $\mu$-calculus formula the tuple of integers denote 
iterations of the $\mu$-calculus formula such that the state got included for 
the first time in the nested evaluation of the $\mu$-calculus formula (for details
see~\cite{EJ91,Kozen83mu}).

The NP $\cap$ coNP bound follows directly from the $\mu$-calculus expressions 
as the players can guess the \emph{ranking function} 
of the $\mu$-calculus formula and the support of the uniform distribution at 
every state to witness that the predecessor operator is satisfied, and the 
guess can be verified in polynomial time.
Observe that the computation through $\mu$-calculus formulas is symbolic and 
more efficient than enumeration over the set of all uniform memoryless 
strategies of size $O(\prod_{s\in S} |\mov_1(s) \cup \mov_2(s)|)$ 
(for example, with constant action size and constant $d$, the $\mu$-calculus 
formula is polynomial, whereas enumeration of strategies is exponential).
The $\mu$-calculus formulas of~\cite{EJ91} can be obtained as a special case 
of the $\mu$-calculus formula of Theorem~\ref{theo-uniform-memless} by 
replacing all predecessor operators with the $\pre_1$ predecessor operator. 


\begin{prop}{}
$\almost_1(\IP,\FM,\Phi)=\almost_1(U,\FM,\Phi)=\almost_1(U,M,\Phi)$.
\end{prop}
\begin{proof}
Consider a finite-memory strategy $\stra_1$ that is almost-sure winning. 
Once the strategy $\stra_1$ is fixed, we obtain a finite-state MDP and 
in MDPs almost-sure winning is independent of the precise transition 
probabilities~\cite{CY95,KrishThesis}.
Hence the strategy $\stra_1^u$ obtained from $\stra_1$ by uniformization
is also winning.
It follows that $\almost_1(\IP,\FM,\Phi)=\almost_1(U,\FM,\Phi)$.
The result that $\almost_1(U,\FM,\Phi)=\almost_1(U,M,\Phi)$ follows from 
Proposition~\ref{prop-uniform-fp}.
\qed
\end{proof}

It follows from above that uniform memoryless strategies are as powerful as
finite-precision infinite-memory strategies for almost-sure winning.
We now show that infinite-precision infinite-memory strategies are more 
powerful than uniform memoryless strategies.

\begin{examp}{($\almost_1(U,M,\Phi) \subsetneq \almost_1(\IP,\IM,\Phi)$).}\label{examp-counter-almost}
We show with an example that for a concurrent parity game with 
three priorities we have $\almost_1(U,M,\Phi) \subsetneq \almost_1(\IP,\IM,\Phi)$.
Consider the game shown in Fig~\ref{figure:three}. 
The moves available for player~1 and player~2 at $s_0$ is $\set{a,b}$ and 
$\set{c,d}$, respectively.
The priorities are as follows: $p(s_0)=1, p(s_2)=3$ and $p(s_1)=2$. 
In other words, player~1 wins if $s_1$ is visited infinitely often and 
$s_2$ is visited finitely often.
We show that for all uniform memoryless strategy for player~1 there 
is counter strategy for player~2 to ensure that the co-parity condition 
is satisfied with probability~1. 
Consider a memoryless strategy $\stra_1$ for player~1, and the counter strategy 
$\stra_2$ is defined as follows: (i)~if $b \in \supp(\stra_1(s_0))$, then play $d$,
(ii)~otherwise, play $c$.
It follows that (i)~if $b \in \supp(\stra_1(s_0))$, then the closed recurrent set $C$ 
of the Markov chain obtained by fixing $\stra_1$ and $\stra_2$ contains $s_2$, and 
hence $s_2$ is visited infinitely often with probability~1;
(ii)~otherwise, player~1 plays the deterministic memoryless strategy that plays 
$a$ at $s_0$, and the counter move $c$ ensures that only $s_0$ is visited infinitely
often.
It follows from our results that for all finite-memory strategies for player~1, 
player~2 can ensure that player~1 cannot win with probability~1.

\begin{figure}[t]
\begin{center}
\begin{picture}(75,20)(0,0)
\node[Nmarks=n](n0)(40,12){$s_0$}
\node[Nmarks=n](n1)(10,12){$s_1$}
\node[Nmarks=n](n2)(70,12){$s_2$}
\drawloop(n0){$ac$}
\drawedge[ELpos=50, ELside=l, ELdist=0.5, curvedepth=0](n0,n2){$bd$}
\drawedge[ELpos=50, ELside=r, ELdist=0.5, curvedepth=0](n0,n1){$ad,bc$}
\drawedge[curvedepth=-5](n1,n0){}
\drawedge[curvedepth=5](n2,n0){}
\end{picture}
\end{center}
  \caption{Three priority concurrent game}
  \label{figure:three}
\end{figure}

We now show that in the game there is an infinite-memory infinite-precision 
strategy for player~1 to win with probability~1 against all player~2 strategies.
Consider a strategy $\stra_1$ for player~1 that is played in rounds, and a 
round is incremented upon visit to $\set{s_1,s_2}$, and in round $k$ the 
strategy plays action $a$ with probability $1-\frac{1}{2^{k+1}}$ and $b$ with 
probability $\frac{1}{2^{k+1}}$.
For $k \geq 0$, let $\cale_k$ denote the event that the game gets stuck at round 
$k$. 
In round $k$, against any strategy for player~2 in every step there is at least 
probability $\eta_k=\frac{1}{2^{k+1}}>0$ to visit the set $\set{s_1,s_2}$.
Thus the probability to be in round $k$ for $\ell$ steps is at most $(1-\eta_k)^\ell$,
and this is~0 as $\ell$ goes to $\infty$.
Thus we have $\Prb_{s_0}^{\stra_1,\stra_2}(\cale_k)=0$. 
Hence the probability that the game is stuck in some round $k$ is 
\[
\Prb_{s_0}^{\stra_1,\stra_2}(\bigcup_{k \geq 0} \cale_k) \leq 
\sum_{k \geq 0} \Prb_{s_0}^{\stra_1,\stra_2}(\cale_k) =0,
\]
where the last equality follows as the countable sum of probability zero 
event is zero.
It follows that 
$\Prb_{s_0}^{\stra_1,\stra_2}(\bo\diam \set{s_1,s_2})=1$, 
i.e., $\set{s_1,s_2}$ is visited infinitely often with probability~1.
To complete the proof we need to show that $\set{s_2}$ is visited 
infinitely often with probability~0.
Consider an arbitrary strategy for player~2.
We first obtain the probability $u_{k+1}$ that $s_2$ is visited $k+1$ times, 
given it has been visited $k$ times. 
Observe that to visit $s_2$ player~2 must play the action $d$, and thus 
\[
u_{k+1} \leq \frac{1}{2^{k+1}} ( 1 + \frac{1}{2} + \frac{1}{4} + \ldots),
\]
where in the infinite sum is obtained by considering the number of
consecutive visits to $s_1$ before $s_2$ is visited.
The explanation of the infinite sum is as follows: the probability to 
reach $s_2$ for $k+1$-th time after the $k$-th visit 
(i)~with only one visit to $s_1$ is $\frac{1}{2^{k+1}}$,
(ii)~with two visits to $s_1$ is $\frac{1}{2^{k+2}}$ (as the probability to play action $b$ is
halved),
(iii)~with three visits to $s_1$ is $\frac{1}{2^{k+3}}$ and so on.
Hence we have $u_{k+1} \leq \frac{1}{2^k}$. 
The probability that $s_2$ is visited infinitely often is 
$\prod_{k=0}^\infty u_{k+1} \leq \prod_{k=0}^\infty \frac{1}{2^{k+1}}=0$. 
It follows that for all strategies $\stra_2$ we have
$\Prb_{s_0}^{\stra_1,\stra_2}(\bo\diam \set{s_2})=0$,
and hence 
$\Prb_{s_0}^{\stra_1,\stra_2}(\bo\diam \set{s_1} \cap \diam \bo\set{s_1,s_0})=1$.
Thus we have shown that player~1 has an infinite-memory infinite-precision 
almost-sure winning strategy.
\qed
\end{examp}

\begin{examp}{($\limit_1(\IP,\FM,\Phi) \subsetneq \limit_1(\IP,\IM,\Phi)$).}
\label{examp:limit-diff}
We show with an example that $\limit_1(\IP,\FM,\Phi) \subsetneq \limit_1(\IP,\IM,\Phi)$.
The example is from~\cite{dAH00} and we present the details for the sake of 
completeness.

\begin{figure}[t]
\begin{center}
\begin{picture}(75,30)(0,0)
\node[Nmarks=n](n0)(40,12){$s_0$}
\node[Nmarks=n](n1)(10,12){$s_1$}
\node[Nmarks=n](n2)(70,12){$s_2$}
\drawloop(n2){}
\drawloop(n0){$ac$}
\drawedge[ELpos=50, ELside=l, ELdist=0.5, curvedepth=0](n0,n2){$bd$}
\drawedge[ELpos=50, ELside=r, ELdist=0.5, curvedepth=0](n0,n1){$ad,bc$}
\drawedge[curvedepth=-5](n1,n0){}
\end{picture}
\end{center}
  \caption{B\"uchi games}
  \label{figure:buchi}
\end{figure}

Consider the game shown in Fig.~\ref{figure:buchi}.
The state $s_2$ is an absorbing state, and from the state $s_1$ the
next state is always $s_0$.
The objective of player~1 is to visit $s_1$ infinitely often,
i.e., $\bo \diam{\set{s_1}}$.
For $\ve>0$, we will construct a strategy $\stra_1^\ve$ for player~1 
that ensures $s_1$ is visited infinitely often with probability at 
least $1-\ve$.
First, given $\ve>0$, we construct a sequence of $\ve_i$, for $i\ge 0$,
such that $\ve_i >0$, and $\prod_{i} (1-\ve_i) \geq (1-\ve)$.
Let $\stra_1^{\ve_i}$ be a memoryless strategy for player~1 
that ensures $s_0$  is reached from $s_1$ with probability at least $1-\ve_i$;
such a strategy can be constructed as in the solution of reachability
games (see 
\cite{crg-tcs07}). 
The strategy $\stra_1^\ve$ is as follows: for a history $w \in S^*$ (finite 
sequence of states), if the number of times $s_1$ has appeared in $w$ is 
$i$, then for the history $w \cdot s_0$ the strategy $\stra_1^\ve$ plays like 
$\stra_1^{\ve_i}$, i.e., $\stra_1^\ve(w\cdot s_0)=\stra_1^{\ve_i}(s_0)$.
The strategy $\stra_\ve$ constructed in this fashion ensures
that against any strategy $\stra_2$, the state $s_1$ is visited
infinitely often with probability at least $ \prod_{i}(1-\ve_i) \geq 1-\ve$. 
However, the strategy $\stra_1^\ve$ counts the number of visits to
$s_1$, and therefore uses infinite memory. 

We now show that the infinite memory requirement cannot be avoided. 
We show now that all finite-memory strategies visit $s_2$ infinitely
often with probability~0. 
Let $\stra$ be an arbitrary finite-memory strategy for player~1, and
let $M$ be the (finite) memory set used by the strategy.
Consider the product game graph defined on the state space 
$\set{s_0,s_1,s_2} \times M$ as follows:
for $s \in \set{s_0,s_1,s_2}$ and $m \in M$, let 
$\stra_u(s,m)=m_1$ (where $\stra_u$ is the memory update function of $\stra$),
then for $a_1 \in \mov_1(s)$ and $b_1 \in \mov_2(s)$ we have 
\[
\ov{\trans}((s,m),a_1,b_1)(s',m')=
\begin{cases}
\trans(s,a_1,b_1)(s') & \text{$m'=m_1$} \\
0 & \text{otherwise}
\end{cases}
\]
where $\ov{\trans}$ is the transition function of the product game graph.
The strategy $\stra$ will be interpreted as a memoryless $\ov{\stra}$ in 
the product game graph as follows: 
for $s \in \set{s_0,s_1,s_2}$ and $m \in M$ we 
have $\ov{\stra}((s,m))=\stra_n((s,m))$, where $\stra_n$ is the next move
function of $\stra$.
Consider now a strategy $\stra_2$ for player~2  constructed as
follows. 
From a state $(s_0, m) \in \set{s_0,s_1,s_2} \times M$, if the strategy 
$\ov{\stra}$ plays~$a$ with probability~1, 
then player~2 plays~$c$ with probability~1, ensuring that the successor 
is $(s_0, m')$ for some $m' \in M$. 
If $\ov{\stra}$ plays~$b$ with positive probability, then player~2
plays~$c$ and~$d$ uniformly at random, ensuring that $(s_2, m')$ is
reached with positive probability, for some $m' \in M$. 
Under $\stra_1, \stra_2$ the game is reduced to a Markov chain, and
since the set $\set{s_2} \times M$ is absorbing, and since all states
in $\set{s_0} \times M$ either stay safe in $\set{s_0}\times M$ or 
reach $\set{s_2} \times M$ in one step with positive probability,
and all states in $\set{s_1}\times M$ reach $\set{s_0}\times M$ in one step, 
the closed recurrent classes must be either entirely contained in 
$\set{s_0} \times M$, or in  $\set{s_2} \times M$. 
This shows that, under $\stra_1, \stra_2$, player~1 achieves the
B\"uchi goal $\bo \diam \set{s_1}$ with probability~0. 
\qed
\end{examp}

The propositions and examples of this section establish all the results for 
equalities and inequalities of the first set of equalities and inequalities of Section~\ref{sec-intro}.
The fact that $\limit_1(\IP,\FM,\Phi) \subsetneq \limit_1(\IP,\IM,\Phi)$
was shown in~\cite{dAH00} (also see Example~\ref{examp:limit-diff}).
The fact that we have $\bigcup_{b>0}\limit_1(b\FP,\IM,\Phi)=\almost_1(U,M,\Phi)$, and the result 
of~\cite{crg-tcs07} that for reachability objectives memoryless limit-sure winning 
strategies exist and limit-sure winning is different from almost-sure winning 
established that $\bigcup_{b>0}\limit_1(b\FP,\IM,\Phi) \subsetneq \limit_1(\IP,M,\Phi)$. 
Thus we have all the results of the first and second set of equalities and inequalities of Section~\ref{sec-intro},
other than $\limit_1(\IP,M,\Phi)=\limit_1(\IP,\FM,\Phi)=\limit_1(\FP,M,\Phi)=\limit_1(\FP,\IM,\Phi)$ 
and we establish this in the next section.

\newcommand{\St}{\mathsf{St}}
\newcommand{\Wk}{\mathsf{Wk}}

\section{Infinite-precision Strategies} 
The results of the previous section already characterize that 
for almost-sure winning infinite-precision finite-memory strategies 
are no more powerful than uniform memoryless strategies.
In this section we characterize the limit-sure winning for infinite-precision 
finite-memory strategies.
We define two new operators, 
$\lpre$ (limit-pre) and $\fpre$ (fractional-pre). 
For $s \in S$ and $X,Y\subs S$, these two-argument predecessor
operators are defined as follows:   
\begin{eqalignno}
  \label{eq-lpre}
  \lpre_1(Y,X) & =\set{s \in S \mid 
		   \forall \alpha >0 \qdot 
		   \exists \dis_1 \in \sd^s_1 \qdot 
		   \forall \dis_2 \in \sd^s_2 \qdot 
		 \bigl[  \pr_s^{\dis_1,\dis_2}(X) > 
		       \alpha \cdot \pr_s^{\dis_1,\dis_2} (\no Y) \bigr]
		       }; \\
 \label{eq-fpre}
 \fpre_2(X,Y) & =\set{s \in S \mid
	\exists \beta >0  \qdot 
	\forall \dis_1 \in \sd^s_1 \qdot 
	\exists \dis_2 \in \sd^s_2 \qdot 
	\bigl[ \pr_s^{\dis_1,\dis_2} (Y) 
	\geq \beta \cdot \pr_s^{\dis_1,\dis_2} (\no X) \bigr] }
	\eqpun .
\end{eqalignno}
The operator $\lpre_1(Y,X)$ is the set of states such that player~1 can choose 
distributions to ensure that the probability to progress to $X$ 
(i.e., $\pr_s^{\dis_1,\dis_2}(X)$)  can be made arbitrarily large 
as compared to the probability of escape from $Y$ (i.e., 
$\pr_s^{\dis_1,\dis_2} (\no Y)$).
Note that $\alpha>0$ can be an arbitrarily large number. 
In other words, the probability to progress to $X$ divided by the sum of the 
probability to progress to $X$ and to escape $Y$ can be made arbitrarily 
close to~1 (in the limit~1).
The operator $\fpre_2(X,Y)$ is the set of states such that against all player~1 
distributions, player~2 can choose a
distribution to ensure that the probability to progress to $Y$ can be made 
greater than a positive constant times the probability of escape from $X$, 
(i.e., progress to $Y$ is a positive fraction of the probability to escape 
from $X$).

\smallskip\noindent{\bf Limit-sure winning for memoryless strategies.}
The results of~\cite{crg-tcs07} show that for reachability objectives,
memoryless strategies suffice for limit-sure winning. 
We now show with an example that limit-sure winning for B\"uchi 
objectives with memoryless strategies is not 
simply limit-sure reachability to the set of almost-sure winning states.
Consider the game shown in Fig~\ref{figure:buchi-lim} with actions 
$\set{a,b}$ for player~1 and $\set{c,d,e}$ for player~2 at $s_0$. 
States $s_1,s_2$ are absorbing, and the unique successor of $s_3$ is $s_0$. 
The B\"uchi objective is to visit $\set{s_1,s_3}$ infinitely often. 
The only almost-sure winning state is $\set{s_1}$. 
The state $s_0$ is not almost-sure winning because at $s_0$ if 
player~1 plays $b$ with positive probability the counter move is $d$, 
otherwise the counter move is $c$. Hence either $s_2$ is 
reached with positive probability or $s_0$ is never left. 
Moreover, player~1 cannot limit-sure reach the state $s_1$ from $s_0$,
as the move $e$ ensures that $s_1$ is never reached. 
Thus in this game the limit-sure reach to the almost-sure winning set is 
only state $s_1$. 
We now show that for all $\ve$, there is a memoryless strategy to ensure 
the B\"uchi objective with probability at least $1-\ve$ from $s_0$. 
At $s_0$ the memoryless strategy plays $a$ with probability $1-\ve$ and 
$b$ with probability $\ve$. Fixing the strategy for player~1 we obtain an 
MDP for player~2, and in the MDP player~2 has an optimal pure memoryless strategy.
If player~2 plays the pure memoryless strategy $e$, then $s_3$ is visited 
infinitely often with probability~1; if player~2 plays the pure memoryless
strategy $c$, then $s_1$ is reached with probability~1; and if player~2 
plays the pure memoryless strategy $d$, then $s_1$ is reached with probability 
$1-\ve$. 
Thus for all $\ve>0$, player~1 can win from $s_0$ and $s_2$ with probability 
at least $1-\ve$ with a memoryless strategy.

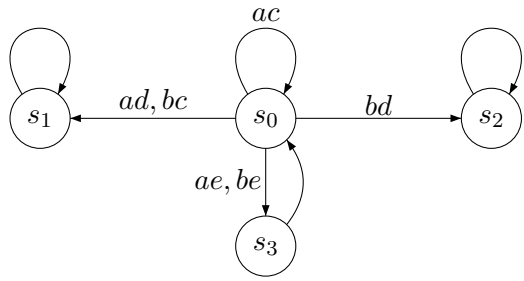
\begin{figure}[t]
\begin{center}
\begin{picture}(75,20)(0,0)
\node[Nmarks=n](n0)(40,12){$s_0$}
\node[Nmarks=n](n1)(10,12){$s_1$}
\node[Nmarks=n](n2)(70,12){$s_2$}
\node[Nmarks=n](n3)(40,-5){$s_3$}
\drawloop(n1){}
\drawloop(n2){}
\drawloop(n0){$ac$}
\drawedge[ELpos=50, ELside=l, ELdist=0.5, curvedepth=0](n0,n2){$bd$}
\drawedge[ELpos=50, ELside=r, ELdist=0.5, curvedepth=0](n0,n1){$ad,bc$}
\drawedge[ELpos=50, ELside=r, ELdist=0.5, curvedepth=0](n0,n3){$ae,be$}
\drawedge[curvedepth=-5](n3,n0){}
\end{picture}
\end{center}
  \caption{A B\"uchi game}
  \label{figure:buchi-lim}
\end{figure}

\smallskip\noindent{\bf Limit-winning set for B\"uchi objectives.} 
We first present the characterization of the set of limit-sure winning states 
for concurrent B\"uchi games from~\cite{dAH00} for infinite-memory and
infinite-precision strategies. 
The limit-sure winning set is characterized by the following formula
\[
\nu Y_0. \mu X_0. [(B \cap \pre_1(Y_0)) \cup (\no B \cap \lpre_1(Y_0,X_0)) ]
\]
Our characterization of the limit-sure winning set for memoryless 
infinite-precision strategies would be obtained as follows: 
we will obtain sequence of chunk of states $X_0 \subs X_1 \subs \ldots \subs X_k$ 
such that from each $X_i$ for all $\ve>0$ there is a memoryless 
strategy to ensure that $\diam X_{i-1} \cup (\bo \diam B \cap \bo 
(X_i\setm X_{i-1}))$ 
is satisfied with probability at least $1-\ve$.
We consider the following $\mu$-calculus formula:
\[
\nu Y_1. \mu X_1. \nu Y_0. \mu X_0. [(B \cap \pre_1(Y_0) \sdcup \lpre_1(Y_1,X_1)) 
\cup (\no B \cap \apre_1(Y_0,X_0) \sdcup \lpre_1(Y_1,X_1)) ]
\]
Let $Y^*$ be the fixpoint, and since it is a fixpoint we have 
\[
\begin{array}{l}
Y^* = \mu X_1. \nu Y_0. \mu X_0. 
\left[
\begin{array}{l} 
\big(B \cap \pre_1(Y_0) \sdcup \lpre_1(Y^*,X_1)\big) 
\cup 
\\
\big(\no B \cap \apre_1(Y_0,X_0) \sdcup \lpre_1(Y^*,X_1)\big)
\end{array}
\right]
\end{array}
\]
Hence $Y^*$ is computed as least fixpoint as sequence of sets 
$X_0 \subs X_1 \ldots \subs X_k$, and 
$X_{i+1}$ is obtained from $X_i$ as
\[
\nu Y_0. \mu X_0. [(B \cap \pre_1(Y_0) \sdcup \lpre_1(Y^*,X_i)) 
\cup (\no B \cap \apre_1(Y_0,X_0) \sdcup \lpre_1(Y^*,X_i)) ]
\]
The $\lpre_i(Y^*,X_i)$ is similar to limit-sure reachability to $X_i$, 
and once we rule out $\lpre_1(Y^*,X_i)$, the formula simplifies to 
the almost-sure winning under memoryless strategies. 
In other words, from each $X_{i+1}$ player~1 can ensure with a memoryless
strategy that either (i) $X_i$ is reached with
limit probability~1 or (ii) the game stays in $X_{i+1}\setm X_i$ and 
the B\"uchi objective is satisfied with probability~1.
It follows that $Y^* \subs \limit_1(\IP,M,\bo \diam B)$. 
We will show that in the complement set there exists constant $\eta>0$ 
such that for all finite-memory infinite-precision 
strategies for player~1 there is a counter strategy to ensure the 
complementary objective with probability at least $\eta>0$.

\smallskip\noindent{\bf The general principle.} 
The general principle to obtain the $\mu$-calculus formula for 
limit-sure winning for memoryless infinite-precision strategies is 
as follows: we consider the $\mu$-calculus formula for the 
almost-sure winning for uniform memoryless strategies, then add 
a $\nu Y_{n+1} \mu X_{n+1}$ quantifier and add the 
$\lpre_1(Y_{n+1},X_{n+1}) \sdcup$ to every predecessor operator. 
Intuitively, when we replace $Y_{n+1}$ by the fixpoint $Y^*$, then we obtain 
sequence $X_i$ of chunks of states for the least fixpoint computation of $X_{n+1}$,
such that from $X_{i+1}$ either $X_i$ is reached with limit probability~1
(by the $\lpre_1(Y^*,X_{n+1})$ operator), or 
the game stays in $X_{i+1} \setm X_i$ and then the parity objective is 
satisfied with probability~1 by a memoryless strategy. 
Formally, we will show Lemma~\ref{lemm:limit-infprec}, and 
we first present a technical lemma required for the correctness 
proof.

\begin{lem}{(Basic $\lpre$ principle).}\label{lemm:basiclpre}
Let $X \subs Y \subs Z \subs S$ and such that all $s\in Y\setm X$
we have $s \in \lpre_1(Z,X)$.
For all prefix-independent events $\cala \subs \bo (Z \setm Y)$, the following assertion holds:
\begin{quote}
 Assume that for all $\eta >0$ there exists a memoryless strategy
 $\stra_1^\eta \in \bigstra_1^M$ 
 such that for all $\stra_2 \in \bigstra_2$ and for all $z \in Z \setm Y$ 
 we have 
 \[\Prb_z^{\stra_1^\eta,\stra_2}(\cala \cup \diam Y) \geq 1- \eta, \qquad 
 (\text{i.e., } 
 \lim_{\eta \to 0} \Prb_z^{\stra_1^\eta,\stra_2}(\cala \cup \diam Y)=1).
 \]
 Then, for all $s \in Y$ for all $\ve>0$ there exists a 
 memoryless strategy $\stra_1^\ve \in \bigstra_1^M$ 
 such that for all $\stra_2 \in \bigstra_2$ we have 
 \[
 \Prb_s^{\stra_1^\ve,\stra_2}(\cala \cup \diam X) \geq 1- \ve, \qquad
 (\text{i.e., } 
 \lim_{\ve \to 0} \Prb_s^{\stra_1^\ve,\stra_2}(\cala \cup \diam X)=1).
 \]
\end{quote}
\end{lem}
\begin{proof} 
The situation is depicted in Figure~\ref{figure:basic-lpre}.(a).
Since for all $s \in Y \setm X$ we have $s \in \lpre_1(Z,X)$, given $\ve>0$, 
player~1 can play the distribution $\dopp{1}{s}{\lpre}{\ve}{Z,X}$ to ensure that the 
probability of going to $\no Z$ is at most $\ve$ times the probability of going 
to $X$.
Fix a counter strategy $\stra_2$ for player~2.
Let $\gamma$ and $\gamma'$ denote the probability of going to $X$ and 
$\no Z$, respectively. Then $\gamma' \leq \ve \cdot \gamma$.
Observe that $\gamma> \ve^l$, where $l=|\movs_s|$.
Let $\alpha$ denote the probability of the event $\cala$.
We first present an informal argument and then present rigorous calculations.
Since $\cala \subs \cala \cup \diam X$,
the worst-case analysis for the result correspond to the case when 
$\alpha=0$, and the simplified situation is shown as 
Fig~\ref{figure:basic-lpre}.(b).
\begin{figure}[t]
   \begin{center}
\setlength{\unitlength}{0.00037500in}
\begingroup\makeatletter\ifx\SetFigFont\undefined%
\gdef\SetFigFont#1#2#3#4#5{%
  \reset@font\fontsize{#1}{#2pt}%
  \fontfamily{#3}\fontseries{#4}\fontshape{#5}%
  \selectfont}%
\fi\endgroup%
{\renewcommand{\dashlinestretch}{30}
\begin{picture}(17519,4697)(0,-10)
\put(14797,85){\makebox(0,0)[lb]{{\SetFigFont{9}{10.8}{\rmdefault}{\mddefault}{\updefault}(c)}}}
\thicklines
\put(2245.476,2613.963){\arc{2023.116}{3.8026}{5.7130}}
\blacken\path(1573.585,3447.547)(1447.000,3235.000)(1658.715,3362.973)(1573.585,3447.547)
\put(9546.066,4400.436){\arc{3877.286}{0.7998}{2.3958}}
\blacken\path(10753.894,2808.206)(10897.000,3010.000)(10675.768,2899.290)(10753.894,2808.206)
\put(8845.476,2688.963){\arc{2023.116}{3.8026}{5.7130}}
\blacken\path(8173.585,3522.547)(8047.000,3310.000)(8258.715,3437.973)(8173.585,3522.547)
\put(15021.066,4400.436){\arc{3877.286}{0.7998}{2.3958}}
\blacken\path(16228.894,2808.206)(16372.000,3010.000)(16150.768,2899.290)(16228.894,2808.206)
\put(14245.476,2688.963){\arc{2023.116}{3.8026}{5.7130}}
\blacken\path(13573.585,3522.547)(13447.000,3310.000)(13658.715,3437.973)(13573.585,3522.547)
\put(1447,3085){\ellipse{212}{212}}
\put(8047,3160){\ellipse{212}{212}}
\put(13447,3160){\ellipse{212}{212}}
\put(3172,3085){\ellipse{212}{212}}
\put(9772,3160){\ellipse{212}{212}}
\put(15172,3160){\ellipse{212}{212}}
\path(22,4660)(5347,4660)(5347,1510)
	(22,1510)(22,4660)
\path(6397,4660)(11722,4660)(11722,1510)
	(6397,1510)(6397,4660)
\path(12172,4660)(17497,4660)(17497,1510)
	(12172,1510)(12172,4660)
\put(2827,4240){\arc{210}{1.5708}{3.1416}}
\put(2827,4405){\arc{210}{3.1416}{4.7124}}
\put(3592,4405){\arc{210}{4.7124}{6.2832}}
\put(3592,4240){\arc{210}{0}{1.5708}}
\path(2722,4240)(2722,4405)
\path(2827,4510)(3592,4510)
\path(3697,4405)(3697,4240)
\path(3592,4135)(2827,4135)
\path(1147,4660)(1147,1510)
\path(1747,4660)(1747,1510)
\path(3847,4660)(3847,1510)
\path(7747,4660)(7747,1510)
\path(8347,4660)(8347,1510)
\path(10447,4660)(10447,1510)
\path(13747,4660)(13747,1510)
\path(16147,4660)(16147,1510)
\path(13147,4660)(13147,1510)
\path(1597,1135)(22,1135)
\blacken\path(262.000,1195.000)(22.000,1135.000)(262.000,1075.000)(262.000,1195.000)
\path(2197,1135)(3772,1135)
\blacken\path(3532.000,1075.000)(3772.000,1135.000)(3532.000,1195.000)(3532.000,1075.000)
\path(8497,1210)(6397,1210)
\blacken\path(6637.000,1270.000)(6397.000,1210.000)(6637.000,1150.000)(6637.000,1270.000)
\path(9097,1210)(10372,1210)
\blacken\path(10132.000,1150.000)(10372.000,1210.000)(10132.000,1270.000)(10132.000,1150.000)
\path(14422,1210)(12172,1210)
\blacken\path(12412.000,1270.000)(12172.000,1210.000)(12412.000,1150.000)(12412.000,1270.000)
\path(14947,1210)(16147,1210)
\blacken\path(15907.000,1150.000)(16147.000,1210.000)(15907.000,1270.000)(15907.000,1150.000)
\path(1372,3085)(322,3085)
\blacken\path(562.000,3145.000)(322.000,3085.000)(562.000,3025.000)(562.000,3145.000)
\path(1597,3085)(3097,3085)
\blacken\path(2857.000,3025.000)(3097.000,3085.000)(2857.000,3145.000)(2857.000,3025.000)
\path(7897,3160)(6847,3160)
\blacken\path(7087.000,3220.000)(6847.000,3160.000)(7087.000,3100.000)(7087.000,3220.000)
\path(8197,3160)(9697,3160)
\blacken\path(9457.000,3100.000)(9697.000,3160.000)(9457.000,3220.000)(9457.000,3100.000)
\path(13372,3160)(12322,3160)
\blacken\path(12562.000,3220.000)(12322.000,3160.000)(12562.000,3100.000)(12562.000,3220.000)
\path(13597,3160)(15097,3160)
\blacken\path(14857.000,3100.000)(15097.000,3160.000)(14857.000,3220.000)(14857.000,3100.000)
\path(3247,3085)(4222,3085)
\blacken\path(3982.000,3025.000)(4222.000,3085.000)(3982.000,3145.000)(3982.000,3025.000)
\path(3172,3160)(3172,4135)
\blacken\path(3232.000,3895.000)(3172.000,4135.000)(3112.000,3895.000)(3232.000,3895.000)
\path(9922,3160)(10897,3160)
\blacken\path(10657.000,3100.000)(10897.000,3160.000)(10657.000,3220.000)(10657.000,3100.000)
\put(397,4210){\makebox(0,0)[lb]{{\SetFigFont{9}{10.8}{\rmdefault}{\mddefault}{\updefault}X}}}
\put(1447,2635){\makebox(0,0)[lb]{{\SetFigFont{9}{10.8}{\rmdefault}{\mddefault}{\updefault}s}}}
\put(6922,4285){\makebox(0,0)[lb]{{\SetFigFont{9}{10.8}{\rmdefault}{\mddefault}{\updefault}X}}}
\put(12547,4210){\makebox(0,0)[lb]{{\SetFigFont{9}{10.8}{\rmdefault}{\mddefault}{\updefault}X}}}
\put(7972,2710){\makebox(0,0)[lb]{{\SetFigFont{9}{10.8}{\rmdefault}{\mddefault}{\updefault}s}}}
\put(13447,2710){\makebox(0,0)[lb]{{\SetFigFont{9}{10.8}{\rmdefault}{\mddefault}{\updefault}s}}}
\put(1822,985){\makebox(0,0)[lb]{{\SetFigFont{9}{10.8}{\rmdefault}{\mddefault}{\updefault}Z}}}
\put(8722,1060){\makebox(0,0)[lb]{{\SetFigFont{9}{10.8}{\rmdefault}{\mddefault}{\updefault}Z}}}
\put(14572,1060){\makebox(0,0)[lb]{{\SetFigFont{9}{10.8}{\rmdefault}{\mddefault}{\updefault}Z}}}
\put(7222,2785){\makebox(0,0)[lb]{{\SetFigFont{9}{10.8}{\rmdefault}{\mddefault}{\updefault}$\gamma$}}}
\put(14047,2785){\makebox(0,0)[lb]{{\SetFigFont{9}{10.8}{\rmdefault}{\mddefault}{\updefault}$\beta$}}}
\put(8797,2785){\makebox(0,0)[lb]{{\SetFigFont{9}{10.8}{\rmdefault}{\mddefault}{\updefault}$\beta$}}}
\put(10222,3235){\makebox(0,0)[lb]{{\SetFigFont{9}{10.8}{\rmdefault}{\mddefault}{\updefault}$\eta$}}}
\put(8647,3760){\makebox(0,0)[lb]{{\SetFigFont{9}{10.8}{\rmdefault}{\mddefault}{\updefault}$1-\eta$}}}
\put(9097,2185){\makebox(0,0)[lb]{{\SetFigFont{9}{10.8}{\rmdefault}{\mddefault}{\updefault}$\gamma\cdot \ve$}}}
\put(12697,2860){\makebox(0,0)[lb]{{\SetFigFont{9}{10.8}{\rmdefault}{\mddefault}{\updefault}$\gamma$}}}
\put(14497,2185){\makebox(0,0)[lb]{{\SetFigFont{9}{10.8}{\rmdefault}{\mddefault}{\updefault}$\gamma\cdot \ve$}}}
\put(2647,2035){\makebox(0,0)[lb]{{\SetFigFont{9}{10.8}{\rmdefault}{\mddefault}{\updefault}$\gamma\cdot \ve$}}}
\put(2047,2710){\makebox(0,0)[lb]{{\SetFigFont{9}{10.8}{\rmdefault}{\mddefault}{\updefault}$\beta$}}}
\put(3547,3160){\makebox(0,0)[lb]{{\SetFigFont{9}{10.8}{\rmdefault}{\mddefault}{\updefault}$\eta$}}}
\put(622,2785){\makebox(0,0)[lb]{{\SetFigFont{9}{10.8}{\rmdefault}{\mddefault}{\updefault}$\gamma$}}}
\put(3247,3835){\makebox(0,0)[lb]{{\SetFigFont{9}{10.8}{\rmdefault}{\mddefault}{\updefault}$\alpha$}}}
\put(2947,4210){\makebox(0,0)[lb]{{\SetFigFont{9}{10.8}{\rmdefault}{\mddefault}{\updefault}$\cala$}}}
\put(1372,3610){\makebox(0,0)[lb]{{\SetFigFont{9}{10.8}{\rmdefault}{\mddefault}{\updefault}$1-\alpha-\eta$}}}
\put(1972,85){\makebox(0,0)[lb]{{\SetFigFont{9}{10.8}{\rmdefault}{\mddefault}{\updefault}(a)}}}
\put(8797,85){\makebox(0,0)[lb]{{\SetFigFont{9}{10.8}{\rmdefault}{\mddefault}{\updefault}(b)}}}
\put(2946.066,4325.436){\arc{3877.286}{0.7998}{2.3958}}
\blacken\path(4153.894,2733.206)(4297.000,2935.000)(4075.768,2824.290)(4153.894,2733.206)
\end{picture}
}
   \end{center}
  \caption{Basic $\lpre$ principle; in the figures $\beta=1-\gamma-\gamma\cdot \ve$}
  \label{figure:basic-lpre}
\end{figure}
Once we let $\eta \to 0$, then we only have an edge from 
$Z \setm Y$ to $Y$ and the situation is shown in 
Fig~\ref{figure:basic-lpre}.(c).
If $q$ is the probability to reach $X$, then the probability to reach $\no Z$
is $q \cdot \ve$ and we have $q + q\ve=1$, i.e., $q =\frac{1}{1+\ve}$, and 
given $\ve'>0$ we can chose $\ve$ to ensure that $q \geq 1-\ve'$.

We now present detailed calculations.
Given $\ve'>0$ we construct a strategy $\stra_1^{\ve'}$ as follows:
let $\ve=\frac{\ve'}{2(1-\ve')}$ and $\eta=\ve^{l+1}>0$; 
and fix the strategy $\stra_1^\eta$ for states in 
$Z \setm Y$ and the distribution $\dopp{1}{s}{\lpre}{\ve}{Z,X}$ at $s$.
Observe that by choice we have $\eta \leq \gamma \cdot \ve$.
Let $q=\Prb_s^{\stra_1^{\ve'},\stra_2}(\cala \cup \diam X)$.
Then we have 
$q \geq \gamma + \beta\cdot \big(\alpha+ (1-\eta-\alpha)\cdot q\big)$; 
since the set $Z\setm Y$ is reached with probability at most $\beta$ and then again $Y$ 
is reached with probability at least $1-\eta-\alpha$ and event $\cala$ happens 
with probability at least $\alpha$.
Hence we have 
\[
q \geq 
\gamma + \beta\cdot \big(\alpha+ (1-\eta-\alpha)\cdot q\big) 
\geq
\gamma + \beta\cdot \big(\alpha \cdot q + (1-\eta-\alpha)\cdot q\big)
=\gamma + \beta\cdot (1-\eta)\cdot q;
\] 
the first inequality follows as $q \leq 1$.
Thus we have
\[
\begin{array}{rcl}
q & \geq & \gamma + (1-\gamma -\gamma \cdot \ve)\cdot (1-\eta) \cdot q; \\
q & \geq  & \displaystyle \frac{\gamma}{\gamma + \gamma \cdot \ve + \eta 
	- \eta \cdot \gamma - \eta \cdot \gamma \cdot \ve} \\[2ex]
  & \geq & \displaystyle \frac{\gamma}{\gamma + \gamma \cdot \ve + \eta } \\[2ex]
  & \geq & \displaystyle \frac{\gamma}{\gamma + \gamma \cdot \ve + \gamma \cdot \ve } \qquad (\text{since } \eta \leq \gamma\cdot \ve)\\[2ex]
  & \geq & \frac{1}{1 + 2\ve} \geq 1-\ve'.
\end{array}
\]
The desired result follows.
\qed
\end{proof}

\begin{lem}{}\label{lemm:limit-infprec}
For a parity function $p:S \mapsto [1..2n]$ and $T \subs S$, we have 
$W \subs \limit_1(\IP,M, \ParityCond \cup \diam T)$, where $W$ is defined 
as follows:
\[
\begin{array}{l}
\nu Y_{n}. \mu X_{n}. \nu Y_{n-1}.  \mu X_{n-1}.  \cdots \nu Y_1. \mu X_1. \nu Y_0. \mu X_0. 
\\ 
\left[
\begin{array}{c}
T \\
B_{2n} \cap \pre_1(Y_{n-1}) \sdcup \lpre_1(Y_n,X_n) \\
\cup \\
B_{2n-1} \cap \lpreodd_1(0,Y_{n-1},X_{n-1}) \sdcup \lpre_1(Y_{n},X_{n}) \\
\cup \\ 
B_{2n-2} \cap \lpreeven_1(0,Y_{n-1},X_{n-1},Y_{n-2})  \sdcup \lpre_1(Y_{n},X_{n}) \\
\cup \\
B_{2n-3} \cap \lpreodd_1(1,Y_{n-1},X_{n-1},Y_{n-2},X_{n-2})  \sdcup \lpre_1(Y_{n},X_{n}) \\
\cup \\
\vdots \\
B_2 \cap \lpreeven_1(n-2,Y_{n-1},X_{n-1},\ldots,Y_1,X_1,Y_0)  \sdcup \lpre_1(Y_{n},X_{n}) \\
\cup \\
B_{1} \cap \lpreodd_1(n-1,Y_{n-1},X_{n-1}, \ldots,Y_0,X_0)  \sdcup \lpre_1(Y_{n},X_{n}) \\
\end{array}
\right]
\end{array}
\]
\end{lem}
\begin{proof}
We first reformulate the algorithm for computing $W$ in an equivalent form.
\[
\begin{array}{l}
\mu X_{n}. \nu Y_{n-1}.  \mu X_{n-1}.  \cdots \nu Y_1. \mu X_1. \nu Y_0. \mu X_0. 
\\ 
\left[
\begin{array}{c}
T \\
B_{2n} \cap \pre_1(Y_{n-1}) \sdcup \lpre_1(W,X_n) \\
\cup \\
B_{2n-1} \cap \lpreodd_1(0,Y_{n-1},X_{n-1}) \sdcup \lpre_1(W,X_{n}) \\
\cup \\ 
B_{2n-2} \cap \lpreeven_1(0,Y_{n-1},X_{n-1},Y_{n-2})  \sdcup \lpre_1(W,X_{n}) \\
\cup \\
B_{2n-3} \cap \lpreodd_1(1,Y_{n-1},X_{n-1},Y_{n-2},X_{n-2})  \sdcup \lpre_1(W,X_{n}) \\
\cup \\
\vdots \\
B_2 \cap \lpreeven_1(n-2,Y_{n-1},X_{n-1},\ldots,Y_1,X_1,Y_0)  \sdcup \lpre_1(W,X_{n}) \\
\cup \\
B_{1} \cap \lpreodd_1(n-1,Y_{n-1},X_{n-1}, \ldots,Y_0,X_0)  \sdcup \lpre_1(W,X_{n}) \\
\end{array}
\right]
\end{array}
\]
The reformulation is obtained as follows: since $W$ is the fixpoint of 
$Y_{n+1}$ we replace $Y_{n+1}$ by $W$ everywhere in the $\mu$-calculus formula,
and get rid of the outermost fixpoint.
The above mu-calculus formula is a least fixpoint and thus computes $W$ 
as an increasing sequence 
$T= T_0 \subset T_1 \subset T_2 \subset \cdots \subset T_m = W$ 
of states, where $m \geq 0$. 
Let $L_i = T_i \setm T_{i-1}$ and 
the sequence is computed by computing $T_i$ as follows, for $0 < i \leq m$: 
\[
\begin{array}{l} 
\nu Y_{n-1}.  \mu X_{n-1}.  \cdots \nu Y_1. \mu X_1. \nu Y_0. \mu X_0. \\ 
\left[
\begin{array}{c}
T \\
B_{2n} \cap \pre_1(Y_{n-1}) \sdcup \lpre_1(W,T_{i-1}) \\
\cup \\
B_{2n-1} \cap \lpreodd_1(0,Y_{n-1},X_{n-1}) \sdcup \lpre_1(W,T_{i-1}) \\
\cup \\ 
B_{2n-2} \cap \lpreeven_1(0,Y_{n-1},X_{n-1},Y_{n-2})  \sdcup \lpre_1(W,T_{i-1}) \\
\cup \\
B_{2n-3} \cap \lpreodd_1(1,Y_{n-1},X_{n-1},Y_{n-2},X_{n-2})  \sdcup \lpre_1(W,T_{i-1}) \\
\cup \\
\vdots \\
B_2 \cap \lpreeven_1(n-2,Y_{n-1},X_{n-1},\ldots,Y_1,X_1,Y_0)  \sdcup \lpre_1(W,T_{i-1}) \\
\cup \\
B_{1} \cap \lpreodd_1(n-1,Y_{n-1},X_{n-1}, \ldots,Y_0,X_0)  \sdcup \lpre_1(W,T_{i-1}) \\
\end{array}
\right]
\end{array}
\]
The above formula is obtained by simply replacing the variable $X_{n}$ by 
$T_{i-1}$.
The proof that $W \subs \limit_1(\IP,M, {\ParityCond \cup \diam T})$ 
is based on an induction on the sequence 
$T = T_0 \subset T_1 \subset T_2 \subset \cdots \subset T_m = W$.
For $1 \leq i \leq m$, let $V_i = W \setm T_{m-i}$, so that
$V_1$ consists of the last block of states that has been added, $V_2$
to the two last blocks, and so on until $V_m = W$. 
We prove by induction on $i \in \{1, \ldots, m\}$, from $i=1$ to $i=m$, 
that for all $s \in V_i$, for all $\eta >0$, 
there exists a memoryless strategy $\stra_1^\eta$ for player~1 such that for all 
$\stra_2 \in \bigstra_2$ we have
\[
\Prb_s^{\stra_1^\eta,\stra_2}\big(\diam T_{m-i} \cup \ParityCond \big) \geq 1-\eta. 
\]
Since the base case is a simplified version of the induction step, we
focus on the latter.

For $V_i \setm V_{i-1}$ we analyze the 
predecessor operator that $s\in V_i \setm V_{i-1}$ satisfies.
The predecessor operators are essentially the predecessor operators of 
the almost-expression for case~1 modified by the addition of the operator
$\lpre_1(W,T_{m-i}) \sdcup$. 
Note that since we fix memoryless strategies for player~1, the analysis
of counter-strategies for player~2 can be restricted to pure memoryless 
(as we have player-2 MDP).
We fix the memoryless strategy for player~1 according to the 
witness distribution of the predecessor operators, and 
consider a pure memoryless counter-strategy for player~2.
Let $Q$ be the set of states where player~2 plays such the 
$\lpre_1(W,T_{m-i})$ part of the predecessor operator
gets satisfied. 
Once we rule out the possibility of $\lpre_1(W,T_{m-i})$, then the 
$\mu$-calculus expression simplifies to the almost-expression of case~2
with $Q \cup T$ as the set of target, i.e.,
\beq 
\nonumber
\nu Y_{n-1}.  \mu X_{n-1}.  \cdots \nu Y_1. \mu X_1. \nu Y_0. \mu X_0. 
\\ 
\left[
\begin{array}{c}
(T \cup Q) \\
B_{2n} \cap \pre_1(Y_{n-1})  \\
\cup \\
B_{2n-1} \cap \lpreodd_1(0,Y_{n-1},X_{n-1})  \\
\cup \\ 
B_{2n-2} \cap \lpreeven_1(0,Y_{n-1},X_{n-1},Y_{n-2})  \\
\cup \\
B_{2n-3} \cap \lpreodd_1(1,Y_{n-1},X_{n-1},Y_{n-2},X_{n-2})  \\
\cup \\
\vdots \\
B_2 \cap \lpreeven_1(n-2,Y_{n-1},X_{n-1},\ldots,Y_1,X_1,Y_0)   \\
\cup \\
B_{1} \cap \lpreodd_1(n-1,Y_{n-1},X_{n-1}, \ldots,Y_0,X_0)  \\
\end{array}
\right]
\eeq
This ensures that if we rule out 
$\lpre_1(W,T_{m-i})$ from the predecessor operators 
and treat the set $Q$ as target, 
then by correctness of the almost-expression for case~2 we 
have that the $\ParityCond \cup \diam (Q \cup T)$ is 
satisfied with probability~1.
By applying the Basic Lpre Principle (Lemma~\ref{lemm:basiclpre}) 
with $Z=W$, $X=T_{m-i}$, $\cala=\ParityCond$ and $Y= X \cup Q$, 
we obtain that for all $\eta>0$ player~1 can ensure with a 
memoryless strategy that $\ParityCond \cup \diam T_{m-i}$ 
is satisfied with probability at least $1-\eta$.
This completes the inductive proof.
With $i=m$ we obtain that for all $\eta>0$, 
there exists a memoryless strategy 
$\stra_1^\eta$ such that for all states $s \in V_m=W$ and for all $\stra_2$ 
we have 
$\Prb_s^{\stra_1^\eta,\stra_2}(\diam T_0 \cup \ParityCond) 
\geq 1-\eta$.
Since $T_0=T$, the desired result follows.
\qed
\end{proof}

We now define the dual predecessor operators (the duality will be shown 
in Lemma~\ref{lemm:dual-limit}). 
We will first use the dual operators to characterize the complement of the 
set of limit-sure winning states for finite-memory infinite-precision 
strategies.
We now introduce two fractional predecessor operators as follows:
\begin{eqnarray*}
\lefteqn{\frpreodd_2(i,Y_n,X_n,\ldots,Y_{n-i},X_{n-i})} \\[1ex]
& = & \fpre_2(X_n,Y_n) \sdcup
\apre_2(X_{n},Y_{n-1}) \sdcup \cdots \sdcup \apre_2(X_{n-i+1},Y_{n-i}) \sdcup \pre_2(X_{n-i}) \\[2ex]
\lefteqn{\frpreeven_2(i,Y_n,X_n,\ldots,Y_{n-i},X_{n-i},Y_{n-i-1})} \\[1ex]
& = & \fpre_2(X_n,Y_n) \sdcup
\apre_2(X_{n},Y_{n-1}) \\[1ex]
& & \qquad \sdcup \cdots \sdcup \apre_2(X_{n-i+1},Y_{n-i}) \sdcup \apre_2(X_{n-i},Y_{n-i-1}) 
\end{eqnarray*}
The fractional operators are same as the $\fpreodd$ and $\fpreeven$ operators,
the difference is the $\epre_2(Y_n)$ is replaced by $\fpre_2(X_n,Y_n)$.
 
\begin{remark}
Observe that if we rule out the predicate $\fpre_2(X_{n},Y_n)$ 
the predecessor operator $\frpreodd_2(i,Y_{n},X_{n},Y_{n},\ldots,Y_{n-i},X_{n-i})$
(resp. $\frpreeven_2(i,Y_n,X_{n},Y_{n-1},\ldots,Y_{n-i},X_{n-i},Y_{n-i-1})$), then 
we obtain the simpler predecessor operator $\lpreeven_2(i,X_{n},Y_{n-1},\ldots,Y_{n-i},X_{n-i})$ 
(resp. $\lpreodd_2(i,X_{n},Y_{n-1},\ldots,Y_{n-i},X_{n-i},Y_{n-i-1})$). 
\end{remark}

The formal expanded definitions of the above operators are as follows:
\[
\begin{array}{l}
\lpreodd_1(i,Y_n,X_n,\ldots, Y_{n-i}, X_{n-i}) \sdcup 
\lpre_1(Y_{n+1}, X_{n+1}) =  \\[1ex] 
\setb{s \in S  \mid  \forall \alpha  > 0 \qdot 
\exists \dis_1 \in \sd^s_1.
\forall \dis_2 \in \sd^s_2. 
\left[\begin{array}{c}
  (\pr_s^{\dis_1,\dis_2}(X_{n+1}) > \alpha \cdot \pr_s^{\dis_1,\dis_2}(\no Y_{n+1})) \\
  \bigvee \\
  (\pr_s^{\dis_1,\dis_2}(X_n) > 0 \land \pr_s^{\dis_1,\dis_2}(Y_n)=1) \\
  \bigvee \\
  (\pr_s^{\dis_1,\dis_2}(X_{n-1}) > 0 \land \pr_s^{\dis_1,\dis_2}(Y_{n-1}) =1) \\
  \bigvee \\
  \vdots \\
  \bigvee \\
  (\pr_s^{\dis_1,\dis_2}(X_{n-i}) > 0 \land  \pr_s^{\dis_1,\dis_2}(Y_{n-i}) =1) 
\end{array} \right]
} \eqpun.
\end{array}
\]

\[
\begin{array}{l}
\lpreeven_1(i,Y_n,X_n,\ldots, Y_{n-i}, X_{n-i},Y_{n-i-1}) \sdcup 
\lpre_1(Y_{n+1},X_{n+1})=  \\[1ex] 
\setb{s \in S  \mid  \forall \alpha  > 0 \qdot 
\exists \dis_1 \in \sd^s_1.
\forall \dis_2 \in \sd^s_2. 
\left[\begin{array}{c}
  (\pr_s^{\dis_1,\dis_2}(X_{n+1}) > \alpha \cdot \pr_s^{\dis_1,\dis_2}(\no Y_{n+1})) \\
  \bigvee \\
  (\pr_s^{\dis_1,\dis_2}(X_n) > 0 \land \pr_s^{\dis_1,\dis_2}(Y_n)=1) \\
  \bigvee \\
  (\pr_s^{\dis_1,\dis_2}(X_{n-1}) > 0 \land \pr_s^{\dis_1,\dis_2}(Y_{n-1}) =1) \\
  \bigvee \\
  \vdots \\
  \bigvee \\
  (\pr_s^{\dis_1,\dis_2}(X_{n-i}) > 0 \land  \pr_s^{\dis_1,\dis_2}(Y_{n-i}) =1) \\
  \bigvee \\
  (\pr_s^{\dis_1,\dis_2}(Y_{n-i-1})=1)
\end{array} \right]
} \eqpun.
\end{array}
\]

The formal expanded definitions of the above operators are as follows:
\[
\begin{array}{l}
\frpreodd_2(i,Y_n,X_n,\ldots, Y_{n-i},X_{n-i}) = \\[1ex]
 \setb{s \in S \mid   
  \exists \beta > 0. 
  \forall \dis_1  \in \sd^s_1. 
  \exists \dis_2  \in \sd^s_2. 
  \left[ \begin{array}{c}
	(\pr_s^{\dis_1,\dis_2}(Y_n) \geq \beta \cdot \pr_s^{\dis_1,\dis_2}(\no X_n)) \\
	\bigvee \\
	(\pr_s^{\dis_1,\dis_2}(Y_{n-1})  > 0 \land \pr_s^{\dis_1,\dis_2}(X_{n}) =1) \\
        \bigvee \\
	(\pr_s^{\dis_1,\dis_2}(Y_{n-2})  > 0 \land \pr_s^{\dis_1,\dis_2}(X_{n-1}) =1) \\
        \bigvee \\
        \vdots \\
        \bigvee \\
	(\pr_s^{\dis_1,\dis_2}(Y_{n-i})  > 0 \land \pr_s^{\dis_1,\dis_2}(X_{n-i+1}) =1) \\
        \bigvee \\
	(\pr_s^{\dis_1,\dis_2}(X_{n-i})=1) 
  \end{array} \right]
 } \eqpun .
\end{array}
\]
\[
\begin{array}{l}
\frpreeven_2(i,Y_n,X_n,\ldots, Y_{n-i},X_{n-i},Y_{n-i-1}) = \\[1ex]
\setb{s \in S \mid   
\exists \beta > 0.  
\forall \dis_1  \in \sd^s_1. 
\exists \dis_2  \in \sd^s_2. 
  \left[ \begin{array}{c}
	(\pr_s^{\dis_1,\dis_2}(Y_n) \geq \beta\cdot \pr_s^{\dis_1,\dis_2}(\no X_n)) \\
	\bigvee \\
	(\pr_s^{\dis_1,\dis_2}(Y_{n-1})  > 0 \land \pr_s^{\dis_1,\dis_2}(X_{n}) =1) \\
        \bigvee \\
	(\pr_s^{\dis_1,\dis_2}(Y_{n-2})  > 0 \land \pr_s^{\dis_1,\dis_2}(X_{n-1}) =1) \\
        \bigvee \\
        \vdots \\
        \bigvee \\
	(\pr_s^{\dis_1,\dis_2}(Y_{n-i-1})  >0 \land \pr_s^{\dis_1,\dis_2}(\no X_{n-i}) =1) 
  \end{array} \right]
 } \eqpun .
\end{array}
\]

We now show the dual of Lemma~\ref{lemm:limit-infprec}.

\begin{lem}{}\label{lemm:limit-infprec1}
For a parity function $p:S \mapsto [1..2n]$  we have 
$Z \subs \no \limit_1(\IP,\FM, \ParityCond)$, where $Z$ is defined 
as follows:
\[
\begin{array}{l}
\mu Y_{n}.  \nu X_{n}. \mu Y_{n-1}. \nu X_{n-1}. \cdots \mu Y_1. \nu X_1. \mu Y_0. \nu X_0. \\
\left[
\begin{array}{c}
B_{2n} \cap \frpreeven_2(0,Y_n,X_n,Y_{n-1} \\
\cup \\
B_{2n-1} \cap \frpreodd_2(1,Y_{n},X_{n},Y_{n-1},X_{n-1}) \\
\cup \\ 
B_{2n-2} \cap \frpreeven_2(1,Y_{n},X_{n},Y_{n-1},X_{n-1},Y_{n-2}) \\
\cup \\
B_{2n-3} \cap \frpreodd_2(2,Y_{n},X_{n},Y_{n-1},X_{n-1},Y_{n-2},X_{n-2}) \\
\cup \\
B_{2n-4} \cap \frpreeven_2(2,Y_{n},X_{n},Y_{n-1},X_{n-1},Y_{n-2},X_{n-2},Y_{n-3}) \\
\vdots \\
B_{3} \cap \frpreodd_2(n-1,Y_{n},X_{n},Y_{n-1},X_{n-1}, \ldots,Y_1,X_1) \\
\cup \\
B_2 \cap \frpreeven_2(n-1,Y_{n},X_{n},Y_{n-1},X_{n-1},\ldots,Y_1,X_1,Y_0) \\
\cup \\
B_1 \cap \frpreodd_2(n,Y_{n},X_{n},Y_{n-1},X_{n-1},\ldots,Y_1,X_1,Y_0,X_0)
\end{array}
\right]
\end{array}
\]
\end{lem}
\begin{proof} 
For $k \geq 0$, let $Z_k$ be the set of states of level $k$ in the above 
$\mu$-calculus expression.
We will show that in $Z_k$, there exists constant $\beta_k>0$, 
such that for every finite-memory strategy for player~1, 
player~2 can ensure that either $Z_{k-1}$ is reached with probability at least $\beta_k$ 
or else $\coParityCond$ is satisfied with probability~1 by staying in 
$(Z_k\setm Z_{k-1})$.
Since $Z_0=\emptyset$, it would follow by induction that 
$Z_k \cap \limit_1(\IP,\FM,\ParityCond)=\emptyset$ and the desired result 
will follow.

 We obtain  $Z_k$ from $Z_{k-1}$ by adding a set of states 
satisfying the following condition:
\[
\begin{array}{l} 
\nu X_{n}. \mu Y_{n-1}. \nu X_{n-1}. \cdots \mu Y_1. \nu X_1. \mu Y_0. \nu X_0. \\
\left[
\begin{array}{c}
B_{2n} \cap \frpreeven_2(0,Z_{k-1},X_n,Y_{n-1} \\
\cup \\
B_{2n-1} \cap \frpreodd_2(1,Z_{k-1},X_{n},Y_{n-1},X_{n-1}) \\
\cup \\ 
B_{2n-2} \cap \frpreeven_2(1,Z_{k-1},X_{n},Y_{n-1},X_{n-1},Y_{n-2}) \\
\cup \\
B_{2n-3} \cap \frpreodd_2(2,Z_{k-1},X_{n},Y_{n-1},X_{n-1},Y_{n-2},X_{n-2}) \\
\cup \\
B_{2n-4} \cap \frpreeven_2(2,Z_{k-1},X_{n},Y_{n-1},X_{n-2},Y_{n-2},X_{n-2},Y_{n-3}) \\
\vdots \\
B_{3} \cap \frpreodd_2(n-1,Z_{k-1},X_{n},Y_{n-1},X_{n-1}, \ldots,Y_1,X_1) \\
\cup \\
B_2 \cap \frpreeven_2(n-1,Z_{k-1},X_{n},Y_{n-1},X_{n-1},\ldots,Y_1,X_1,Y_0) \\
\cup \\
B_1 \cap \frpreodd_2(n,Z_{k-1},X_{n},Y_{n-1},X_{n-1},\ldots,Y_1,X_1,Y_0,X_0)
\end{array}
\right]
\end{array}
\]
The formula is obtained by removing the outer $\mu$ operator, and replacing
$Y_{n+1}$ by $Z_{k-1}$ (i.e., we iteratively obtain the outer fixpoint 
of $Y_{n+1}$).
If the probability of reaching to $Z_{k-1}$ is not positive, then the 
following conditions hold:
\begin{itemize}
\item If the probability to reach $Z_{k-1}$ is not positive, then the 
predicate $\fpre_2(X_{n},Z_{k-1})$ vanishes from the predecessor operator 
$\frpreodd_2(i,Z_{k-1},X_{n},Y_{n-1},\ldots,Y_{n-i},X_{n-i})$, and thus 
the operator simplifies to the simpler predecessor operator 
$\lpreeven_2(i,X_{n},Y_{n-1},\ldots,Y_{n-i},X_{n-i})$.  
\item If the probability to reach $Z_{k-1}$ is not positive, then the 
predicate $\fpre_2(X_{n},Z_{k-1})$ vanishes from the predecessor operator
$\frpreeven_2(i,Z_{k-1},X_{n},Y_{n-1},\ldots,Y_{n-i},X_{n-i},Y_{n-i-1})$, 
and thus the operator simplifies to the simpler predecessor operator
$\lpreodd_2(i,X_{n},Y_{n-1},\ldots,Y_{n-i},X_{n-i},Y_{n-i-1})$. 
\end{itemize}
Hence either the probability to reach $Z_{k-1}$ is positive, and if the probability to reach $Z_{k-1}$ is
not positive,
then the above $\mu$-calculus expression simplifies to
\beq 
\nonumber
Z^*= \nu X_n. \mu Y_{m-1} \nu X_{m-1} \cdots \mu Y_1. \nu X_1. \mu Y_0. 
\left[
\begin{array}{c}
B_{2n} \cap \lpreodd_2(0,X_n,Y_{n-1}) \\
\cup \\
B_{2n-1} \cap \lpreeven_2(1,X_n,Y_{n-1},X_{n-1}) \\
\cup \\
B_{2n-2} \cap \lpreodd_2(1,X_n,Y_{n-1},X_{n-1},Y_{n-2}) \\
\vdots \\
B_{3} \cap \lpreeven_2(n-2,X_n, \ldots,Y_1,X_1) \\
\cup \\
B_2 \cap \lpreodd_2(n-1,X_n,\ldots,Y_1,X_1,Y_0) \\
\cup \\
B_1 \cap \lpreeven_2(n-1,X_n,\ldots,Y_1,X_1,Y_0,X_0) \\
\end{array}
\right].
\eeq
We now consider the parity function $p-1:S\mapsto[0 .. 2n-1]$, and 
observe that the above formula is same as the dual almost-expression for 
case~1.
By correctness of the dual almost-expression we we have $Z^* \subs 
\set{s \in S \mid \forall \stra_1 \in \bigstra_1^M. \exists \stra_2 
\in \bigstra_2. \Pr_s^{\stra_1,\stra_2} 
({\coParityCond})=1}$ 
(since $\Parity(p+1)=\coParityCond$).
It follows that if probability to reach $Z_{k-1}$ is not positive, then against every memoryless 
strategy for player~1, player~2 can fix a pure memoryless strategy to ensure that 
player~2 wins with probability~1. 
In other words, against every distribution of player~1, there is a counter-distribution 
for player~2 (to satisfy the respective $\lpreeven_2$ and $\lpreodd_2$ operators) 
to ensure to win with probability~1.
It follows that for every memoryless strategy for player~1, 
player~2 has a pure memoryless strategy to ensure that for every closed recurrent 
$C \subs Z^*$ we have $\min(p(C))$ is odd. 
It follows that for any finite-memory strategy for player~1 with $\mem$, 
player~2 has a finite-memory strategy to ensure that for every closed recurrent set 
$C'\times \mem'\subs Z^* \times \mem$, the closed recurrent set $C'$ is a union 
of closed recurrent sets $C$ of $Z^*$, and hence  $\min(p(C'))$ is odd
(also see Example~\ref{examp:limit-diff} as an illustration).
It follows that against all finite-memory strategies, player~2 can ensure if 
the game stays in $Z^*$, then $\coParityCond$ is satisfied with 
probability~1.
The $\fpre_2$ operator ensures that if $Z^*$ is left and $Z_{k-1}$ is 
reached, then the probability to reach $Z_{k-1}$ is at least a positive fraction 
$\beta_k$ of the probability to leave $Z_k$. 
In all cases it follows that 
$Z_{k} \subs \set{s \in S \mid \exists \beta_k >0. \forall \stra_1 
\in \bigstra_1^{\FM}. \exists \stra_2 \in \bigstra_2. 
\Prb_s^{\stra_1,\stra_2}(\coParityCond \cup \diam Z_{k-1}) \geq \beta_k
}$.
Thus the desired result follows.
\qed
\end{proof}

\begin{lem}{(Duality of limit predecessor operators).}\label{lemm:dual-limit}
The following assertions hold.
\begin{enumerate}
\item Given $X_{n+1} \subs X_n \subs X_{n-1} \subs \cdots \subs X_{n-i} \subs Y_{n-i} \subs
Y_{n-i+1} \subs \cdots \subs Y_n \subs Y_{n+1} $, 
we have 
\[
\begin{array}{rcl}
&\frpreodd_2 &(i+1,\no Y_{n+1},\no X_{n+1}, \no Y_n,\no X_n,  \ldots ,\no Y_{n-i},\no X_{n-i}) \\ & = &
\no (\lpreodd_1(i,Y_n,X_n, \ldots,Y_{n-i},X_{n-i}) \sdcup \lpre_1(Y_{n+1},X_{n+1}) ).
\end{array}
\]

\item Given $X_{n+1} \subs X_n \subs X_{n-1} \subs \cdots \subs X_{n-i} \subs Y_{n-i-1}
\subs Y_{n-i} \subs
Y_{n-i+1} \subs \cdots \subs Y_n \subs Y_{n+1} $ and $s \in S$,
we have
\[
\begin{array}{rcl}
&\frpreeven_2&(i+1,\no Y_{n+1}, \no X_{n+1},  \no Y_n,\no X_n, \ldots ,\no Y_{n-i},\no X_{n-i},\no Y_{n-i-1}) \\ 
& = &
\no (\lpreeven_1(i,Y_n,X_n, \ldots,Y_{n-i},X_{n-i},Y_{n-i-1}) \sdcup \lpre_1(Y_{n+1},X_{n+1}) ).
\end{array}
\]
\end{enumerate}
\end{lem}
\begin{proof} 
We present the proof for part 1, and the proof for second part is analogous.
To present the proof of the part 1, we present the proof for the 
case when $n=1$ and $i=1$. This proof already has all the ingredients of the 
general proof, and the generalization is straightforward as in 
Lemma~\ref{lemm:dual}.

\smallskip\noindent{\bf Claim.}
We show that for $X_1 \subs X_0 \subs Y_0 \subs Y_1$ we have 
$\fpre_2(\no X_1,\no Y_1) \sdcup \apre_2(\no X_1, \no Y_0) \sdcup \pre_2(\no X_0) = 
\no (\lpre_1(Y_1, X_1) \sdcup  \apre_1(Y_0,X_0))$.
We start with a few notations.
Let $\St \subs \mov_2(s)$ and $\Wk \subs \mov_2(s)$ be set of \emph{strongly} 
and \emph{weakly} covered actions for player~2.
Given $\St \subs \Wk \subs \mov_2(s)$, we say that 
a set $U \subs \mov_1(s)$ satisfy \emph{consistency} condition if 
\[ 
\begin{array}{l}
\forall b \in \St.\ \dest(s,U,b) \cap X_1 \neq \emptyset
\\
\forall b \in \Wk.\ 
(\dest(s,U,b) \cap X_1 \neq \emptyset) \lor 
(\dest(s,U,b) \subs Y_0 \land \dest(s,U,b) \cap X_0 \neq \emptyset) 
\end{array}
\]
A triple $(U,\St,\Wk)$ is consistent if $U$ satisfies the consistency condition. 
We define a function $f$ that takes as argument a triple $(U,\St,\Wk)$ that 
is consistent, and returns three sets 
$f(U,\St,\Wk)=(U',\St',\Wk')$ satisfying the following conditions:
\[
\begin{array}{l}
(1)\ \dest(s,U', \mov_2(s) \setm \Wk) \subs Y_1; \\
(2)\ \St'=\set{b \in \mov_2(s) \mid \dest(s,U',b) \cap X_1 \neq \emptyset} \\
(3)\ \Wk'=\set{b \in \mov_2(s) \mid 
(\dest(s,U',b) \cap X_1 \neq \emptyset) \lor 
(\dest(s,U',b) \subs Y_0 \land \dest(s,U',b) \cap X_0 \neq \emptyset) 
}
\end{array}
\]
We require that $(U,\St,\Wk) \subs (U',\St',\Wk')$ and also require $f$ to 
return a larger set than the input arguments, if possible.
We now consider a sequence of actions sets until a fixpoint is reached:
$\St_{-1}=\Wk_{-1}=U_{-1} =\emptyset$ and 
for $i \geq 0$ we have $(U_{i},\St_{i},\Wk_{i})= f(U_{i-1},\St_{i-1},\Wk_{i-1})$.
Let $(U_*, \St_*, \Wk_*)$ be the set fixpoints (that is $f$ cannot return 
any larger set).
Observe that every time $f$ is invoked it is ensured that the argument form 
a consistent triple.
Observe that we have $\St_i \subs \Wk_i$ and hence 
$\St_* \subs \Wk_*$.
We now show the following two claims.

\begin{enumerate}
\item We first show that if  $\Wk_*=\mov_2(s)$, then 
$s \in \lpre_1(Y_1,X_1) \sdcup \apre_1(Y_0,X_0)$. 
We first define the rank of actions: for an action 
$a \in U_*$ the rank $\ell(a)$ of the action is 
$\min_{i} a \in U_i$.
For an action $b \in \mov_2(s)$, if $b \in \St_*$, then
the strong rank $\ell_s(b)$ is defined as $\min_{i} b \in \St_i$; 
and for an action $b \in \Wk_*$, the weak rank $\ell_w(b)$ is 
defined as $\min_{i} b \in \Wk_i$.
For $\ve>0$, consider a distribution that plays actions in 
$U_i$ with probability proportional to $\ve^i$.
Consider an action $b$ for player~2. We consider the following cases:
(a) If $b \in \St_*$, then let $j=\ell_s(b)$. Then for all actions 
$a \in U_j$ we have $\dest(s,a,b) \subs Y_1$ and for some action $a \in U_j$
we have $\dest(s,a,b) \cap X_1 \neq \emptyset$, in other words, the 
probability to leave $Y_1$ is at most proportional to $\ve^{j+1}$ and 
the probability to goto $X_1$ is at least proportional to $\ve^j$, and 
the ratio is $\ve$.
Since $\ve>0$ is arbitrary, the $\lpre_1(Y_1,X_1)$ part can be ensured.
(b)~If $b \not\in \St_*$, then let $j=\ell_w(b)$. Then for all $a \in U_*$
we have $\dest(s,a,b) \subs Y_0$ and there exists $a \in U_*$
such that $\dest(s,a,b) \cap X_0 \neq \emptyset$.
It follows that in first case the condition for $\lpre_1(Y_1,X_1)$ is 
satisfied, and in the second case the condition for $\apre_1(Y_0,X_0)$ 
is satisfied. 
The desired result follows.

\item We now show that $\mov_2(s) \setm \Wk_* \neq \emptyset$, then 
$s \in \fpre_2(\no X_1,\no Y_1) \sdcup \apre_2(\no X_1, \no Y_0) \sdcup \pre_2(\no X_0)$. 
Let $\ov{U}=\mov_1(s) \setm U_*$, and let $B_k=\mov_2(s) \setm \Wk_*$ and 
$B_s =\mov_2(s) \setm \St_*$.
We first present the required properties about the actions that follows 
from the fixpoint characterization. 
\begin{enumerate}
\item \emph{Property~1.} For all $b \in B_k$, for all $a \in U_*$ we 
have 
\[ 
\dest(s,a,b) \subs \no X_1 \land (\dest(s,a,b) \subs \no X_0 \lor 
\dest(s,a,b) \cap \no Y_0 \neq \emptyset).
\]  
Otherwise the action $b$ would have been included in $\Wk_*$ and 
$\Wk_*$ could be enlarged.

\item \emph{Property~2.} For all $b \in B_s$ and for all $a \in U_*$ 
we have $\dest(s,a,b) \subs \no X_1$.
Otherwise $b$ would have been included in $\St_*$ and $\St_*$ could be
enlarged.

\item \emph{Property~3.} For all $a \in \ov{U}$, 
either 
\begin{enumerate}
\item $\dest(s,a,B_k) \cap \no Y_1 \neq \emptyset$; 
or 
\item  for all $b \in B_s$, $\dest(s,a,b) \subs \no X_1$ 
and for all $b \in B_k$, 
\[
\dest(s,a,b) \subs \no X_1 \land  
(\dest(s,a,b) \subs \no X_0 \lor \dest(s,a,b) \cap \no Y_0 \neq \emptyset)
\]
\end{enumerate}
The property is proved as follows: 
if $\dest(s,a, B_k) \subs Y_1$ and for some $b \in B_s$ we have 
$\dest(s,a,b) \cap X_1 \neq \emptyset$, then $a$ can be included in $U_*$ and 
$b$ can be included in $\St_*$;
if $\dest(s,a, B_k) \subs Y_1$ and for some $b \in B_k$ we have 
\[
(\dest(s,a,b) \cap X_1 \neq \emptyset) \lor 
(\dest(s,a,b) \cap X_0 \neq \emptyset \land \dest(s,a,b) \subs  Y_0)
\]
then $a$ can be included in $U_*$ and 
$b$ can be included in $\Wk_*$.
This would contradict that $(U_*,\St_*,\Wk_*)$ is a fixpoint.
\end{enumerate}

Let $\dis_1$ be a distribution for player~1. Let $Z=\supp(\dis_1)$. 
We consider the following cases to establish the result.
\begin{enumerate}
\item We first consider the case when $Z \subs U_*$. We consider the counter distribution 
$\dis_2$ that plays all actions in $B_k$ uniformly. Then by property~1 we have
(i)~$\dest(s,\dis_1,\dis_2) \subs \no X_1$; and 
(ii)~for all $a \in Z$ we have $\dest(s,a,\dis_2) \subs \no X_0$ or 
$\dest(s,a,\dis_2) \cap \no Y_0 \neq \emptyset$. 
If for all $a \in Z$ we have $\dest(s,a,\dis_2) \subs \no X_0$, then 
$\dest(s,\dis_1,\dis_2) \subs \no X_0$ and $\pre_2(\no X_0)$ is satisfied.
Otherwise we have $\dest(s,\dis_1,\dis_2) \subs \no X_1$ and $\dest(s,\dis_1,\dis_2) 
\cap \no Y_0 \neq \emptyset$, i.e., $\apre_2(\no X_1, \no Y_0)$ is 
satisfied.

\item We now consider the case when $Z \cap \ov{U} \neq \emptyset$. 
Let $U_0=U_*$, and we will iteratively compute sets $U_0 \subs U_i \subs Z$ 
such that (i)~$\dest(s,U_i,B_s) \subs \no X_1$ and 
(ii)~for all $a \in U_i$ we have $\dest(s,a,B_k) \subs \no X_0$ or 
$\dest(s,a,B_k) \subs \no Y_0$  
(unless we have already witnessed that player~2 can satisfy the predecessor 
operator).
In base case the result holds by property~2.
The argument of an iteration is as follows, and we use $\ov{U}_i =Z \setm U_i$. 
Among the actions of $Z \cap\ov{U}_i$, let $a^*$ be the action played 
with maximum probability. 
We have the following two cases.
\begin{enumerate}
\item 
If there exists $b \in B_s$ such that $\dest(s,a^*,b) \cap \no Y_1 
\neq \emptyset$, consider the counter action $b$. 
Since $b \in B_s$, by hypothesis we have $\dest(s,U_i,b) \subs \no X_1$.
Hence the probability to go out of $\no X_1$ is at most the total 
probability of the actions in $Z \cap \ov{U}_i$ and for the
maximum probability action $a^* \in Z \cap \ov{U}_i$ the set 
$\no Y_1$ is reached. Let $\eta>0$ be the minimum positive transition 
probability, then fraction of probability to go to $\no Y_1$ 
as compared to go out of $\no X_1$ is at least 
$\beta=\eta \cdot \frac{1}{|\mov_1(s)|}>0$.
Thus $\fpre_2(\no X_1,\no Y_1)$ can be ensured by playing $b$.  
 
\item Otherwise, by property~3, 
(i)~either $\dest(s,a^*,B_k) \cap \no Y_1 \neq \emptyset$, 
or 
(ii)~for all $b \in B_s$ we have $\dest(s,a_*,b) \subs \no X_1$ and for all $b \in B_k$ 
\[
\dest(s,a^*,b) \subs \no X_1 \land  
(\dest(s,a^*,b) \subs \no X_0 \lor \dest(s,a^*,b) \cap \no Y_0 \neq \emptyset)
\]
If $\dest(s,a^*,B_k) \cap \no Y_1 \neq \emptyset$, then chose the action 
$b \in B_k$ such that $\dest(s,a^*,b) \cap \no Y_1 \neq \emptyset$. 
Since $b \in B_k \subs B_s$, and by hypothesis $\dest(s,U_i,B_s) 
\subs \no X_1$,  we have $\dest(s,U_i,b) \subs \no X_1$.
Thus we have a witness action $b$ exactly as in the previous case, and 
like the proof above $\fpre_2(\no X_1,\no Y_1)$ can be ensured.
If $\dest(s,a^*,B_k) \subs Y_1$, then we claim that 
$\dest(s,a^*, B_s) \subs \no X_1$. 
The proof of the claim is as follows: if $\dest(s,a^*,B_k) \subs Y_1$
and $\dest(s,a^*,B_s) \cap X_1 \neq \emptyset$, then chose the action $b^*$ 
from $B_s$ such that $\dest(s,a^*,b^*) \cap X_1 \neq \emptyset$, and 
then we can include $a^*$ to $U_*$ and $b^*$ to $\St_*$ (contradicting that 
they are the fixpoints). 
It follows that we can include $a^* \in U_{i+1}$ and continue.
\end{enumerate}
Hence we have either already proved that player~2 can ensure the 
predecessor operator or $U_i=Z$ in the end. 
If $U_i$ is $Z$ in the end, then $Z$ satisfies the property used in the 
previous cases of $U_*$ (the proof of part a), and then as in the 
previous proof (of part a), the uniform distribution over $B_k$ is a 
witness that player~2 can ensure $\pre_2(X_0) \sdcup \apre_2(\no X_1, \no Y_0)$.

\end{enumerate}

\end{enumerate}

\smallskip\noindent{\bf General case.} The proof for the general case is a
tedious extension of the result presented for $n=1$ and $i=1$. 
We present the details for the sake of completeness.
We show that for $X_{n+1} \subs X_n \subs X_{n-1} \subs \cdots \subs X_{n-i} \subs Y_{n-i} \subs
Y_{n-i+1} \subs \cdots \subs Y_n \subs Y_{n+1} $, 
we have 
\[
\begin{array}{rcl}
&\frpreodd_2 &(i+1,\no Y_{n+1},\no X_{n+1}, \no Y_n,\no X_n,  \ldots ,\no Y_{n-i},\no X_{n-i}) \\ & = &
\no (\lpreodd_1(i,Y_n,X_n, \ldots,Y_{n-i},X_{n-i}) \sdcup \lpre_1(Y_{n+1},X_{n+1}) ).
\end{array}
\]
We use notations similar to the special case.
Let $\St \subs \mov_2(s)$ and $\Wk \subs \mov_2(s)$ be set of \emph{strongly} 
and \emph{weakly} covered actions for player~2.
Given $\St \subs \Wk \subs \mov_2(s)$, we say that 
a set $U \subs \mov_1(s)$ satisfy \emph{consistency} condition if 
\[ 
\begin{array}{l}
\forall b \in \St.\ \dest(s,U,b) \cap X_{n+1} \neq \emptyset
\\
\forall b \in \Wk.\ 
(\dest(s,U,b) \cap X_{n+1} \neq \emptyset) \lor 
\exists 0 \leq j \leq i. (\dest(s,U,b) \subs Y_{n-j} \land \dest(s,U,b) \cap 
X_{n-j} \neq \emptyset) 
\end{array}
\]
A triple $(U,\St,\Wk)$ is consistent if $U$ satisfies the consistency condition. 
We define a function $f$ that takes as argument a triple $(U,\St,\Wk)$ that 
is consistent, and returns three sets 
$f(U,\St,\Wk)=(U',\St',\Wk')$ satisfying the following conditions:
\[
\begin{array}{l}
(1)\ \dest(s,U', \mov_2(s) \setm \Wk) \subs Y_{n+1}; \\
(2)\ \St'=\set{b \in \mov_2(s) \mid \dest(s,U',b) \cap X_{n+1} \neq \emptyset} \\
(3)\ \Wk'=\set{b \in \mov_2(s) \mid 
(\dest(s,U',b) \cap X_{n+1} \neq \emptyset) \lor \\
\qquad \qquad \qquad \qquad \qquad \qquad \exists 0\leq j \leq i.  
(\dest(s,U',b) \subs Y_{n-j} \land \dest(s,U',b) \cap X_{n-j} \neq \emptyset) 
}
\end{array}
\]
We require that $(U,\St,\Wk) \subs (U',\St',\Wk')$ and also require $f$ to 
return a larger set than the input arguments, if possible.
We now consider a sequence of actions sets until a fixpoint is reached:
$\St_{-1}=\Wk_{-1}=U_{-1} =\emptyset$ and 
for $i \geq 0$ we have $(U_{i},\St_{i},\Wk_{i})= f(U_{i-1},\St_{i-1},\Wk_{i-1})$.
Let $(U_*, \St_*, \Wk_*)$ be the set fixpoints (that is $f$ cannot return 
any larger set).
Observe that every time $f$ is invoked it is ensured that the argument form 
a consistent triple.
Observe that we have $\St_i \subs \Wk_i$ and hence 
$\St_* \subs \Wk_*$.
We now show the following two claims.

\begin{enumerate}
\item We first show that if  $\Wk_*=\mov_2(s)$, then 
$s \in \lpre_1(Y_{n+1},X_{n+1}) \sdcup 
\lpreodd_1(i,Y_n,X_n,\ldots,Y_{n-i},X_{n-i})$. 
We first define the rank of actions: for an action 
$a \in U_*$ the rank $\ell(a)$ of the action is 
$\min_{i} a \in U_i$.
For an action $b \in \mov_2(s)$, if $b \in \St_*$, then
the strong rank $\ell_s(b)$ is defined as $\min_{i} b \in \St_i$; 
and for an action $b \in \Wk_*$, the weak rank $\ell_w(b)$ is 
defined as $\min_{i} b \in \Wk_i$.
For $\ve>0$, consider a distribution that plays actions in 
$U_i$ with probability proportional to $\ve^i$.
Consider an action $b$ for player~2. We consider the following cases:
(a) If $b \in \St_*$, then let $j=\ell_s(b)$. Then for all actions 
$a \in U_j$ we have $\dest(s,a,b) \subs Y_{n+1}$ and for some action $a \in U_j$
we have $\dest(s,a,b) \cap X_{n+1} \neq \emptyset$, in other words, the 
probability to leave $Y_{n+1}$ is at most proportional to $\ve^{j+1}$ and 
the probability to goto $X_{n+1}$ is at least proportional to $\ve^j$, and 
the ratio is $\ve$.
Since $\ve>0$ is arbitrary, the $\lpre_1(Y_{n+1},X_{n+1})$ part can be ensured.
(b)~If $b \not\in \St_*$, then let $j=\ell_w(b)$. Then for all $a \in U_*$ 
there exists $0 \leq j \leq i$ such that 
we have $\dest(s,a,b) \subs Y_{n-j}$ and there exists $a \in U_*$
such that $\dest(s,a,b) \cap X_{n-j} \neq \emptyset$.
It follows that in first case the condition for $\lpre_1(Y_{n+1},X_{n+1})$ is 
satisfied, and in the second case the condition for 
$\lpreodd_1(i,Y_n,X_n,\ldots,Y_{n-i},X_{n-i})$ is satisfied. 
The desired result follows.

\item We now show that $\mov_2(s) \setm \Wk_* \neq \emptyset$, then 
\[
s \in \frpreodd_2(i+1,\no Y_{n+1},\no X_{n+1}, \no Y_n,\no X_n,  \ldots ,\no Y_{n-i},\no X_{n-i}).
\] 
Let $\ov{U}=\mov_1(s) \setm U_*$, and let $B_k=\mov_2(s) \setm \Wk_*$ and 
$B_s =\mov_2(s) \setm \St_*$.
We first present the required properties about the actions that follows 
from the fixpoint characterization. 
\begin{enumerate}
\item \emph{Property~1.} For all $b \in B_k$, for all $a \in U_*$ we 
have 
\[ 
\dest(s,a,b) \subs \no X_{n+1} \land \exists 0\leq j \leq i. 
(\dest(s,a,b) \subs \no X_{n-j} \lor 
\dest(s,a,b) \cap \no Y_{n-j} \neq \emptyset).
\]  
Otherwise the action $b$ would have been included in $\Wk_*$ and 
$\Wk_*$ could be enlarged.

\item \emph{Property~2.} For all $b \in B_s$ and for all $a \in U_*$ 
we have $\dest(s,a,b) \subs \no X_{n+1}$.
Otherwise $b$ would have been included in $\St_*$ and $\St_*$ could be
enlarged.

\item \emph{Property~3.} For all $a \in \ov{U}$, 
either 
\begin{enumerate}
\item $\dest(s,a,B_k) \cap \no Y_{n+1} \neq \emptyset$; 
or 
\item  for all $b \in B_s$, $\dest(s,a,b) \subs \no X_{n+1}$ 
and for all $b \in B_k$, 
\[
\dest(s,a,b) \subs \no X_{n+1} \land \exists 0 \leq j \leq i. 
(\dest(s,a,b) \subs \no X_{n-j} \lor \dest(s,a,b) \cap \no Y_{n-j} \neq \emptyset)
\]
\end{enumerate}
The property is proved as follows: 
if $\dest(s,a, B_k) \subs Y_{n+1}$ and for some $b \in B_s$ we have 
$\dest(s,a,b) \cap X_{n+1} \neq \emptyset$, then $a$ can be included in $U_*$ and 
$b$ can be included in $\St_*$;
if $\dest(s,a, B_k) \subs Y_{n+1}$ and for some $b \in B_k$ we have 
\[
(\dest(s,a,b) \cap X_{n+1} \neq \emptyset) \lor 
\exists 0 \leq j \leq i. 
(\dest(s,a,b) \cap X_{n-j} \neq \emptyset \land \dest(s,a,b) \subs  Y_{n-j})
\]
then $a$ can be included in $U_*$ and 
$b$ can be included in $\Wk_*$.
This would contradict that $(U_*,\St_*,\Wk_*)$ is a fixpoint.
\end{enumerate}

Let $\dis_1$ be a distribution for player~1. Let $Z=\supp(\dis_1)$. 
We consider the following cases to establish the result.
\begin{enumerate}
\item We first consider the case when $Z \subs U_*$. We consider the counter distribution 
$\dis_2$ that plays all actions in $B_k$ uniformly. Then by property~1 we have
(i)~$\dest(s,\dis_1,\dis_2) \subs \no X_{n+1}$; and 
(ii)~for all $a \in Z$ there exists $j\leq i$ such that $\dest(s,a,\dis_2) \subs \no X_{n-j}$ or 
$\dest(s,a,\dis_2) \cap \no Y_{n-j} \neq \emptyset$. 
If for all $a \in Z$ we have $\dest(s,a,\dis_2) \subs \no X_{n-i}$, then 
$\dest(s,\dis_1,\dis_2) \subs \no X_{n-i}$ and $\pre_2(\no X_{n-i})$ is satisfied.
Otherwise, there must exists $j \leq i$ such that $\dest(s,\dis_1,\dis_2) \subs \no X_{n+1-j}$ 
and $\dest(s,\dis_1,\dis_2) \cap \no Y_{n-j} \neq \emptyset$, i.e., 
$\lpreodd_2(i,\no X_{n+1}, \no Y_{n}\ldots, \no X_{n-i+1}, \no Y_{n-i})$ is satisfied.

\item We now consider the case when $Z \cap \ov{U} \neq \emptyset$. 
Let $U_0=U_*$, and we will iteratively compute sets $U_0 \subs U_\ell \subs Z$ 
such that (i)~$\dest(s,U_\ell,B_s) \subs \no X_{n+1}$ and 
(ii)~for all $a \in U_\ell$ there exists $j \leq i$ such that 
$\dest(s,a,B_k) \subs \no X_{n-j}$ or $\dest(s,a,B_k) \subs \no Y_{n-j}$  
(unless we have already witnessed that player~2 can satisfy the predecessor 
operator).
In base case the result holds by property~2.
The argument of an iteration is as follows, and we use $\ov{U}_\ell =Z \setm U_\ell$. 
Among the actions of $Z \cap\ov{U}_\ell$, let $a^*$ be the action played 
with maximum probability. 
We have the following two cases.
\begin{enumerate}
\item 
If there exists $b \in B_s$ such that $\dest(s,a^*,b) \cap \no Y_{n+1} 
\neq \emptyset$, consider the counter action $b$. 
Since $b \in B_s$, by hypothesis we have $\dest(s,U_\ell,b) \subs \no X_{n+1}$.
Hence the probability to go out of $\no X_{n+1}$ is at most the total 
probability of the actions in $Z \cap \ov{U}_\ell$ and for the
maximum probability action $a^* \in Z \cap \ov{U}_\ell$ the set 
$\no Y_{n+1}$ is reached. Let $\eta>0$ be the minimum positive transition 
probability, then fraction of probability to go to $\no Y_{n+1}$ 
as compared to go out of $\no X_{n+1}$ is at least 
$\beta=\eta \cdot \frac{1}{|\mov_1(s)|}>0$.
Thus $\fpre_2(\no X_{n+1},\no Y_{n+1})$ can be ensured by playing $b$.  
 
\item Otherwise, by property~3, 
(i)~either $\dest(s,a^*,B_k) \cap \no Y_{n+1} \neq \emptyset$, 
or 
(ii)~for all $b \in B_s$ we have $\dest(s,a_*,b) \subs \no X_{n+1}$ and for all $b \in B_k$ 
\[
\dest(s,a^*,b) \subs \no X_{n+1} \land \exists 0 \leq j \leq i.  
(\dest(s,a^*,b) \subs \no X_{n-j} \lor \dest(s,a^*,b) \cap \no Y_{n-j} \neq \emptyset)
\]
If $\dest(s,a^*,B_k) \cap \no Y_{n+1} \neq \emptyset$, then chose the action 
$b \in B_k$ such that $\dest(s,a^*,b) \cap \no Y_{n+1} \neq \emptyset$. 
Since $b \in B_k \subs B_s$, and by hypothesis $\dest(s,U_\ell,B_s) 
\subs \no X_1$,  we have $\dest(s,U_\ell,b) \subs \no X_{n+1}$.
Thus we have a witness action $b$ exactly as in the previous case, and 
like the proof above $\fpre_2(\no X_{n+1},\no Y_{n+1})$ can be ensured.
If $\dest(s,a^*,B_k) \subs Y_{n+1}$, then we claim that 
$\dest(s,a^*, B_s) \subs \no X_{n+1}$. 
The proof of the claim is as follows: if $\dest(s,a^*,B_k) \subs Y_{n+1}$
and $\dest(s,a^*,B_s) \cap X_{n+1} \neq \emptyset$, then chose the action $b^*$ 
from $B_s$ such that $\dest(s,a^*,b^*) \cap X_{n+1} \neq \emptyset$, and 
then we can include $a^*$ to $U_*$ and $b^*$ to $\St_*$ (contradicting that 
they are the fixpoints). 
It follows that we can include $a^* \in U_{\ell+1}$ and continue.
\end{enumerate}
Hence we have either already proved that player~2 can ensure the 
predecessor operator or $U_\ell=Z$ in the end. 
If $U_\ell$ is $Z$ in the end, then $Z$ satisfies the property used in the 
previous cases of $U_*$ (the proof of part a), and then as in the 
previous proof (of part a), the uniform distribution over $B_k$ is a 
witness that player~2 can ensure $\pre_2(\no X_{n-i}) \sdcup 
\lpreodd_2(i,\no X_{n+1}, \no Y_{n}\ldots, \no X_{n-i+1}, \no Y_{n-i})$.

\end{enumerate}

\end{enumerate}

The desired result follows.
\qed
\end{proof}

\medskip\noindent{\bf Characterization of $\limit_1(\IP,M,\Phi)$ set.}
From Lemma~\ref{lemm:limit-infprec}, Lemma~\ref{lemm:limit-infprec1}, and
the duality of predecessor operators (Lemma~\ref{lemm:dual-limit})
we obtain the following result characterizing the limit-sure winning 
set for memoryless infinite-precision strategies for parity objectives.

\begin{theo}{}\label{theo-limit-fin}
For all concurrent game structures $\game$ over state space $S$, 
for all parity objectives $\Phi=\ParityCond$ for player~1, 
with $p: S \mapsto [1..2n]$, the following assertions hold.
\begin{enumerate}
\item 
We have $\limit_1(\IP,M,\Phi)= \limit_1(\IP,\FM,\Phi)$, 
and $\limit_1(\IP,\FM,\Phi)=W$, where 
$W$ is defined as the $\mu$-calculus formula in Fig~\ref{mu-formula},
and $B_i=p^{-1}(i)$ is the set of states with priority $i$, for 
$i \in [1..2n]$.
\begin{figure*}
\[
\begin{array}{l}
\nu Y_{n}. \mu X_{n}. \nu Y_{n-1}.  \mu X_{n-1}.  \cdots \nu Y_1. \mu X_1. \nu Y_0. \mu X_0. 
\\ 
\left[
\begin{array}{c}
B_{2n} \cap \pre_1(Y_{n-1}) \sdcup \lpre_1(Y_n,X_n) \\
\cup \\
B_{2n-1} \cap \lpreodd_1(0,Y_{n-1},X_{n-1}) \sdcup \lpre_1(Y_{n},X_{n}) \\
\cup \\ 
B_{2n-2} \cap \lpreeven_1(0,Y_{n-1},X_{n-1},Y_{n-2})  \sdcup \lpre_1(Y_{n},X_{n}) \\
\cup \\
B_{2n-3} \cap \lpreodd_1(1,Y_{n-1},X_{n-1},Y_{n-2},X_{n-2})  \sdcup \lpre_1(Y_{n},X_{n}) \\
\cup \\
\vdots \\
B_2 \cap \lpreeven_1(n-2,Y_{n-1},X_{n-1},\ldots,Y_1,X_1,Y_0)  \sdcup \lpre_1(Y_{n},X_{n}) \\
\cup \\
B_{1} \cap \lpreodd_1(n-1,Y_{n-1},X_{n-1}, \ldots,Y_0,X_0)  \sdcup \lpre_1(Y_{n},X_{n}) \\
\end{array}
\right]
\end{array}
\]
\caption{$\mu$-calculus formula for limit-sure winning with finite-memory 
infinite-precision strategies}\label{mu-formula}
\end{figure*}

\item 
The set  $\limit_1(\IP,\FM,\Phi)$ can be computed symbolically 
using the $\mu$-calculus expression of Fig~\ref{mu-formula} in time 
$\bigo(|S|^{2n+2} \cdot \sum_{s \in S} 2^{|\mov_1(s) \cup \mov_2(s)|})$.

\item For  $s \in S$ whether $s \in \limit_1(\IP,\FM,\Phi)$ 
can be decided in NP $\cap$ coNP.
\end{enumerate}
\end{theo}

The NP $\cap$ coNP bound follows directly from the $\mu$-calculus expressions:  
the players can guess the ranking function of the $\mu$-calculus formula and 
for each state the players guess the sequence of $(A_i,\St_i,\Wk_i)$ to witness that the 
predecessor operators are satisfied.
The witnesses are polynomial and can be verified in polynomial time.

\smallskip\noindent{\bf Construction of infinite-precision strategies.} 
Note that for infinite-precision strategies we are interested
in the limit-sure winning set, i.e., for every $\ve>0$, 
there is a strategy to win with probability $1-\ve$, but not
necessarily a strategy to win with probability~1. 
The proof of Theorem~\ref{theo-limit-fin} (Lemma~\ref{lemm:limit-infprec}) 
constructs for every $\ve>0$ a memoryless strategy that ensures winning 
with probability at least $1-\ve$.

\smallskip\noindent{\bf Equalities of finite and infinite-precision.}
We now establish the last set of equalities required to establish all the 
desired equalities and inequalities described in Section~\ref{sec-intro}.

\begin{theo}{} Given a concurrent game structure $G$ and a parity objective $\Phi$ 
we have $\limit_1(\IP,M,\Phi)=\limit_1(\FP,M,\Phi)=\limit_1(\FP,\FM,\Phi)=\limit_1(\FP,\IM,\Phi)$
\end{theo}
\begin{proof}
We need to show that following two inclusions: 
(1)~$\limit_1(\IP,M,\Phi) \subseteq \limit_1(\FP,M,\Phi)$ (note since trivially we have $\limit_1(\FP,M,\Phi) \subseteq \limit_1(\IP,M,\Phi)$,
it would follow that $\limit_1(\IP,M,\Phi) = \limit_1(\FP,M,\Phi)$);
and 
(2)~$\limit_1(\FP,\IM,\Phi) \subseteq \limit_1(\FP,M,\Phi)$ (note that since trivially we have 
$\limit_1(\FP,M,\Phi) \subseteq \limit_1(\FP,\FM,\Phi) \subseteq \limit_1(\FP,\IM,\Phi)$ it would follow
that $\limit_1(\FP,M,\Phi) = \limit_1(\FP,\FM,\Phi) = \limit_1(\FP,\IM,\Phi)$).
We establish the above inclusions below.

\begin{enumerate}
\item \emph{(First inclusion: $\limit_1(\IP,M,\Phi) \subseteq \limit_1(\FP,M,\Phi)$).}
Consider $\varepsilon>0$, and consider $j \in \Nats$ such that $\frac{1}{j} \leq \varepsilon$. 
Then the construction of a witness memoryless strategy for $\Phi$ for the set 
$\limit_1(\IP,M,\Phi)$ to ensure winning with probability at least $1-\varepsilon'$ 
(as established in Theorem~\ref{theo-limit-fin}), for $\varepsilon'=\frac{1}{j}$,
plays every action with probabilities multiple of $b$, where $b\leq j^{2^{\bigo(|S|\cdot |A|)}}$, 
where $S$ is the set of states and $A$ is the set of actions. 
It follows that for every $\varepsilon>0$, there is a memoryless finite-precision strategy to 
ensure that the objective $\Phi$ is satisfied with probability at least $1-\varepsilon$ from 
all states in $\limit_1(\IP,M,\Phi)$.
Note that the precision of the strategy depends on $\varepsilon>0$.
This establishes the first desired inclusion.

\item \emph{(Second inclusion: $\limit_1(\FP,\IM,\Phi) \subseteq \limit_1(\FP,M,\Phi)$).}
We now show that $\limit_1(\FP,\IM,\Phi) \subseteq \limit_1(\FP,M,\Phi)$.
From the previous item it follows that we have $U=\limit_1(\FP,M,\Phi)= \limit_1(\IP,M,\Phi)$.
We have the following fact (by definition): for every state $s \in S \setminus U$ there exists a constant $c>0$ such that 
for every memoryless strategy for player~1 there is counter strategy for player~2 to ensure
that $\Phi$ is not satisfied with probability at least $c$.
Assume towards contradtiction that there is a state $s \in S\setminus U$ such that $s \in \limit_1(\FP,\IM,\Phi)$.
Fix $\ve>0$ such that $\ve<c$, and since $s \in \limit_1(\FP,\IM,\Phi)$ there is some finite-precision 
(possible infinite-memory) strategy $\stra_1$ to ensure that $\Phi$ is satisfied with probability $1-\ve$,
and let the precision of the strategy be $b$. 
Then consider the turn-based game $\wt{G}$ constructed in Proposition~\ref{prop-uniform-fp} for $b$-precision. 
Then in $\wt{G}$ there is a strategy to ensure that $\Phi$ is satisfied with 
probability at least $1-\ve$. 
However since $\wt{G}$ is a turn-based stochastic game, and in turn-based 
stochastic parity games pure memoryless optimal strategies exist, there is a 
pure memoryless strategy in $\wt{G}$ that ensures winning with probability 
at least $1-\ve$, and from the pure memoryless strategy in $\wt{G}$ 
we obtain a memoryless strategy in $G$ that ensures winning with probability
at least $1-\ve> 1-c$. 
Thus we have a contradiction to the fact.
Thus the desired result follows.

\end{enumerate}
The desired result follows.
\qed
\end{proof}

\smallskip\noindent{\bf Independence from precise probabilities.}
Observe that the computation of all the predecessor operators
only depends on the supports of the transition function, and does not
depend on the precise transition probabilities. 
Hence the computation of the almost-sure and limit-sure 
winning sets is independent of the precise transition probabilities, 
and depends only on the supports.
We formalize this in the following result.

\begin{theo}{}\label{thrm:conc-noprob-equal}
Let $\game_1=(S,\moves,\mov_1,\mov_2,\trans_1)$ and 
$\game_2=(S,\moves,\mov_1,\mov_2,\trans_2)$ be two concurrent game structures  
that are equivalent, i.e., $\game_1 \equiv \game_2$. 
Then for all parity objectives $\Phi$, for all $C_1 \in \set{P,U,\FP,\IP}$ and 
$C_2 \in \set{M,\FM,\IM}$ we have 
(a)~$\almost_1^{\game_1}(C_1,C_2,\Phi)= \almost_1^{\game_2}(C_1,C_2,\Phi)$; and 
(b)~$\limit_1^{\game_1}(C_1,C_2,\Phi)= \limit_1^{\game_2}(C_1,C_2,\Phi)$. 
\end{theo}

All cases of the above theorem, other than when $C_1=\IP$ and $C_2=\IM$ follows 
from our results, and the result for $C_1=\IP$ and $C_2=\IM$ follows from the 
results of~\cite{dAH00}.

\section{Conclusion}
 In this work we studied the bounded rationality problem for qualitative
analysis in concurrent parity games, and presented a precise characterization.
The theory of bounded rationality for quantitative analysis is future 
work, and we believe the results of this paper will be helpful in 
developing the theory.

\paragraph{Acknowledgements.} We are indebted to and thank 
anonymous reviewers for extremely helpful comments.

\bibliographystyle{alpha}
\bibliography{MyArt,PRart}

\end{document}